\documentclass[sigconf,10pt]{acmart}
\usepackage[ruled, vlined, linesnumbered]{algorithm2e}
\usepackage{balance}
\def\BibTeX{{\rm B\kern-.05em{\sc i\kern-.025em b}\kern-.08emT\kern-.1667em\lower.7ex\hbox{E}\kern-.125emX}}

\allowdisplaybreaks
\copyrightyear{2020}
\acmYear{2020}
\setcopyright{acmcopyright}\acmConference[SIGMOD'20]{Proceedings of the 2020 ACM SIGMOD International Conference on Management of Data}{June 14--19, 2020}{Portland, OR, USA}
\acmBooktitle{Proceedings of the 2020 ACM SIGMOD International Conference on Management of Data (SIGMOD'20), June 14--19, 2020, Portland, OR, USA}
\acmPrice{15.00}
\acmDOI{10.1145/3318464.3389781}
\acmISBN{978-1-4503-6735-6/20/06}

\acmSubmissionID{mod0482s}

\newtheorem{definition}{Definition}
\newtheorem{lemma}{Lemma}
\newtheorem{theorem}{Theorem}
\def\header{\vspace{2mm}\noindent}

\newcommand{\pushright}[1]{\ifmeasuring@#1\else\omit\hfill$\displaystyle#1$\fi\ignorespaces}
\newcommand{\pushleft}[1]{\ifmeasuring@#1\else\omit$\displaystyle#1$\hfill\fi\ignorespaces}

\def\e{\varepsilon}

\def\S{\hat{S}}
\def\D{\hat{D}}

\def\l{\ell}
\def\E{\mathrm{E}}
\def\Var{\mathrm{Var}}

\def\scw{\sqrt{c}}

\def\exsim{ExactSim\xspace}

\def\inN{\mathcal{I}}
\def\outN{\mathcal{O}}
\def\din{d_{in}}

\begin{document}
	
	\fancyhead{}
	% do not delete this code.

	% The "title" command has an optional parameter, allowing the author to define a "short title" to be used in page headers.
	\title{Exact Single-Source SimRank Computation on Large Graphs}
	\subtitle{[Technical Report]}

	% The "author" command and its associated commands are used to define the authors and their affiliations.
	% Of note is the shared affiliation of the first two authors, and the "authornote" and "authornotemark" commands
	% used to denote shared contribution to the research.
	\author{Hanzhi Wang}
	\email{hanzhi_wang@ruc.edu.cn}
	\affiliation{%
		\institution{School of Information\\Renmin University of China}
	}
	
	\author{Zhewei Wei}
	\authornote{Zhewei Wei is the corresponding author.}
	\email{zhewei@ruc.edu.cn}
	\affiliation{%
		\institution{Gaoling School of Artificial Intelligence\\Renmin University of China}}
	
	\author{Ye Yuan}
	\email{yuan-ye@bit.edu.cn}
	\affiliation{%
		\institution{School of Computer Science and technology\\Beijing Institute of Technology}
	}
	
	\author{Xiaoyong Du}
	\email{duyong@ruc.edu.cn}
	\affiliation{%
		\institution{MOE Key Lab DEKE\\Renmin University of China}}
	
	\author{Ji-Rong Wen}
	\email{jrwen@ruc.edu.cn}
	\affiliation{%
		\institution{Beijing Key Lab of Big Data Management and Analysis Method\\Renmin University of China}}

	%
	% By default, the full list of authors will be used in the page headers. Often, this list is too long, and will overlap
	% other information printed in the page headers. This command allows the author to define a more concise list
	% of authors' names for this purpose.
	%\renewcommand{\shortauthors}{Trovato and Tobin, et al.}
	
	%
	% The abstract is a short summary of the work to be presented in the article.
	\begin{abstract}
%  {\em SimRank} is a classic metric that measures the similarities of nodes in a graph. Given a node $u$ in graph $G =(V, E)$, a {\em single-source SimRank query} returns the SimRank similarities $s(u, v)$ between node $u$ and each node $v \in V$. This type of queries is widely used in web search and social networks, such as link prediction, web mining and spam detections. Existing work on  single-source SimRank queries, however, suffer from three major deficiencies. First, all previous algorithms incur query cost at least linear to the graph size $n$, which limits the scalability of these algorithms. Second, previous methods do not take into account the structure of input graphs. Third, most existing work is unable to provide empirical study for the accuracy of the algorithms on large graphs.

  {\it SimRank} is a popular measurement for evaluating the node-to-node similarities based on the graph topology. 
  In recent years, single-source and top-$k$ SimRank queries have received increasing attention due to their applications in web mining, social network analysis, and spam detection.  
  However, a fundamental obstacle in studying SimRank has been the lack of ground truths. 
  The only exact algorithm, Power Method, is computationally infeasible on graphs with more than $10^6$ nodes. 
  %Consequently, no existing work has evaluated the actual accuracy of various single-source and top-$k$ SimRank algorithms on large real-world graphs. 
  Consequently, no existing work has evaluated the actual trade-offs between query time and accuracy on large real-world graphs. 

 In this paper, we present \exsim, the first algorithm that computes the exact single-source and top-$k$ SimRank results on large graphs. 
 With high probability, this algorithm produces ground truths with a rigorous theoretical guarantee. 
 We conduct extensive experiments on real-world datasets to demonstrate the efficiency of ExactSim. 
 The results show that ExactSim provides the ground truth for any single-source SimRank query with a precision up to 7 decimal places within a reasonable query time.

\end{abstract}

%%% Local Variables:
%%% mode: latex
%%% TeX-master: "paper"
%%% End:

	%
	% The code below is generated by the tool at http://dl.acm.org/ccs.cfm.
	% Please copy and paste the code instead of the example below.
	%
	


	%
	% Keywords. The author(s) should pick words that accurately describe the work being
	% presented. Separate the keywords with commas.
	
	%
	% A "teaser" image appears between the author and affiliation information and the body 
	% of the document, and typically spans the page. 
	
	%
	% This command processes the author and affiliation and title information and builds
	% the first part of the formatted document.
\maketitle

\section{Introduction} \label{sec:intro}
Computing link-based similarity is an overarching problem in graph
analysis and mining.
Amid the existing similarity
measures~\cite{page1999pagerank,xi2005simfusion,zhao2009p,zhang2015panther}, SimRank has emerged as a popular metric for assessing
structural similarities between nodes in a graph.
SimRank was introduced by Jeh and Widom~\cite{JW02} to formalize the intuition that
``two pages are similar if they are referenced by similar pages.''
Given  a directed graph $G = (V, E)$ with $n$ nodes $ \{v_1, \ldots,
v_n\}$ and $m$ edges, the  SimRank matrix $S$ defines the similarity between any two
nodes $v_i$ and $v_j$ as follows: 
\begin{equation} \label{eqn:intro-simrank}
S(i, j) =
\begin{cases}
1, & \text{for $i = j$;}\\
{\displaystyle \sum_{v_{i'} \in \inN(v_i)}{\sum_{v_{j'} \in \inN(v_j)}{c\cdot S(i', j') \over
      \din (v_i) \cdot \din(v_j) }} }, & \text{for $i \neq j$.}
\end{cases}
\end{equation}
Here,  $c$ is a decay factor typically set to 0.6 or 0.8~\cite{JW02,LVGT10}. $\inN(v_i)$
denotes the set of in-neighbors of $v_i$, and $\din(v_i)$
denotes the in-degree of $v_i$. SimRank aggregates
similarities of multi-hop neighbors of $v_i$ and $v_j$ to produce
high-quality similarity measure, and  has
been adopted in various applications such as recommendation systems \cite{li2013mapreduce},
link prediction \cite{lu2011link}, and graph embeddings \cite{tsitsulin2018verse}. 

A fundamental obstacle for studying SimRank is the lack
of ground truths on large
graphs. Currently, the only methods that compute the SimRank matrix is
Power Method and its variations~\cite{JW02,lizorkin2010accuracy}, which inherently
takes $O(n^2)$ space and at least $O(n^2)$ time as there are
$O(n^2)$ node pairs in the graphs. This complexity is infeasible on
large graphs ($n\ge 10^6$). Consequently, the majority of recent
works~\cite{KMK14,MKK14,TX16,FRCS05,LeeLY12,LiFL15,SLX15,YuM15b,jiang2017reads,liu2017probesim,wei2019prsim} focus on {\em single-source and top-$k$ queries}. Given a source node
$v_i$, a single-source query asks for the SimRank similarity between
every node and $v_i$, and a top-$k$ query asks for the $k$ nodes with the
highest SimRank similarities to $v_i$. % These two types of queries
% effectively model various applications in recommendation and link
% analysis, and thus has been the center of SimRank research in recent
% years.
Unfortunately, computing ground truths for  the single-source and
top-$k$ queries on large graphs still remains an open problem. 
% although the single-source and top-$k$ queries only require a space
% overhead of $O(n)$, it remains an open problem whether it is possible
% to compute the ground truths of these two
% queries on large graphs. 
To the best of our knowledge, Power
Method is still the only way to obtain exact single-source and top-$k$
results, which is not feasible on large graphs. Due to the hardness of exact computation, existing
works
on single-source and top-$k$ queries focus on approximate computations
with efficiency and accuracy guarantees.

The lack of ground truths has severely limited our understanding
towards 
SimRank and SimRank algorithms. First of all, designing approximate
algorithms without the ground truths is like shooting in the dark.  % we are unable to
% evaluate the actual accuracy of the existing
% approximate SimRank algorithms on large graphs.
Most existing works take the following approach: they
evaluate the accuracy on small graphs where the ground truths can be
obtained by the Power Method with $O(n^2)$ cost. 
Then they report the efficiency/scalability results on large graphs with consistent parameters. 
%This approach is flawed for two reasons. Firstly,
This approach is flawed for the reason that
consistent parameters may still lead to unfair comparisons. For
example, some of the existing
methods generate a fixed number of random walks
from each node, while others fix the maximum error $\e$ and generate
${\log n\over \e^2}$ random walks from each node. If we increase the graph size
$n$, the comparison becomes unfair as
the latter methods require more random walks from each node. 
Secondly, it is known that the structure of 
large real-world graphs can be very different from that of small
graphs. Consequently, the accuracy results on small graphs can only
serve as a rough guideline for accessing the actual error of the
algorithms in real-world applications.  We believe that the only right
way to evaluate the effectiveness of a SimRank algorithm is to
evaluate its results against the ground truths on large real-world graphs.

\begin{comment}
Second, the lack of ground truths has also prevented us from
exploiting the distribution of SimRank on real-world graphs.
For example, it is known~\cite{BahmaniCG10} that the PageRank of most real-world graphs follows the power-law
distribution. The natural question is that, does SimRank also follow the power-law
distribution on real-world graphs? 
Furthermore, the performance of some existing methods~\cite{Wang2019FBLPMC}
depends on the {\em density} of the
SimRank, which is defined as the percentage of node pairs with SimRank
similarities larger than some threshold $\e$. Analyzing the distribution or
density of SimRank is clearly infeasible without the
ground truths. 
\end{comment}

% 2) even with consistent
% parameters, the experiments can still be unfair, since some of the
% parameters should relate to the graph size.
%    \begin{itemize}
%     \item Problem with evaluating effectiveness on small graphs and
%       efficiency on large graphs. 1) Parameters may change; 2) Consist
%       parameters does not imply fair comparison.
      
%       \item Problem with pooling. 1) Evaluation results are only valid
%         for methods in the pool; 2) Unable to evaluate top-$k$ with large
%         $k$'s; 3) Unable to evaluate actual error.
%       \end{itemize}

%   \header{\bf SimRank v.s. Graph Structure.} 
% Secondly, the lack of ground truth limits our understanding of the
% SimRank. In particular, does SimRank on Power-Law graphs also follows
% the Power-Law distribution? Density estimation.

\header{\bf Exact Single-Source SimRank Computation.} In this paper,
we study the problem of computing the exact single-source SimRank
results on large graphs. A key insight is 
that exactness does not imply absolutely zero error. This is because SimRank values may be
infinite decimals, and we can only store these values with finite
precision. Moreover, we note that the ground truths computed by Power
Method also incur an error of at most $c^L$, where $L$ is the
number of iterations in Power Method. In most applications, $L$ is set to be large
enough such that  $c^L$ is smaller than the numerical error and thus
can be ignored. In this paper, we aim to develop an
algorithm that answers single-source SimRank queries with an additive 
error of at most  $\e_{min} = 10^{-7}$. Note that the float type in
various programming languages usually support precision of up to 6
or 7 decimal places, so by setting $\e_{min} = 10^{-7}$, we guarantee
the algorithm returns the same answer as the ground truths in the
float type. As we shall see,
this precision is extremely challenging for existing
methods. To make
the exact computation possible, we are also going to allow a small
probability to fail. We define the probabilistic exact
single-source SimRank algorithm as follows.

\begin{definition}%[Probabilistic Exact Single-source SimRank Algorithm]
	%\vspace{-2mm}
  With probability at least $1-1/n$, for {\em every} source node $v_i \in V$, a probabilistic exact
  single-source SimRank algorithm  answers the single-source SimRank
  query of $v_i$ with additive error of at most $\e_{min} = 10^{-7}$. 
  	%\vspace{-2.9mm}
\end{definition}

\header{\bf Our Contributions.} In this paper, we propose \exsim, the
first algorithm that enables probabilistic
exact single-source SimRank queries on large graphs.
We show that existing single-source methods share a common complexity term $O\left(
      {n \log n \over \e_{min}^2} \right)$, and thus are unable to
    achieve exactness on large graphs. 
However, \exsim runs in $O\left({\log n \over \e_{min}^2}+m\log {1\over \e_{min}}\right)$
 time, which is feasible for both large graph size $m$ and small error
 guarantee $\e_{min}$. We also apply several non-trivial optimization
 techniques to reduce the query cost and space overhead of \exsim. 
 In our empirical study, 
 we show that ExactSim is able to compute the ground truth with a precision of up to 7 decimal places within one hour on graphs with billions of edges. 
 Hence, we believe ExactSim is an effective tool for producing the ground truths for single-source SimRank queries on large graphs.
 %In our empirical study, %\exsim is able to compute the ground truth of a single-source query within one hour on a billion-edge graph.
 %we show that \exsim is able to compute 7 decimal places' precision results of a single-source query within one hour on a billion-edge graph. 
 %Hence, \exsim can be regarded as an effective probabilistic exact single-source SimRank algorithm. 

   %Based on the ground truths provided by \exsim, 
   %we also conduct the first empirical study on the accuracy/cost tradeoffs of existing approximate single-source algorithms on large real-world graphs. 
   %Finally, we use \exsim to exploit various properties of SimRank on large real-world graphs. 
   %we exploit various properties of SimRank on large real-world graphs. 
   %In particular, we show that the single-source SimRank values follow the power-law distribution. 
   %We also study the density of SimRank values on large graphs.
   
%    we also present theoretical and empirical
%    evidence on why existing methods are unable to achieve exactness
%    within reasonable time and space budget.
% We use \exsim to evaluate the tradeoffs between accuracy and
%      efficiency for various single-source and top-$k$
%      algorithms on large real-world and synthetic graphs. This is the first study of accuracy
%      on large graphs. 

%%% Local Variables:
%%% mode: latex
%%% TeX-master: "paper"
%%% End:

%\vspace{-5mm}
\section{Preliminaries and Related Work} \label{sec:prelim}
In this section, we review the state-of-the-art single-source SimRank
algorithms. Our \exsim algorithm is largely inspired by three prior
works:
Linearization~\cite{MKK14}, PRSim ~\cite{wei2019prsim} and
pooling~\cite{liu2017probesim},  and we will describe them in details. 
%In Section~\ref{sec:experiments}, we will also use the ground truths provided by \exsim to evaluate the algorithms mentioned in this section. 
Table~\ref{tbl:def-notation} summaries the notations used in this paper. 

\begin{table} [t]
	\centering
	\renewcommand{\arraystretch}{1.3}
	\begin{small}
		%\tblcapup
		\caption{Table of notations.}\label{tbl:def-notation}
		\vspace{-2mm}
		%\tblcapdown
		%p{2.3in}
		\begin{tabular} {|l|p{2.2in}|} \hline
			{\bf Notation}       &   {\bf Description}                                       \\ \hline
			$n, m$      &   the numbers of nodes and edges in $G$                            \\ \hline
			$\inN(v_i), \outN(v_i)$       &   the in/out-neighbor set of node $v_i$
			\\ \hline
			$S, S(i, j)$    &  the SimRank matrix and the SimRank similarity of $v_i$ and $v_j$     \\ \hline
			% $W(u)$    & a $\scw$-walk from a node $u$                                    \\ \hline
			$c$          &   the decay factor in the definition of SimRank                   \\ \hline
			$\e, \e_{min}$         &   additive error parameter and error required for
			exactness ($\e_{min} = 10^{-7}$)          \\ \hline
			% $\e_r$         &   the maximum relative error allowed in top-$k$ SimRank queries
			$P$, $D$   & the transition matrix and the diagonal correction matrix\\
			\hline
			$\vec{\pi}_i, \vec{\pi}_i^\ell,$   & the Personalized PageRank
			and $\ell$-hop
			Personalized PageRank
			vectors of node $v_i$\\
			\hline
			$ \vec{h}_i^\ell$   &  the $\ell$-hop Hitting Probability vector of $v_i$\\
			\hline
			
			%     $\rf(s,t)$, $\pif(s,t)$ & The reserve and residue of $t$ from $s$ in the forward search \\
			%     \hline
			%     $\frsum$ & The sum of all nodes' residues during in the forward
			%                search from $s$\\
			% \hline
			%   $h^{\l}(v_i, v_j)$ & the hitting probability (HP) from node $v_i$ to node $v_j$ at step $\l$ (see Section~\ref{sec:our-overview}) \\ \hline
		\end{tabular}
		%\vspace{-5mm}
	\end{small}
\end{table}
%\vspace{-5mm}

% Then we will give some intuitions on why existing methods are
% inherently unable to achieve exactness on large graphs. Our \exsim algorithm is largely inspired by some of these existing
% techniques, and we will use \exsim to evaluate the cost-accuracy tradeoffs of
% these algorithms on large graphs.

%\vspace{-1mm}\subsection{Monte Carlo Methods} \label{sec:iterative}

\header{\bf MC}~\cite{FR05} 
A popular interpretation of SimRank is the {\em meeting probability} of random walks.
In particular, we consider a random walk from node $u$ that, at each
step, moves to a random {\em in-neighbor} with probability $\scw$, and
stops at the current node with probability $1-\scw$. Such a random walk is called a {\em $\scw$-walk}. Suppose we
start a $\scw$-walk from node $v_i$ and a $\scw$-walk from node $v_j$,
we call the two $\scw$-walks {\em meet} if they visit the same node at the
same step. It is known~\cite{TX16}  that  %$S(i,j) =\Pr[\textrm{two } \scw\textrm{-walks from } v_i \textrm{ and } v_j \textrm{ meet}].$
\begin{equation}
\label{eqn:scwwalk}
S(i,j) =\Pr[\textrm{two } \scw\textrm{-walks from } v_i \textrm{ and } v_j \textrm{ meet}].
%\vspace{-2mm}
\end{equation}
MC makes use of this equation %~\eqref{eqn:scwwalk} 
to derive a Monte-Carlo algorithm for computing single-source SimRank. 
In the preprocessing phase, we simulate $R$  $\scw$-walks from each node in
$V$. Given a
source node $v_i$, we compare the $\scw$-walks from $v_i$ and from each node
$v_j \in V$, and use the fraction of $\scw$-walks that meet as an
estimator for $S(i,j)$. By standard concentration inequalities, the
maximum error is bounded by $\e$ with high probability if we set
$R=O\left({\log n \over\e^2}\right)$, leading to a preprocessing time of
$O\left({n\log n \over\e^2}\right)$.

\header{\bf Linearization and ParSim.} Given a graph $G=(V,E)$, let
$P$ denote the (reverse) {\em transition
matrix}, that is, $P(i, j) = 1/\din(v_j)$ for $v_i \in  \inN(v_j)$ and
$P(i, j) = 0$ otherwise. Let $S$  denote the SimRank matrix with $S(i, j) = s(v_i,
v_j)$. It is shown in two independent works, Linearization~\cite{MKK14}
and ParSim~\cite{yu2015efficient}, that  $S$ can be expressed as the 
following linear summation: %$S = \sum_{\l = 0}^{+\infty} c^\l \left(P^\l\right)^\top D P^\l$, 
\begin{equation} \label{eqn:def-lsim2}
S = \sum_{\l = 0}^{+\infty} c^\l \left(P^\l\right)^\top D P^\l, 
\end{equation}
where $D$ is the {\em diagonal  correction matrix} with each diagonal element $D(k,k)$
taking value from $1-c$ to $1$. Consequently, a single-source query
for node $v_i$ can be computed by %$S\cdot \vec{e_i}= \sum_{\l = 0}^{+\infty} c^\l \left(P^\l \right)^\top D P^\l \cdot \vec{e_i},$
%\vspace{-2mm}
\begin{equation} \label{eqn:def-lsim-single-source}
%\vspace{-2mm}
S\cdot \vec{e_i}= \sum_{\l = 0}^{+\infty} c^\l \left(P^\l \right)^\top D P^\l \cdot \vec{e_i},
\end{equation}
where $\vec{e_i}$ denotes the one-hot vector with the $i$-th
 element being $1$ and all other elements being $0$. Assuming the
 diagonal matrix $D$ is correctly given, the single-source query for node
 $v_i$ can be computed by %$S_L\cdot \vec{e_i}= \sum_{\l = 0}^{L} c^\l \left(P^\l \right)^\top D P^\l \cdot \vec{e_i},$
  \begin{equation} \label{eqn:def-lsim-single-source2}
 S_L\cdot \vec{e_i}= \sum_{\l = 0}^{L} c^\l \left(P^\l \right)^\top D P^\l \cdot \vec{e_i},
 \end{equation}
where $L$ is the number of iterations. After $L$ iterations, the
additive error reduces to $c^L$, so setting $L=O\left(\log {1\over \e}\right)$ is
sufficient to guarantee a maximum error of $\e$. At the $\l$-th iterations,
the algorithms performs $2\ell+1$ matrix-vector multiplications to
calculate $c^\l \left(P^\l \right)^\top D P^\l \cdot \vec{e_i}$, and
each matrix-vector multiplication takes $O(m)$ time. 
Consequently, the total query time is bounded
by $O\left(\sum_{\ell=1}^{L}m\ell \right) = O(mL^2) = O\left(m \log^2 {1\over \e}\right)$.
\cite{MKK14} and \cite{yu2015efficient} also show that if we first
compute and store the transition probability vectors $\vec{u}_i=
P^\l \cdot \vec{e_i}$ for $\ell = 0,\ldots, L$, then we can use the
following equation to compute %$S_L\cdot \vec{e_i}$ in $O(mL) = O\left(m \log {1\over \e}\right)$ time: $S_L\cdot \vec{e_i}= D\cdot  \vec{u}_0 + c P (D\cdot  \vec{u}_1 + \cdots +  c P (D\cdot  \vec{u}_{T-1} +  c P D\cdot \vec{u}_{T} ) \cdots),$
%\vspace{-1mm}
\begin{equation}
  \label{eqn:def-lsim-single-source3}
   S_L\cdot \vec{e_i}= D\cdot  \vec{u}_0 + c P^\top (D\cdot  \vec{u}_1 + \cdots +  c P^\top  (D\cdot  \vec{u}_{T-1} +  c P^\top \cdot D\cdot \vec{u}_{T} ) \cdots),
   %\vspace{-2mm}
\end{equation}
However, this optimization requires a memory size of $O(nL) \\= O\left(n\log {1\over \e} \right)$,
which is usually several times larger than the graph
size $m$. Therefore,
\cite{MKK14} only uses the $O\left(m \log^2 {1\over \e}\right)$ algorithm in the
experiments. 

Besides the large space overhead, another problem with
Linearization and ParSim is that the
diagonal correction matrix $D$ is hard to
compute. Linearization~\cite{MKK14} formulates $D$ as the solution to
a linear system, and propose a Monte Carlo solution that
takes $O\left({n\log n \over\e^2}\right)$ to derive an
$\e$-approximation of $D$. On the other hand, ParSim directly sets $D=(1-c)I$, where $I$ is the
identity matrix. This approximation basically ignores the first
meeting constraint and has been adopted in many other SimRank
works \cite{FNSO13,He10,Yu13,Li10,Yu14,YuM15b,KMK14}. It is shown that
the similarities calculated by this approximation are different from
the actual SimRank \cite{KMK14}. However,  the quality of this approximation is still a myth due to the lack of
ground truths on large graphs.

\header{\bf PRSim}~\cite{wei2019prsim} %extends from ProbeSim  and SLING by introducing a partial indexing and a new Probe algorithm. 
introduces a partial indexing and a probe algorithm. 
Let $\vec{\pi}^\ell_i= (1-\scw) \vec{h}_i^\ell= (1-\scw) \left(\scw P\right)^\ell \cdot \vec{e}_i$
denote the {\em $\ell$-hop Personalize PageRank vector} of $v_i$. 
In particular, $\vec{\pi}_i^\ell(k)$ is the probability that a $\scw$-walk from node $v_i$ {\em stops} at  node $v_k$ in exactly $\ell$ steps.  
PRSim suggests that equation~\eqref{eqn:def-lsim-single-source} can be 
re-written as %PRSim uses the following variant of~\eqref{eqn:sling1} to estimate %$S(i,j)$: 
%$S(i,j) ={1\over (1-\scw)^2}\sum_{\ell=0}^{\infty}\sum_{k=1}^n \vec{\pi}_i^\ell(k)\cdot  \vec{\pi}_j^\ell(k) \cdot D(k,k).$
% \begin{equation}
%   \label{eqn:prsim}
% s(v_i,v_j)={1\over (1-\scw)^2}\sum_{\ell=1}^{\infty}\sum_{k=1}^n
% \pi_\ell(v_i,v_k)\pi_\ell(v_j, v_k) D(k,k).
% \end{equation}
%\vspace{-2mm}
\begin{equation}
  \label{eqn:prsim}
S(i,j) ={1\over (1-\scw)^2}\sum_{\ell=0}^{\infty}\sum_{k=1}^n \vec{\pi}_i^\ell(k)\cdot  \vec{\pi}_j^\ell(k) \cdot D(k,k).
%\vspace{-2mm}
\end{equation}
%Similar to SLING, 
%PRSim precomputes $\vec{\pi}_i^\ell(k)$ using the same local push algorithm as SLING. 
PRSim precomputes $\vec{\pi}_j^\ell(k)$ with additive error $\e$ for each $\ell$ and $v_j, v_k \in V$, using a {\em local push} algorithm~\cite{AndersenCL06}.
%However, to avoid the  overwhelming $O(n/\e)$ space overhead, PRSim only
To avoid overwhelming space overhead, PRSim only precomputes $\vec{\pi}_j^\ell(k)$ for a small subset of $v_k$. 
%Furthermore, PRSim avoids the $O\left({n\log n \over\e^2}\right)$ time computation of $D$ by estimating the
Furthermore, PRSim computes $D$ by estimating the product $\vec{\pi}_i^\ell(k) \cdot D(k,k)$ together with an $O\left({\log n \over\e^2}\right)$ time Monte-Carlo algorithm. 
Finally, PRSim proposes a new Probe algorithm that samples each node $v_j$ according to $\vec{\pi}_j^\ell(k)$. %using $O(n \vec{\pi}_i (k))$ time, 
%where $\vec{\pi}_i = \sum_{\ell = 0} ^\infty \vec{\pi}^\ell_i$ denote the {\em Personalized PageRank  vector} of  $v_i$ 
%and $\vec{\pi}_i (k)$ denotes the probability that an $\scw$-walk from $v_i$ stops at $v_k$.
%The average query time of the query cost of PRSim is bounded by $O\left({n \|\vec{\pi}_i\|^2\over \e^2}\log n\right)$.  
The average query time of PRSim is bounded by $O\left({n \cdot \sum_{k=1}^n \vec{\pi}(k)^2 \over \e^2}\log n\right)$, where $\vec{\pi}(k)$ denotes the PageRank of $v_k$.  
It is well-known that on scale-free networks,  the PageRank vector $\vec{\pi}$ follows the power-law distribution, 
and thus $\|\vec{\pi}\|^2 = \sum_{k=1}^n \vec{\pi}(k)^2$ is a value much smaller than $1$. 
However, for  worst-case graphs or even some "bad" source nodes on scale-free networks, 
the running time of PRSim remains $O\left({n\log n \over\e^2}\right)$. 
\begin{comment}
  \header{\bf TopSim} ~\cite{LeeLY12} is an index-free algorithm based
  on local exploitation. Given source node $v_i$,
  TopSim firstly finds all nodes $v_k$ reachable from $v_i$ within
  $\ell= 1, \ldots, L$
  steps. For each such $v_k$ on the $\ell$-th level, TopSim
  deterministically compute 
 $\vec{h}_j^\ell(k)$, the probability that each $v_j$ that can reach $v_k$ in exactly
 $\ell$ steps. \cite{LeeLY12} also proposes various optimization to
 reduce the query cost. Due to the dense structures of real-world
 networks, TopSim is only able to exploit a few levels on large
 graphs, which leads to a low precision.

\end{comment}

  % computes the probability of each $v_i$
  % all $u$ reachable from $v_0$ in exact $l$ steps following
  % \emph{out-edges}. Moreover, $u$ must not be any node in path $v_0
  % \rightsquigarrow v$. Then the probability of $v_0 \rightsquigarrow
  % v, v_0 \rightsquigarrow u$ is aggregated into $score(u,v)$, an
  % estimation of $s(u,v)$. The algorithm stops when the gap between
  % $k$-th and $(k+1)$-th largest score exceeds the upper bound of score
  % any node can gain within $l$ steps, in fact,
  % ${(\frac{c}{D})}^{l+1}$, where $c$ is the decay factor and $D$ is
  % the average degree of the graph. They also propose several
  % optimization algorithms to improve the speed, some with trade of
  % accuracy.

%\vspace{-1mm}
\subsection{Other Related Work} \label{sec:related}
Besides the state-of-the-art methods that we discuss above, there are
several other techniques for SimRank computation, which we review in
the following. {\em Power method} \cite{JW02} is the classic algorithm
that computes all-pair SimRank similarities for a given graph. Let $S$
be the SimRank matrix such that $S_{ij} = s(i, j)$, and $P$ be the
transition matrix of $G$. Power method recursively computes the
SimRank Matrix $S$ using the formula \cite{KMK14} $S = (c P^\top S P) \vee I,$
where $\vee$ is the element-wise maximum operator. 
Several follow-up works \cite{LVGT10,YZL12,YuJulie15gauging} improve the efficiency or effectiveness of the power method in terms of either efficiency or accuracy. However, these methods still incur $O(n^2)$ space overheads, as there are $O(n^2)$ pairs of nodes in the graph. For single-source queries,
READS \cite{jiang2017reads} and TSF ~\cite{SLX15} are MC-based algorithms supporting dynamic graphs. 
Both of them incurs of $O\left({n\log n \over\e^2}\right)$ query time for $\e$ additive error.
SLING~\cite{TX16}  is an index-based SimRank
algorithm that support fast single-source and top-$k$ queries on
static graphs. Its preprocessing phase using $O\left({n\log n \over\e^2}\right)$ time which is infeasible for large graphs. 
ProbeSim~\cite{liu2017probesim} and TopSim~\cite{LeeLY12} are both index-free solutions based on local exploitation. 
Their query time is also bounded by $O\left({n\log n \over\e^2}\right)$. 
Besides, Li et al.\ \cite{LiFL15} propose a distributed version of the Monte
Carlo approach in \cite{FRCS05}, but it achieves scalability at the
cost of significant computation resources. 
Finally, there is existing work on {\em SimRank similarity join} \cite{TaoYL14,MKK15,ZhengZF0Z13}, 
variants of SimRank \cite{AMC08,FR05,Lin12,YuM15a,ZhaoHS09}
and graph applications \cite{bhuiyan2018representing, ye2018using}, but
the proposed solutions are inapplicable for top-$k$ and single-source
SimRank queries.

%\begin{comment}
\header{\bf Pooling.} Finally, pooling~\cite{liu2017probesim} is an experimental method for evaluating the
accuracy of top-$k$ SimRank algorithms without the ground
truths. Suppose the goal is to compare the accuracy of top-$k$ queries
for $\ell$ algorithms  $A_1, \ldots, A_\ell$. Given a query node
$v_i$, we retrieve the top-$k$ nodes returned by each algorithm,
remove the duplicates, and merge them into a pool. Note that there are
at most $\ell k$ nodes in the pool. Then we estimate $S(i,j)$ for
each node $v_j$ in the pool using the Monte Carlo algorithm. We set the
number of random walks to be $O\left( {\log n \over
    \e_{min}^2}\right)$  so that we can obtain the ground truth of
$S(i, j)$ with high probability. After that, we take the $k$ nodes with the highest
SimRank similarity to $v_i$ from the pool as the ground
truth of the top-$k$ query, and use this ``ground truth'' to evaluate
the precision of each of the $\ell$ algorithms. Note that the set of these $k$
nodes is not the actual ground truth. However, it represent the best
possible $k$ nodes that can be found by the $\ell$ algorithms that
participate in the pool and thus can be used to compare the quality of
these algorithms. 

Although pooling is proved to be effective in our scenario where ground
truths are hard to obtain, it has some drawbacks. First of all,
the precision results obtained by  pooling are {\em relative} and thus
cannot be used outside the pool. 
This is because the top-$k$ nodes from the pool are not the actual ground truths. 
Consequently, an algorithm that achieves $100\%$ precision in
the pool may have a precision of $0\%$ when compared to the actual top-$k$
result. Secondly, the complexity of pooling $\ell$ algorithms is
$O\left( {\ell k \log n \over
    \e_{min}^2}\right)$, so pooling is only feasible for evaluating top-$k$ queries
with small $k$. In particular, we cannot use pooling to evaluate the
single-source queries on large graphs.
%\end{comment}

%\vspace{-9mm}
\subsection{Limitations of Existing Methods}
  We now analyze the reasons why existing methods are unable to achieve
  exactness (a.k.a an error of at most $\e_{min}=10^{-7}$). First of all, 
  %ParSim and TSF ignore the first meeting constraint and thus 
  ParSim ignores the first meeting constraint and thus incurs large errors. 
  For other methods that enforce the first meeting
  constraint, they all incur a complexity term of
  $O\left( {n\log n \over \e^2}\right)$, either in the preprocessing
  phase or in the query phase. In particular, SLING and Linearization
  simulate $O\left( {n\log n \over \e^2}\right)$ random walks to estimate the
  diagonal correction matrix $D$. For ProbeSim, MC,  READS and PRSim, this complexity is causing by
simulating random walks in the query phase or  the preprocessing
phase.  The $O\left( {n\log n \over \e^2}\right)$ complexity is
infeasible for exact SimRank computation on large graphs, since it
combines two expensive terms $n$ and ${1\over \e_{min}^2}$. As an
example,  we consider the IT
dataset used in our experiment, with
$4*10^7$ nodes and  over1 billion edges. In order to achieve a maximum error of
$\e_{min}=10^{-7}$, we need to simulate $ {n\log n \over \e^2} \approx
10^{23}$ random walks. This may take years, even with parallelization on a cluster of
thousands of machines.
% \header{\bf ParSim} \cite{yu2015efficient}  also uses
% formula~\eqref{eqn:def-lsim-single-source2} to answer approximate
% single-source queries. The major difference between ParSim and
% Linearization is that ParSim 

% Let $W'(u)$ and $W'(v)$ be two $\scw$-walks from two nodes $u$ and $v$, respectively. We say that two $\scw$-walks {\em meet}, if there exists a positive integer $i$ such that the $i$-th nodes of $W'(u)$ and $W'(v)$ are the same. Then, according to \cite{TX16},
% \begin{equation} \label{eqn:prelim-rw2}
% s(u, v) = \Pr\left[\textrm{$W'(u)$ and $W'(v)$ meet}\right].
% \end{equation}

%%% Local Variables:
%%% mode: latex
%%% TeX-master: "paper"
%%% End:

%\input{reviews.tex}
%\vspace{-1mm}
\section{The \exsim Algorithm} \label{sec:exsim}
In this section, we present \exsim, a probabilistic algorithm that
computes the exact single-source SimRank  results within reasonable
running time. We first present a  basic
version of \exsim, and then introduce some more advanced techniques to optimize
the query and the space cost.

%\vspace{-1mm}
\subsection{Basic \exsim Algorithm}
Our \exsim algorithm is largely inspired by three prio works: pooling~\cite{liu2017probesim},
Linearization~\cite{MKK14} and PRSim ~\cite{wei2019prsim}. We now
discuss how \exsim extends from these
existing methods in details. These discussions will also reveal the
high level ideas of the \exsim algorithm. 
\begin{enumerate}
%\begin{comment}
\item Despite its limitations, pooling~\cite{liu2017probesim}
  provides  a key insight for achieving exactness: while an  $O\left(
      {n \log n \over \e^2}\right)$ algorithm is not feasible for exact SimRank
    computation on large graphs, we can actually afford an $O\left(
      {\log n \over \e^2}\right)$ algorithm. The $ {1\over
      \e^2}$ term is still expensive for $\e=\e_{min} = 10^{-7}$,
    however, the
    new complexity reduces the dependence on the graph size
    $n$ to logarithmic, and thus achieves very high scalability.

  % suggests that although the
  %   $O\left( {n\log n \over \e^2}\right)$ complexity is not feasible
  %   for exactness, we can actually afford an algorithm with $O\left(
  %     {\log n \over \e^2}\right)$. This complexity is an over-kill for
  %   small graphs, it grows very slowly
  %   with the graph size $n$, which means it is scalable. Furthermore,
  %   the random walk can be trivially.
%\end{comment}

    \item Linearization~\cite{MKK14} and ParSim~\cite{yu2015efficient} show that if the diagonal correction
      matrix $D$ is correctly given, then we can compute the exact
      single-source SimRank results  in $O\left(m \log_{1\over c} {1 \over \e_{min}}\right)$ time and $O\left(n
        \log_{1\over c} {1 \over \e_{min}} \right)$ extra space.
      % By setting $L =\log_{1\over c} {1 \over \e_{min}}$,
    % the two algorithms guarantee a maximum additive error of
    % $\e_{min}$.  
    For typical setting of $c$ ($0.6$ to $0.8$), the
    number of iterations  $\log_{1\over c} {1 \over
      \e_{min}} = \log {10^7} \le 73$ is a constant, so this complexity
    is essentially the same as that of performing BFS multiple times on the graphs. The scalability of the algorithm
      is  confirmed in the experiments of
      \cite{yu2015efficient}, where $D$ is set to be $(1-c)I$.
Moreover,  the exact algorithms~\cite{page1999pagerank} for Personalized
PageRank and PageRank also incurs a running time of  $O\left(m \log {1 \over
    \e_{min}}\right)$, and has been widely used for computing ground
truths on large graphs. 

      \item  While the $O\left(
      {n\log n \over \e^2}\right)$ complexity seems unavoidable as we need to
    estimate each entry in the diagonal correction matrix $D$ with
    additive error $\e$,  PRSim~\cite{wei2019prsim}
    shows that  it only takes $O\left( {\log n \over \e^2}\right)$
    time to estimate the product $\vec{\pi}_i^\ell(k) \cdot
    D(k,k)$ with additive error $\e$ for  each $k=1, \ldots, n$ and
    $\ell = 0, \ldots, \infty$, where $\vec{\pi}_i^\ell$ is the
    $\ell$-hop Personalized PageRank vector of $v_i$. This result provides two crucial
    observations:  1) It is possible to
    answer an single-source query without an $\e$-approximation of  each 
$D(k,k)$; 2) The accuracy of each $D(k,k)$ should
depend on $\vec{\pi}_i(k)$,  the Personalized PageRank of  $v_k$ with
respect to the source node $v_i$.
  \end{enumerate} 

We combine the ideas of PRSim and
Linearization/ParSim to derive the basic \exsim algorithm. Given an
error parameter $\e$, \exsim fixes the total
number of $\scw$-walk samples to be $R =
O\left( {\log n \over \e^2 }\right)$, and distribute a fraction of $R 
\vec{\pi}_i(k)$ samples (note that $\sum_{k=1}^n \vec{\pi}_i(k) =1$)  to estimate $D(k,k)$. 
%Then 
It performs
Linearization/ParSim with the estimated $D$ to obtain the single-source result. The
algorithm  runs  in $O\left( {\log n \over \e^2} + m\log {1\over \e}\right)$ time and uses
    $O\left(n \log {1\over \e}\right)$ extra space.
    Since both complexity terms $O\left(\log n \over \e^2 \right)$ and
    $O\left(m\log {1\over \e}\right)$ are feasible for $\e_{min} =
    10^{-7}$ and large graph size $m$, we have a
    working algorithm for exact single-source SimRank queries on large
    graphs. 

      \begin{algorithm}[t]
\begin{small}
\caption{Basic \exsim Algorithm} \label{alg:exsim1}
\BlankLine
	\KwIn{Graph $G$ with transition matrix $P$, source node $v_i$,  maximum error  $\e$}
        \KwOut{Estimated single-source SimRank vector $S \cdot \vec{e_i}$}
        $L = \left \lceil \log_{1\over c} {2 \over \e}  \right\rceil$\;
        $\vec{\pi}_i^{0}, \vec{\pi}_i =  (1-\scw)\vec{e}_i$\;
        \For{$\ell$ from $1$ to $L$}
        {
          $\vec{\pi}_i^{\ell} = \scw P \cdot \vec{\pi}_i^{\ell-1} $\;
          $\vec{\pi}_i= \vec{\pi}_i + \vec{\pi}^{\ell}$\;
        }
        $R = {6\log n \over (1-\scw)^4  \e^2}$\;
        \For{$k$ from $1$ to $n$}
        {
         Invoke Algorithm~\ref{alg:Dk}  with $R(k)= \lceil R \cdot
          \vec{\pi}_i(k)\rceil$ to obtain an estimator $\D(k,k)$ for $D(k,k)$\;
        }
        $\vec{s}^{0} = {1\over 1-\scw}\D\cdot \vec{\pi}_i^{L} $\;
         \For{$\ell$ from $1$ to $L$}
        {
          $\vec{s}^{\ell} = \scw P^\top \cdot \vec{s}^{\ell-1} +  {1\over 1-\scw}\D\cdot  \vec{\pi}_i^{L-\ell} $ \;
          Clear {$\vec{s}^{\ell-1} $}\;
        }
        \Return {$\vec{s}^L$}\;
\end{small}
\end{algorithm}

  % \begin{algorithm}[t]
% %\renewcommand{\arraystretch}{1.3}
% \begin{small}
% \caption{Basic \exsim Algorithm} \label{alg:exsim1}
% \BlankLine
% 	\KwIn{Transition matrix $P$, source node $v_i$, correction diagonal matrix
%           $D$, maximum error  $\e$}
%         \KwOut{Single-source SimRank vector $S \vec{e_i}$}
%         $L = \log 1/\e$\;
%         $\vec{u}^{(0)} =  \vec{e}_i$, $\vec{\pi}_i =  (1-c)\vec{e}_i$\;
%         \For{$j$ from $1$ to $L$}
%         {
%           $\vec{u}^{(j)} = P \cdot \vec{u}^{(j-1)} $\;
%           $\vec{\pi}_i= \vec{\pi}_i + (1-\scw)\scw^j\vec{u}^{(j)}$\;
%         }
%         $R = {\log n \over \e_{min}^2}$\;
%         \For{$k$ from $1$ to $n$}
%         {
%           Estimate $D(k,k)$ using $R(k)= \lceil R
%           \vec{\pi}_i(k)\rceil$ samples\;
%           % \For{$x$ from $1$ to $R(k)$}
%           % {
%           %   Sample two independent $\scw$-walks from $v_k$\;
%           %   \If{The two $\scw$-walks do not meet}{
%           %     $D(k,k) = D(k,k) + 1/R(k)$\;
%           %   }
%           % }
            
%         }
%         $\vec{s}^{(0)} = \vec{u}^{(L)} $\;
%          \For{$j$ from $1$ to $L$}
%         {
%           $\vec{s}^{(j)} = cP^{\top} D\cdot \vec{s}^{(j-1)} + D\cdot  \vec{u}^{(L-j)} $ \;
%           Clear {$\vec{s}^{(j-1)} $}\;
%         }
%         \Return {$\vec{s}_L$}\;
% \end{small}
% \end{algorithm}

\begin{algorithm}[t]

\begin{small}
\caption{Basic method for estimating $D(k,k)$} \label{alg:Dk}
\BlankLine
	\KwIn{Graph $G$, node $v_k$, number of samples $R(k)$ }
        \KwOut{$\D(k,k)$ as an estimation for $D(k,k)$}
        $\D(k,k) = 0$\;
  \For{$x$ from $1$ to $R(k)$}
          {
            Sample two independent $\scw$-walks from $v_k$\;
            \If{The two $\scw$-walks do not meet}{
              $\D(k,k) = \D(k,k) + 1/R(k)$\;
            }
          }
\Return{$\D(k,k)$}\; 
\end{small}
\end{algorithm}

% \header{\bf SimRank Formula.} 
% Consequently, we can use $S_L\cdot \vec{e}_i = {1\over 1-\scw}  \sum_{\l =
%   0}^{L}  \left(\scw P^\l \right)^\top D \cdot \vec{\pi}_i^{\ell}$ as
% an accurate approximation of  $S\cdot \vec{e}_i $. 

Algorithm~\ref{alg:exsim1} illustrates the pseudocode of the basic
\exsim algorithm.  Note that to cope with Personalized PageRank, we use the fact
that $\vec{\pi}_i^\ell = \left(1-\sqrt{c}\right) \cdot \left(\scw P \right)^\ell \cdot \vec{e}_i$
and re-write
equation~\eqref{eqn:def-lsim-single-source}  as %$  S\cdot \vec{e}_i = {1\over 1-\scw}  \sum_{\l =0}^{\infty}  \left(\scw P^\top \right)^\l D \cdot \vec{\pi}_i^{\ell}.$
 \begin{equation}
   \label{eqn:def-lsim-single-source3}
   S\cdot \vec{e}_i = {1\over 1-\scw}  \sum_{\l =0}^{\infty}  \left(\scw P^\top \right)^\l D \cdot \vec{\pi}_i^{\ell}.
 \end{equation}
Given a source node $v_i$ and a maximum error $\e$,
we first set the number of iterations $L$ to be $L = \left \lceil
  \log_{1\over c} {2 \over \e}  \right\rceil$ (line 1). We then iteratively compute the $\ell$-hop Personalized
PageRank vector $\vec{\pi}_i^\ell = \left(\scw P \right)^\ell \cdot \vec{e}_i$ for $\ell = 0, \ldots, L$, as well
as the Personalized PageRank vector $\vec{\pi}_i =
\sum_{\ell=0}^L \vec{\pi}_i^\ell $ (lines 2-5). To obtain an estimator $\D$ for the
diagonal correction matrix $D$, we set the total number of samples to
be $R = {6\log n \over  (1-\scw)^4\e^2}$ (line 6). For each $D(k,k)$, we set
$R(k)= \lceil R \vec{\pi}_i(k)\rceil$ and invoke
Algorithm~\ref{alg:Dk} to estimate $D(k,k)$ (lines
7-8). Algorithm~\ref{alg:Dk} essentially
simulates $R(k)$ pairs of $\scw$-walks from node
$v_k$ and uses the fraction of pairs that do not meet as an estimator $\D(k,k)$
for $D(k,k)$. 
Finally, we use equation~\eqref{eqn:def-lsim-single-source3} to % compute the
% single-source result $\vec{s}^L= {1\over 1-\scw}  \sum_{\l =
%   0}^{L}  \left(\scw P^\l \right)^\top \D \cdot \vec{\pi}_i^{\ell}$. In particular, we
iteratively compute 
$ \vec{s}^0 ={1\over 1- \scw}\D \cdot \vec{\pi}_i^L$, 
$\vec{s}^1 =  \scw P^\top \cdot \vec{s}^{0}+{1\over 1- \scw}\D \cdot   \vec{\pi}_i^{L-1}
%{1\over 1- \scw}  \left( \D \cdot \vec{s}^{0} + \scw P^\top \D \cdot   \vec{\pi}_i^{L-1} \right) 
= {1\over 1- \scw}  \left( \scw P^\top \cdot \D  \cdot  \vec{\pi}_i^{L}+\D \cdot  \vec{\pi}_i^{L-1} \right)$ (lines 9-12), ..., and
%= {1\over 1- \scw}  \left( \D \cdot  \vec{\pi}_i^{L} +\scw P^\top \D  \cdot  \vec{\pi}_i^{L-1} \right)$ (lines 9-12), ..., and
\begin{align}
\vec{s}^{L}
&= \frac{ \left( \scw P^\top \left( \cdots  (\scw P^\top \cdot \D \cdot \vec{\pi}^{L} + \D \cdot  \vec{\pi}^{L-1} )+ \cdots \right)+\D \cdot \vec{\pi}^{0}  \right) }{1-\scw} \nonumber \\
&=  {1\over 1- \scw} \sum_{\l =0}^{L}  \left(\scw P^\top \right)^\l \D \cdot \vec{\pi}_i^{\ell}. \label{eqn:sL}
  %\vec{s}^{L} = {1\over 1- \scw} \sum_{\l =0}^{L}  \left(\scw P ^\top \right)^\l \D \cdot \vec{\pi}_i^{\ell}. \label{eqn:sL}
  %\vspace{-3mm}
  \end{align}
We return $\vec{s}^{L}$ as the single-source query result (line 13).

% transition probability
% vector $\vec{u}^{(j)} = P^j \cdot \vec{e}_i $ for each $j=0,\ldots,
% L$, and the Personalized PageRank vector $\vec{\pi}_i =
% (1-\scw)\sum_{j=0}^L \left(\scw \right)^j P^j \cdot \vec{e}_i $. Then
% we estimate each $D(k,k)$ using $R = {\log n \over \e_{min}^2}$
% samples. In particular, we assign $R(k) = \lceil R
% \vec{\pi}_i(k)\rceil$ samples to estimate $D(k,k)$. Since
% $\sum_{k=1}^n \vec{\pi}_i(k) = 1$, the total number of samples is
% $R$. 

\header{\bf Analysis.} To derive the running time and space overhead of the basic \exsim
algorithm, note that computing and storing each $\ell$-hop
Personalized PageRank vector $\vec{\pi}_i^{\ell}$ takes $O(m)$
time and $O(n)$ space. This results a running time of $O(mL)$ and a
space overhead of $O(nL)$. To estimate the diagonal correction matrix
$D$, the algorithm simulates $R$ pairs of $\scw$-walks, each of which
takes ${1\over \scw} = O(1)$ time. Therefore, the running time for
estimating $D$ can be bounded by $O(R)$. Finally, computing each $\vec{s}^\ell$ also takes $O(m)$
time, resulting an additional running time of $O(mL)$. Summing up all
costs, and we have the
total running time is bounded by $O(mL + R) = O\left(
{\log n \over \e^2} + m\log {1\over \e}\right)$, and the space
overhead is bounded by $O(nL) = O\left(n\log {1\over \e}\right)$. 

%Theorem~\ref{thm:basic_time} offers the conclusion of time and overhead space cost in algorithm~\ref{alg:exsim1}. 
%Detailed explanations and proofs are given in the appendix part.
%\vspace{-1mm}
%\begin{theorem}
%\label{thm:basic_time}
%The running time of algorithm~\ref{alg:exsim1} is bounded by $O\left({\log n \over \e^2} + m\log {1\over \e}\right)$, 
%while the space overhead is bounded by $O\left(n\log {1\over \e}\right)$. 
%\end{theorem}
%\vspace{-1mm}

We now analyze the error of the basic
\exsim algorithm. Recall that \exsim returns $\vec{s}^{L}(j)$ as the
estimator for $S(i,j)$, the SimRank similarity between the source node
$v_i$ and any other node
$v_j$.  We have the following Theorem.
%\vspace{-2mm}
\begin{theorem}%\footnote{All proofs can be found in the appendix}
	\label{thm:basic}
	With probability at least $1-1/n$, for any source node $v_i \in V$, the
	basic \exsim provide an single-source SimRank vector
	$\vec{s}^{L}$ such that, for  any node
	$v_j \in V$, we have $\left|\vec{s}^{L}(j) - S(i,j) \right| \le \e$.
	%\vspace{-1mm}
\end{theorem}
%\vspace{-1mm}
% Note that by applying union bound over all $n$ possible source node
% $v_i$, Theorem~\ref{thm:basic} guarantees that \exsim answer all
% single-source queries with probability at least $1-1/n$.
Theorem~\ref{thm:basic} essentially states that with high probability, the basic \exsim algorithm can compute any
single-source SimRank query with additive $\e$. 
The proof of Theorem~\ref{thm:basic} is fairly technical shown in appendix, however, the basic idea
%The proof's basic idea 
is to show that the variance of the estimator $\vec{s}^{L}(j)$ can be bounded by $O({1\over R}) = O(\e^2)$.
In particular, we have the following Lemma.

\begin{lemma}
	\label{lem:variance}
	The variance of $\vec{s}^{L}(j)$ is bounded by 
	\begin{equation}
	\label{eqn:variance-all}
	\Var[\vec{s}^{L}(j)] \le {1\over (1-\scw)^4R}\sum_{k=1}^n
	{  \vec{\pi}_i(k)^2 \vec{\pi}_j(k)^2\over \rho(k)} \cdot
	{D(k,k)}.
	\end{equation}
	In particular, by setting $\rho(k) = R(k)/R= \lceil R
	\vec{\pi}_i(k)\rceil /R$ in the basic \exsim algorithm, we
	have
	\begin{equation}
	\label{eqn:variance-basic}
	\Var[\vec{s}^{L}(j)] \le   {1\over (1-\scw)^4R}.
	\end{equation}
	%\vspace{-3mm}
\end{lemma}
Note that we only need inequality~\eqref{eqn:variance-basic} to derive
the error bound for the basic \exsim algorithm. The more complex
inequality~\eqref{eqn:variance-all} will be used to design various
optimization techniques.

\subsection{Optimizations}
Although the basic \exsim algorithm is a working algorithm for exact
single-source SimRank computation on large graphs, it suffers from some drawbacks. First of
all,  the $O(n\log {1\over \e})$  space overhead can be several
times larger than the actual graph size $m$. Secondly, we still need to simulate $R=O\left({\log n \over \e^2}\right)$ of 
pairs of $\scw$-walks, which is a significant cost for $\e_{min} = 10^{-7}$. Although parallelization can help, we are still
interested in developing algorithmic techniques that reduces the
number of random walks. In this section, we provide three
optimization techniques that address these drawbacks.

\header{\bf Sparse Linearization.} We design a sparse version of Linearization that
significantly reduces the $O\left(n \log {1 \over \e}\right)$ space overhead while
retaining the $O(\e)$ error guarantee . Recall
that this space overhead is causing by storing the $\ell$-hop
Personalized PageRank vectors  $\vec{\pi}_i^{\ell}$ for $\ell
= 0, \ldots, L$. We propose to make the following simple modification:
Instead of storing the dense vector $\vec{\pi}_i^{\ell}$, we
sparsify  the vector by removing all entries of  
with $\vec{\pi}_i^{\ell}(k) \le (1-\scw)^2\e$. To understand the
effectiveness of this approach, recall that a nice property of the
$\ell$-hop Personalized PageRank vectors is that all
possible entries sum up to $ \sum_{\ell=0}^\infty
\sum_{k=1}^n \vec{\pi}_i^{\ell}(k)= \sum_{k=1}^n
\vec{\pi}^{\ell}(k)   = 1$. By the Pigeonhole principle, the number of
$\vec{\pi}_i^{\ell}(k)$'s that are larger than $(1-\scw)^2\e$ is bounded by  ${1
  \over (1-\scw)^2 \e}$. Thus the
space overhead is reduced to $O\left({1
  \over \e}\right)$. This overhead is acceptable for exact
computations where we set $\e = \e_{min} = 10^{-7}$, as it does not
scale with the graph size.

The following Lemma proves that the sparse Linearization will only
introduce an extra additive error of $\e$. If we scale down $\e$ by a
factor of $2$, the total error guarantee  and the asymptotic running
time of \exsim will remain the same, and the space overhead is reduced
to $O\left( {1\over \e} \right)$.
%\vspace{-1mm}
\begin{lemma}
  \label{lem:sparse}
  The sparse Linearization introduces an extra additive error of
  $\e$ and reduces the space overhead to  $O\left({1
  \over \e}\right)$.
%\vspace{-2mm}
\end{lemma}

\header{\bf Sampling according to $\vec{\pi}_i(k)^2$.} Recall that in
the basic \exsim algorithm, we simulate $R$ pairs of $\scw$-walks in
total, and distribute $ \vec{\pi}_i(k)$ fraction
of the $R$ samples to estimate $D(k,k)$.  A natural question
is that, is there a better scheme to distribute these $R$ samples?
It turns out if we distribute the
samples according to $\vec{\pi}_i(k)^2$, we can further reduce the variance of
the estimator and hence achieve a better running time. More precisely,
we will set   $R(k) =
  R \left \lceil {\vec{\pi}_i(k)^2 \over  \|\vec{\pi}_i\|^2 }
  \right\rceil$, where   $\|\vec{\pi}_i\|^2 = \sum_{k=1}^n  \vec{\pi}_i(k)^2$ is the squared 
  norm of the Personalized PageRank vector $\vec{\pi}_i$. 
  The following Lemma, whose proof can be found in appendix, shows that by sampling according to
  $\vec{\pi}_i(k)^2$, we can reduce the number of sample $R$ by a
  factor of $\|\vec{\pi}_i\|^2$. 
  
  \begin{lemma}
\label{lem:variance1}
 By sampling according to $\vec{\pi}_i(k)^2$, the number of random
 samples required is reduced to $O\left({\|\vec{\pi}_i\|^2 \log n \over
     \e^2}\right)$. 
\end{lemma}

To demonstrate the effectiveness of sampling according to $\vec{\pi}_i(k)^2$, notice that in the worst case, $\|\vec{\pi}_i\|^2$ is as large as $\|\vec{\pi}_i\|_1^2
= 1$, so this optimization technique is essentially useless. However, it is
known~\cite{BahmaniCG10} that on scale-free networks, the Personalized PageRank vector
$\vec{\pi}_i$ follows a power-law distribution: let
$\vec{\pi}_i(k_j)$ denote the $j$-th largest entry of
$\vec{\pi}_i$, we can assume $\vec{\pi}_i(k_j) \sim {j^{-\beta}
	\over n^{1-\beta}}$ for some power-law exponent $\beta \in
(0,1)$. In this case, $\|\vec{\pi}_i\|^2$ can be bounded by $O\left(
\sum_{j=1}^n \left( {j^{-\beta}
	\over n^{1-\beta}} \right)^2 \right) = O\left(\max\left\{{lnn\over
	n}, {1\over n^{2-2\beta}} \right\}\right)$, and the
$\|\vec{\pi}_i\|^2$ factor becomes significant for any power-law
exponent $\beta <1$.

% Our second
%   optimization technique is to change the way we distribute the number
%   of random walks for estimating $D(k,k)$
%   Lemma~\ref{lem:conc}   suggests that to reduce the number of random
%   samples for estimating the diagonal correction matrix $D$, a
%   feasible choice is to reduce the variance of the estimator
%   $\s(v_i,v_j)$. We take two independent approaches: designing better
%   schemes to assign the random samples, and reduce the uncertainty for
%   estimating $D$.

  % The first approach is that, instead of assigning $R(k) =
  % R\lceil \vec{\pi}_i(k) \rceil $ samples to estimate $D(,k,k)$, we assign $R(k) =
  % R \left \lceil {\vec{\pi}_i(k)^2 \over  \|\vec{\pi}_i\|^2 } \right\rceil$ samples to $D(k,k)$. Here

  % \header{\bf Analysis on scale-free networks.} 

% take the minimum when ${R \vec{\pi}_i(k)^2 \over
%  \|\vec{\pi}_i\|^2} = 1$, and consequently  $\rho(k)  \ge {
%  \vec{\pi}_i(k)^2 \over  \|\vec{\pi}_i\|^2 }$.

% In fact, since the minimum of $R(k)$ is
%     $1$, so  $\rho(k) $ can get as small as $ 1/R$, which means $b$
%     can be as large as $R$. This will lead to unbounded number of
%     samples.  Since we set  $\rho(k) =
%  \left \lceil { R \vec{\pi}_i(k)^2 \over  \|\vec{\pi}_i\|^2 }
%  \right\rceil /R$, 

\begin{algorithm}[t]
	\begin{small}
		\caption{Improved method for estimating $D(k,k)$} \label{alg:newD}
		\BlankLine
		\KwIn{Graph $G$, node $v_k$,
			sample number $R(k)$ }
		\KwOut{An estimator $\D(k,k)$ for $D(k,k)$}
		\If{$\din(v_k) = 0$}{
			\Return{$\D(k,k)=1$}\;
		}
		\ElseIf{$\din(v_k) = 1$}{
			\Return{$\D(k,k)=1-c$}\;
		}
		$P^{\ell}(x,k)=0$ for $\ell \ge 0, x\in V$\;
		$P^{0}(k,k)=1$\;
		$E_k=0$\;
		\For {$\ell$ from $0$ to $\infty$}
		{
			\For{each $v_{q}$ with non-zero $\left(P^\top\right)^{\ell}(k,q)$}
			{
				Calculate $Z_\ell(k,q)$ using equation~\eqref{eqn:zl}\;
			}
			\For{$\ell'$ from $0$ to $\ell$}
			{
				\For{each $v_{q'}$ with non-zero $\left(P^\top\right)^{\ell-\ell'}(k,q')$}
				{
					\For{each $v_x$ with non-zero $\left(P^\top\right)^{\ell'}(q',x)$}
					{
						\For{each $v_q \in \inN(v_x)$}
						{
							$\left(P^\top\right)^{\ell'+1}(q',q) +=
							{\left(P^\top\right)^{\ell'}(q',x) \over \din(v_x)}$\;
							$E_k += 1$\;
							\If{ $E_k \ge {2R(k) \over \scw}$}
							{
								$\ell(k) = \ell$ and goto OUTLOOP\;
							}
						}
					}
				}
			}
			$\ell = \ell+1$\;
		}
	    OUTLOOP\;
		$\D(k,k) = 1- \sum_{\ell=1}^{\ell(k)}\sum_{q=1}^n
		Z_{\ell}(k, q)$\;
		\For{$z$ from $1$ to $R(k)$}
		{
			Sample two independent non-stop random walks from $v_k$\;
			\If{Two random walks reaches nodes $v_x$ and $v_y$
				at the $\ell(k)$ steps without meeting}
			{
				Sample a $\scw$-walks from $v_x$ and $v_y$\;
				\If{the two $\scw$-walks meet}{
					$\D(k,k) = \D(k,k) - c^{\ell(k)}/R(k)$\;
				}
			}                 
		}
		\Return{$\D(k,k)$}\; 
	\end{small}
\end{algorithm}
%\vspace{-3mm}

  \header{\bf Local deterministic exploitation for $D$.}

 The inequality~\eqref{eqn:variance-all}  in
Lemma~\ref{lem:variance} also suggests that we can reduce the variance of the
estimator $\vec{s}^L(j)$ by refining the Bernoulli estimator $\D(k,k)$.
Recall that we  sample $R(k) = \lceil R\vec{\pi}_i(k) \rceil$ or $R(k) =
R \left \lceil {\vec{\pi}_i(k)^2 \over  \|\vec{\pi}_i\|^2 }
\right\rceil$ pairs
of $\scw$-walks to estimate $D(k,k)$. If $\vec{\pi}_i(k)$ is large, we will
simulate a large number of  $\scw$-walks from $v_k$ to estimate
$D(k,k)$. In that case, the first few steps of these random walks will
most likely visit the same local structures around $v_k$, so it
makes sense to exploit these local structures deterministically, and
use the random walks to approximate the global structures. More
precisely, let $Z_\ell(k)$ denote  the probability that two
$\scw$-walks from $v_k$ first meet at the $\ell$-th step. Since these
events are mutually exclusive for different $\ell$'s, we have 
\begin{align*}
\label{eqn:lmp}
D(k,k)&= 1- \Pr[\textrm{two } \scw\textrm{-walks from } v_k \textrm{
	meet}]= 1- \sum_{\ell=1}^\infty Z_\ell(k).
\end{align*}
The idea is to deterministically compute $\sum_{\ell=1}^{\ell(k)}Z_\ell(k)$ for some tolerable step
$\ell(k)$, and using random walks to estimate the other part
$\sum_{\ell=\ell(k)+1}^{\infty}Z_\ell(k)$. It is easy to see that by
deterministically computing $Z_\ell(k)$ for the first $\ell(k)$
levels, we reduce the variance $\Var(D(k,k))$ by at least
$c^{\ell(k)}$.

A simple algorithm to compute $Z_\ell(k)$ is to list all possible
paths of length $\ell$ from $v_k$ and aggregate all meeting
probabilities of any two paths. However, the number of paths increases
rapidly with the length $\ell$, which makes this algorithm inefficient
on large graphs. Instead, we will derive the close form for $Z_\ell(k)$
in terms of the transition probailities.  In
particular,  let $Z_\ell(k,q)$ denote the  probability that two
$\scw$-walks first meet at node $v_q$ at their $\ell$-th steps. We have
$Z_\ell(k) = \sum_{q=1}^n Z_\ell(k,q)$, and hence %$D(k,k) = 1-\sum_{\ell=1}^\infty \sum_{q=1}^n Z_\ell(k,q)$. 
\begin{align}
D(k,k) = 1-\sum_{\ell=1}^\infty \sum_{q=1}^n Z_\ell(k,q)
\end{align}
Recall that $P^{\ell}$ (the $\ell$-th power
of the (reverse) transition matrix $P$) is
the $\ell$-step (reverse) transition matrix.  We have the
following Lemma that relates $Z_\ell(k,q)$ with the transition
probabilities. 

\begin{lemma}
	\label{lem:zl}
	$ Z_\ell(k,q) $ satisfies the following recursive form:
	\begin{equation} \label{eqn:zl}
	\begin{aligned}
		Z_\ell(k,q)= 
		&c^\ell \left(P^\top\right)^{\ell}(k,q)^2 \\
		&- \sum_{\ell'=1}^{\ell-1}
	\sum_{q'=1}^{n} c^{\ell'} \left(P^\top \right)^{\ell'}(q',q)^2 Z_{\ell-\ell'}(k, q').
	\end{aligned}
	\end{equation}
\end{lemma}

Given a node $v_k$ and a pre-determined target level $\ell(k)$, Lemma~\ref{lem:zl} also
suggests a simple algorithm to compute
$Z_\ell(k,q)$ for all $\ell \le \ell(k)$. We start by
performing BFS from node $v_k$  for up to $\ell(k)$ levels to calculate the
transition probabilities $\left(P^\top\right)^\ell(k,q)$ for $\ell =
0,\ldots, \ell(k)$ and $v_q \in V$. For each node $q'$ visited
at the $\ell'$-th level, we start a BFS from $q'$ for $\ell(k) -
\ell'$ levels to calculate  $\left(P^\top\right)^{\ell(k)-\ell'}(q',q)$ for $\ell = 1,\ldots, \ell(k)$ and $v_q \in V$. Then we use
equation~\eqref{eqn:zl} to calculate $Z_\ell(k,q) $ for $\ell
=0,\ldots, \ell(k)$ and $q\in V$. Note that this approach exploits strictly
less edges than listing all possible paths of length $\ell(k)$, as
some of the paths are combined in the computation of the transition
probabilities. 

However, a major problem with the above method is that the target
level $\ell(k)$ has to be predetermined, which makes the running time
unpredictable. An improper value of $\ell(k)$ could lead to the
explosion of the running time.  Instead, we will use an adaptive algorithm
to compute $Z_\ell(k)$.

Algorithm~\ref{alg:newD} illustrates the new
method for estimating $D(k,k)$. Given a node $v_k$ and a sample
number $R(k)$, the goal is to give an estimator for $D(k,k)$. For the
two trivial case
$\din(k) = 0$ and $\din(k)=1$, we return $D(k,k) = 1$ and $1-c$
accordingly (lines 1-4). 
For other cases, we iteratively compute all possible transition probabilities $\left(P^\top\right)^{\ell'+1}(q',q)$ for all $v_{q'}$ that is reachable
from $k$ with $\ell-\ell'$ steps (lines 5-10). Note that these $v_{q'}$'s
are the nodes with $\left(P^\top\right)^{\ell-\ell'}(k,q')>0$. To ensure
the deterministic exploitation stops in
time, we use a counter $E_k$ to record the total number of edges traversed
so far (line 11). If $E_k$ exceeds ${2R(k) \over \scw}$, the expected number of steps for simulating 
$R(k)$ pairs of $\scw$-walks, we terminate the deterministic exploitation and set
$\ell(k)$ as the current target  level for $v_k$ (lines 12-13). After we fix
$\ell(k)$ and compute $\sum_{\ell = 1}^{\ell(k)} Z_\ell(k)$ (lines 14-17), we
will use random walk sampling to estimate
$\sum_{\ell=\ell(k)+1}^{\infty}Z_\ell(k)$ (lines 18-23). In particular, we start two
special random walks from $v_k$. The random walks
do not stop in its first $\ell(k)$ steps; after the $\ell(k)$-th
step, each random walk stops with probability $\scw$ at each step. It is easy to see
that the probability of the two special random walks meet after
$\ell(k)$ steps is ${1\over
c^{\ell(k)}}\sum_{\ell=\ell(k)+1}^{\infty}Z_\ell(k)$. Consequently,
we can use the fraction of the random walks that meet multiplied by
$c^{\ell(k)}$ as an unbiased estimator for
$\sum_{\ell=\ell(k)+1}^{\infty}Z_\ell(k)$.

%\vspace{+2mm}
 \header{\bf Parallelization.} The \exsim algorithm is highly
 parallelizable as it only uses two primitive operations:
 matrix-(sparse) vector
 multiplication and random walk simulation. Both operations are embarrassingly
 parallelizable on GPUs or multi-core CPUs. The only exception is the
 local deterministic exploitation for $D(k,k)$. To parallelize this operation, we
 can apply Algorithm~\ref{alg:newD} to multiple $v_k$
 simultaneously. Furthermore, we can balance the load of each thread by
 applying Algorithm~\ref{alg:newD} to nodes $v_k$'s with similar
 number of samples  $R(k)$ in each epoch.

\section{Experiments} \label{sec:experiments} 
In this section,  we experimentally study \exsim and the other
single-source algorithms. We first evaluate \exsim against other
methods to prove \exsim's ability  of exact computation (i.e., $\e_{min} = 10^{-7}$) both on small and large graphs. 
%Then we use the ground truths computed by \exsim to perform an empirical study on existing SimRank algorithms and SimRank distributions. 
%Then we use the ground truths computed by \exsim to perform an empirical study on SimRank distributions. 
Then we conduct an ablation study to demonstrate the effectiveness of the optimization techniques.

% test \exsim and other algorithms' performance on small graphs to verify the exactness which could be only achieved by \exsim. Second, we use the high-precision results calculated by \exsim on large graphs to analyze other algorithms' accuracy/cost tradeoffs.  Finally, we draw SimRank distribution plots on large graphs to check whether they obey power law distribution. Additionally, we make a ablation study to show the effectiveness of \exsim's optimization techniques . 

% we present two kinds of experiments. Section 5.1 shows the detailed performance of \exsim.  Section 5.2 evaluates existing algorithms  and discusses the trade-off between accuracy and complexity for each method,  based on the high-precision results computed by \exsim.   \\
\header{\bf Datasets and Environment. } We use four small 
datasets and four large datasets, obtained from \cite{snapnets,LWA}. The
details of these datasets can be found in Table~\ref{tbl:datasets}.
All experiments are conducted on a machine with an Intel(R) Xeon(R)
E7-4809 @2.10GHz CPU and 196GB memory. 

\begin{table}[t]
  %\vspace{-2mm}
	\centering
	%\tblcapup
	\caption{Data Sets.}
	\vspace{-2mm}
	%\tblcapdown
	\begin{small}
		\begin{tabular}{|l|l|r|r|} %p{1.3in}|}
			\hline
			{\bf Data Set} & {\bf Type} & {\bf $\boldsymbol{n}$} & {\bf $\boldsymbol{m}$}	 \\ \hline
			ca-GrQc (GQ) & undirected & 5,242 & 28,968\\
			%AS-2000(AS) & undirected & 6,474 & 25,144\\
			CA-HepTh(HT) & undirected & 9,877 & 51,946\\
			Wikivote (WV) & directed & 7,115 & 103,689\\
			CA-HepPh (HP) & undirected & 12008 & 236978\\
			% Wiki-Vote(WV)	& 	directed &	7,155	&	103,689 \\
			% {HepTh(HT)}	    & 	undirected &	9,877	&	25,998		\\
			% {AS-Caida(AC)}	    &	directed &	26,475	&	106,762		\\
			% {HepPh(HP)}	        &	directed &	34,546	&	421,578 \\
			% Cnr-2000 (CN) & directed & 325,557 &	3,216,152 \\
			% {Web-Google(WG)} & directed & 875,713 &	5,105,039 \\)
			%As-Skitter (AS) & undirected & 1,696,415	& 11,095,298 \\
			%      {\color{red} In-2004(IN) } & directed & 1,382,908	& 16,917,053
			% \\
			DBLP-Author (DB) & undirected & 5,425,963 & 17,298,032 \\
			% {Com-LiveJournal(CL)} & undirected & 3,997,962	& 34,681,189 \\
			%LiveJournal (LJ) &	directed & 4,847,571 & 68,475,391\\
			IndoChina (IC)	& 	directed &	7,414,768  &	191,606,827		\\
			%Orkut-Links (OL) & undirected & 3,072,441 & 234,369,798 \\
			
			%DBpediaLink (DL) & directed & 18,268,992 & 172,183,984 \\
			%WikiLink(WL) & directed & 12,150,976 & 378,142,420 \\
			%Twitter(TW) & directed & 41,652,230 & 1,468,365,182 \\
			% { Web-Base}   & { directed} & {
			% 118,142,155} & { 1,019,903,190} \\
			%Web-Base (WB) & directed & 118,142,155& 1,019,903,190 \\
			It-2004 (IT)	&	directed & 41,290,682 & 1,135,718,909\\
			
			Twitter (TW) & directed & 41,652,230& 1,468,364,884 \\
			%SK-2005 (SK) & directed & 50,636,154	& 1,949,412,601 \\
			%Friendster (FD)  & directed & 68,349,466 & 2,586,147,869 \\
			%UK-Union (UK) & directed & 133,633,040 & 5,507,679,822 \\
			\hline
		\end{tabular}
	\end{small}
	\label{tbl:datasets}
	%\tbldown
	%\vspace{-3mm}
\end{table}

\begin{comment}
(1) methods and parameters: \exsim, Monte-Carlo methods: MC, READS,
TSF, iterative methods: Linearization, ParSim, Local push/sampling
methods: ProbeSim, PRSim, SLING, TopSim. Static v.s. dynamic. \\
(2) Ground truth calculation: power method on small graphs, \exsim on
large graphs \\
(3) Metrics: \\
- accuracy (Maximum erro, precision) \\
- time ( preprocessing time; query time )\\
- size ( index size) \\
\end{comment}

\begin{figure*}[!t]
	\begin{small}
		\centering
		%\vspace{-5mm}
		%    \begin{footnotesize}
		\begin{tabular}{cccc}
			%\multicolumn{4}{c}{\hspace{-4mm} \includegraphics[height=5mm]{./Figs/legend_large.eps}} \vspace{-1mm} \\
			\hspace{-6mm} \includegraphics[height=35mm]{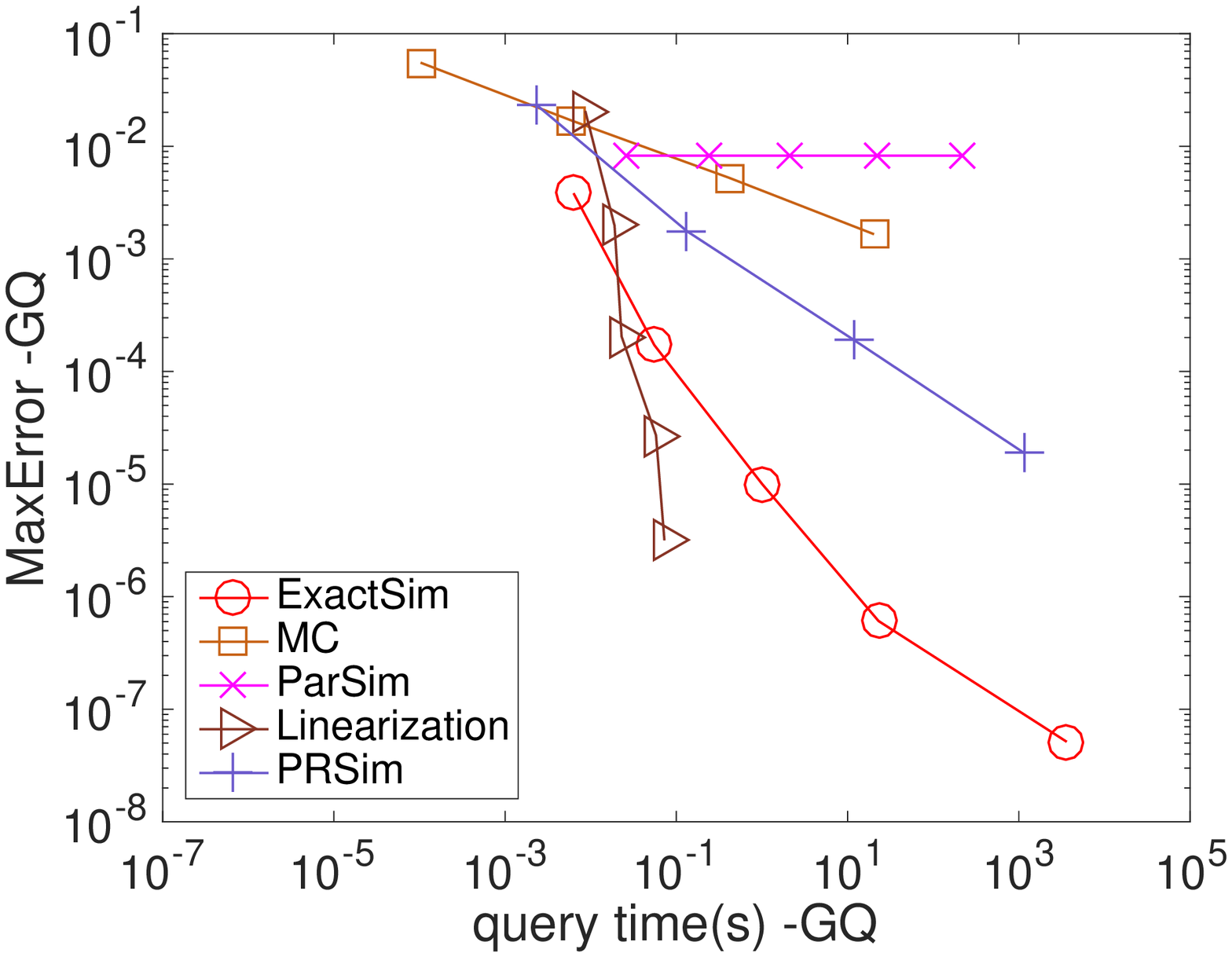} &
			\hspace{-3mm} \includegraphics[height=35mm]{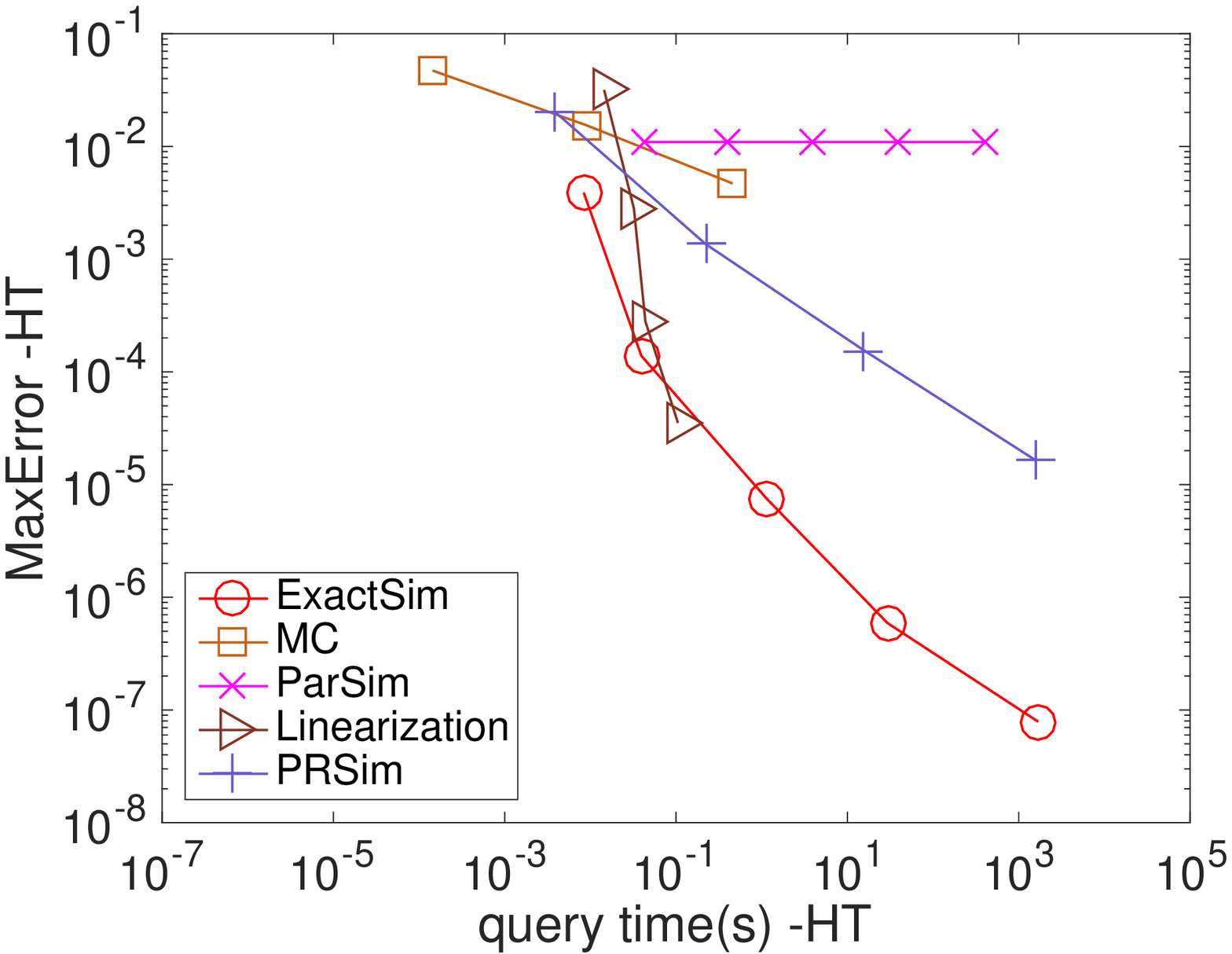} &
			\hspace{-3mm} \includegraphics[height=35mm]{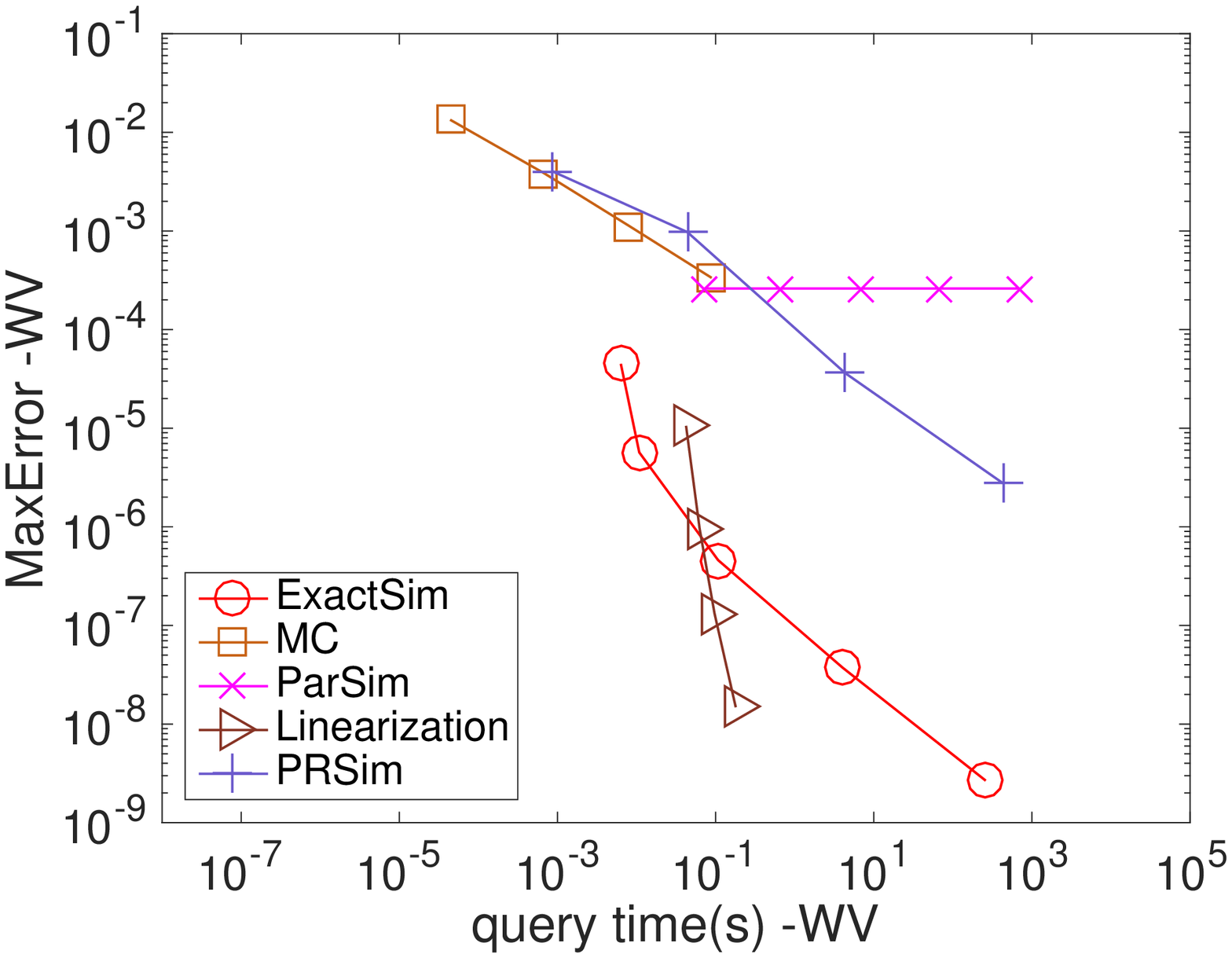} &
			\hspace{-3mm} \includegraphics[height=35mm]{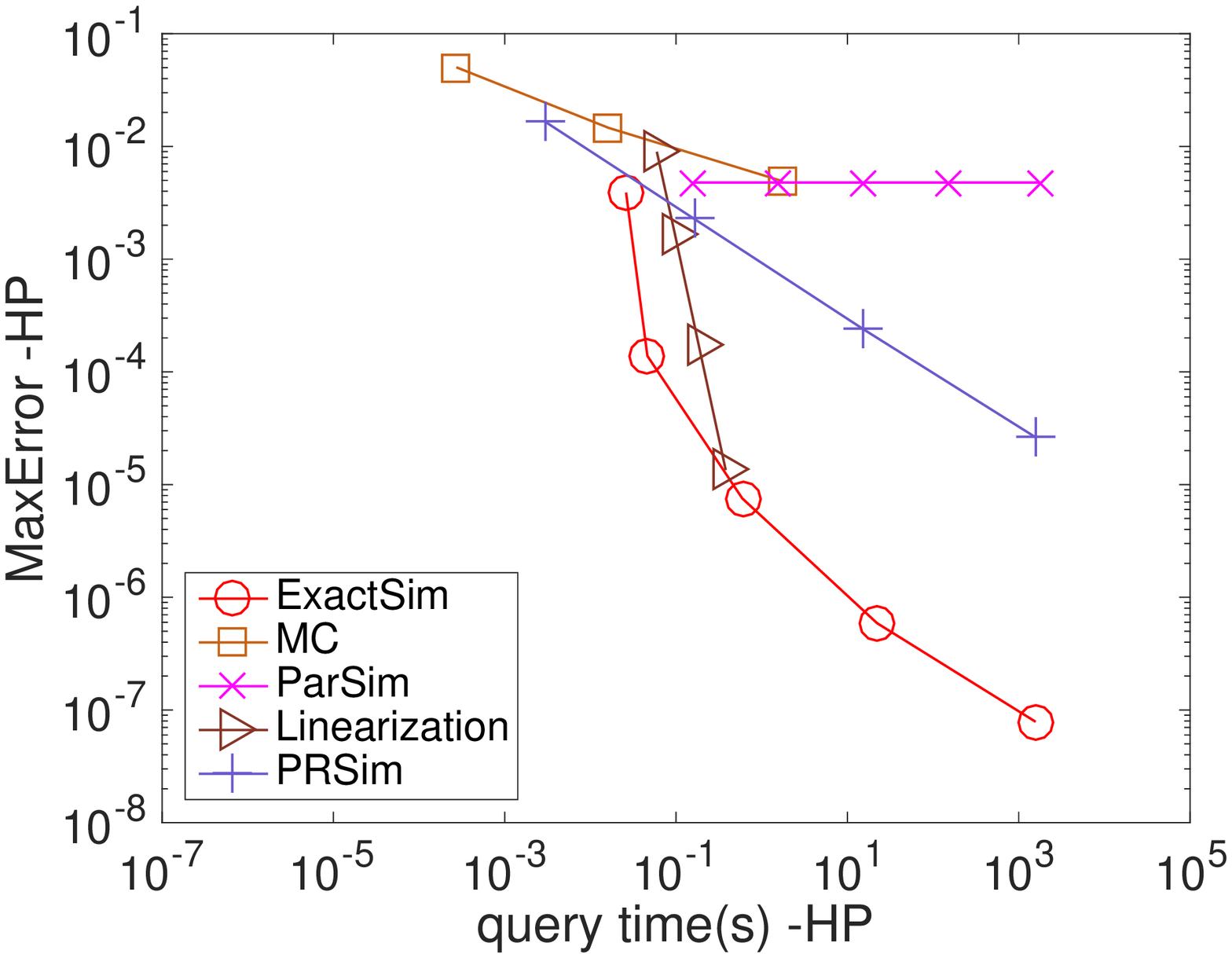}
		\end{tabular}
		\vspace{-3mm}
		\caption{ {MaxError} v.s. Query time on small graphs} 
		\label{fig:maxerr-query-smallgraph}
		%\vspace{-1mm}
	\end{small}
\end{figure*}

\begin{figure*}[!t]
	\begin{small}
		\centering
		%\vspace{-1mm}
		%    \begin{footnotesize}
		\begin{tabular}{cccc}
			%\multicolumn{4}{c}{\hspace{-4mm} \includegraphics[height=5mm]{./Figs/legend_large.eps}} \vspace{-1mm} \\
			\hspace{-6mm} \includegraphics[height=34mm]{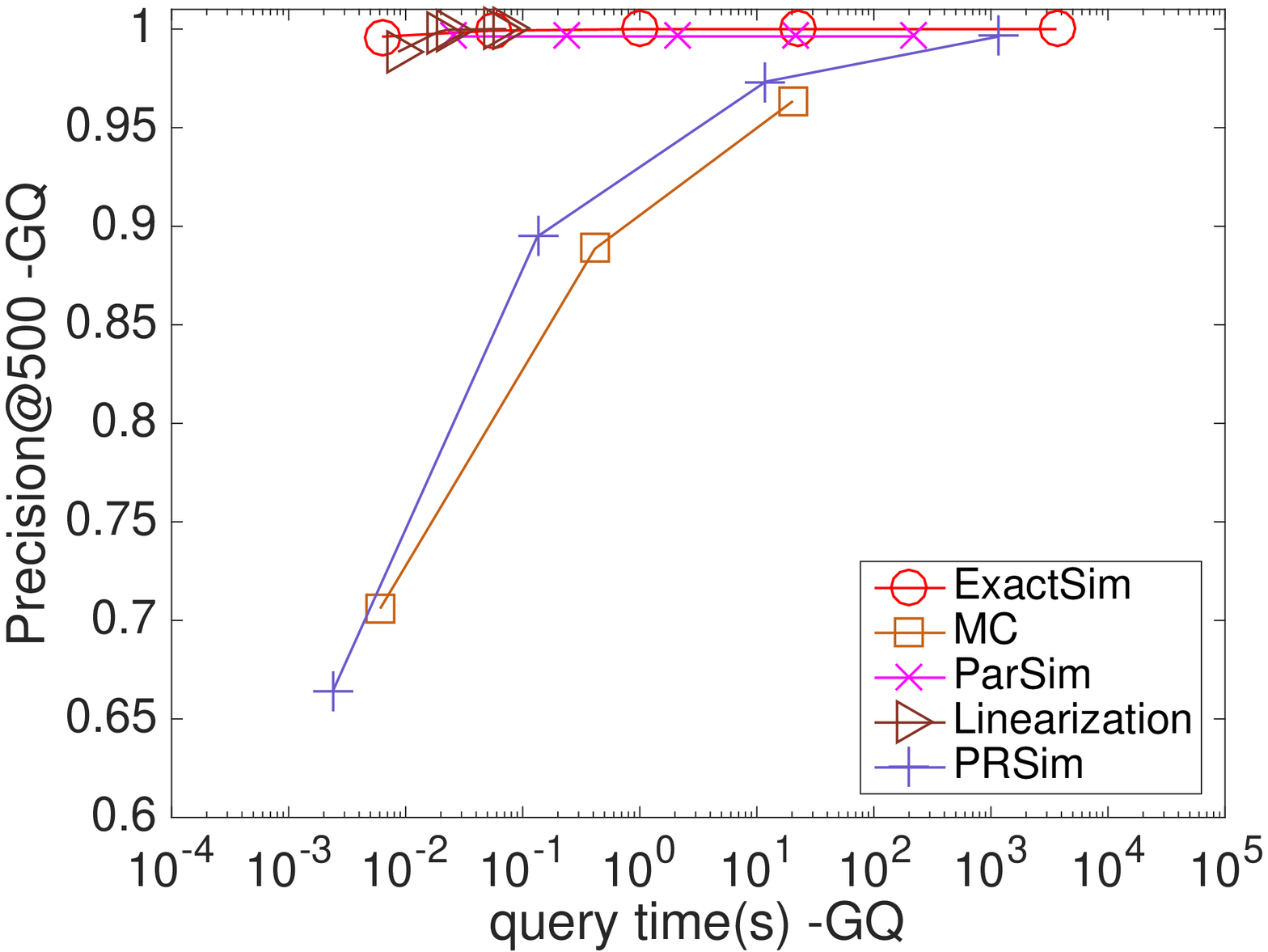} &
			\hspace{-3mm} \includegraphics[height=34mm]{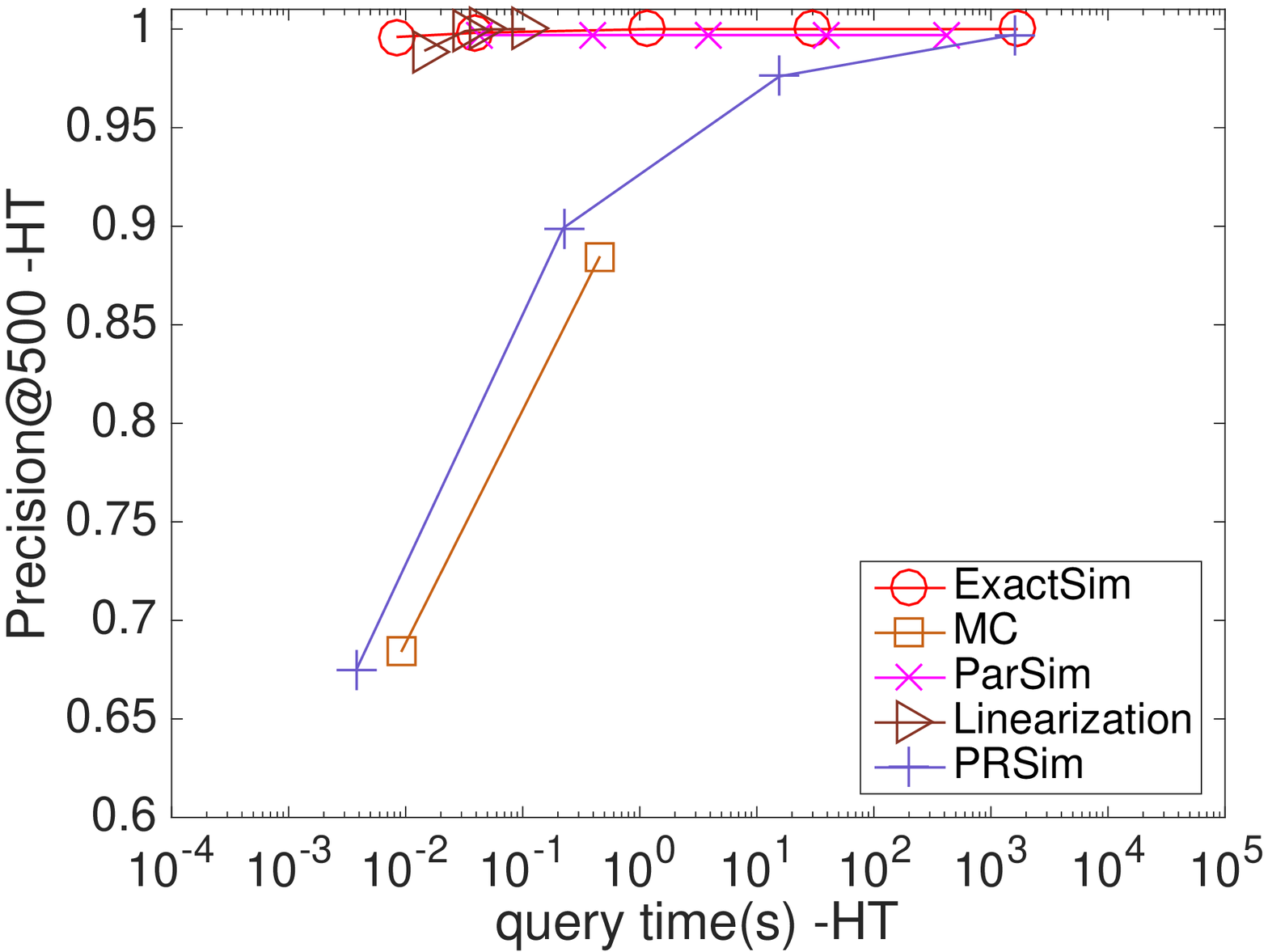} &
			\hspace{-3mm} \includegraphics[height=34mm]{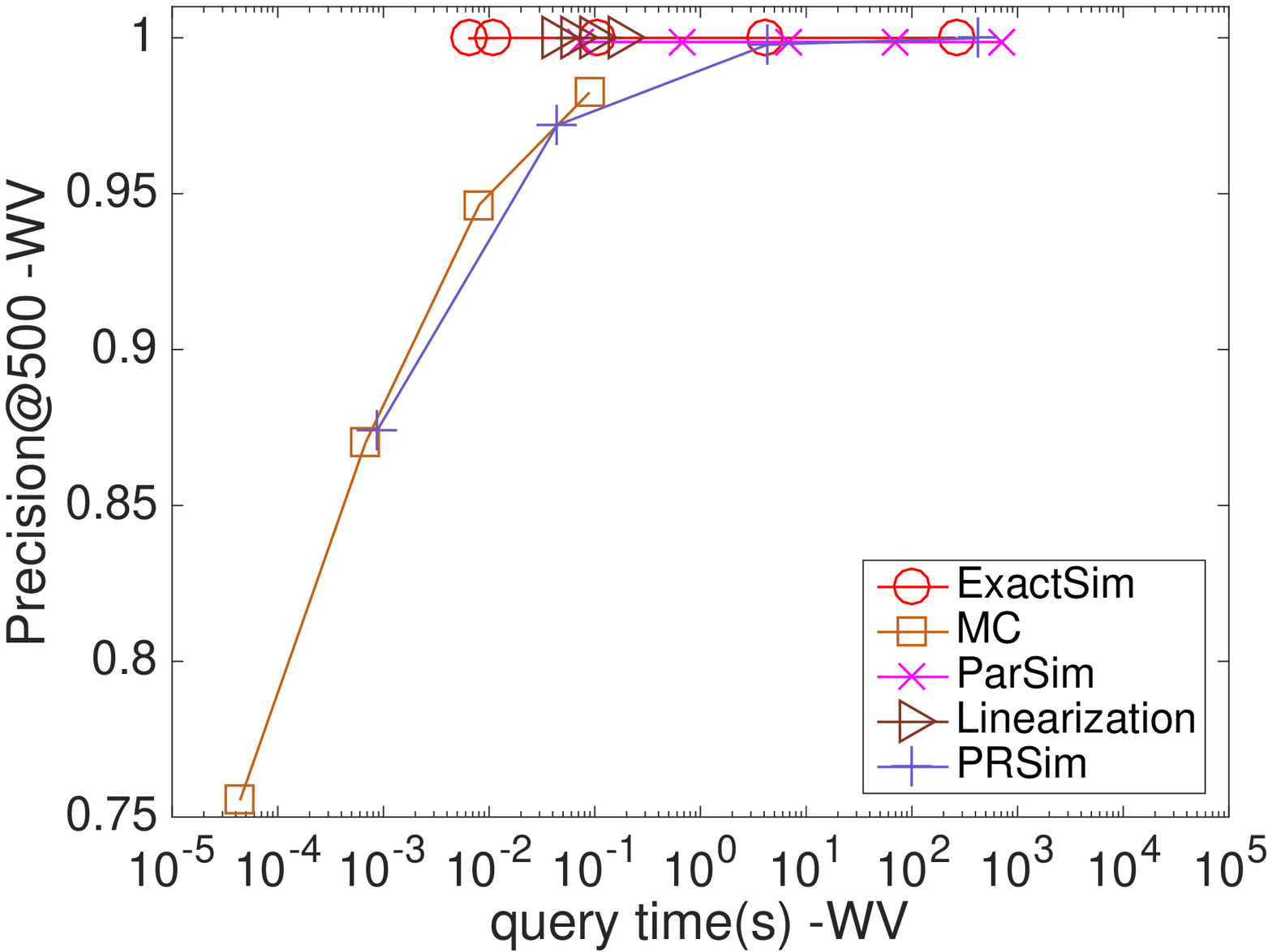}&
			\hspace{-3mm} \includegraphics[height=34mm]{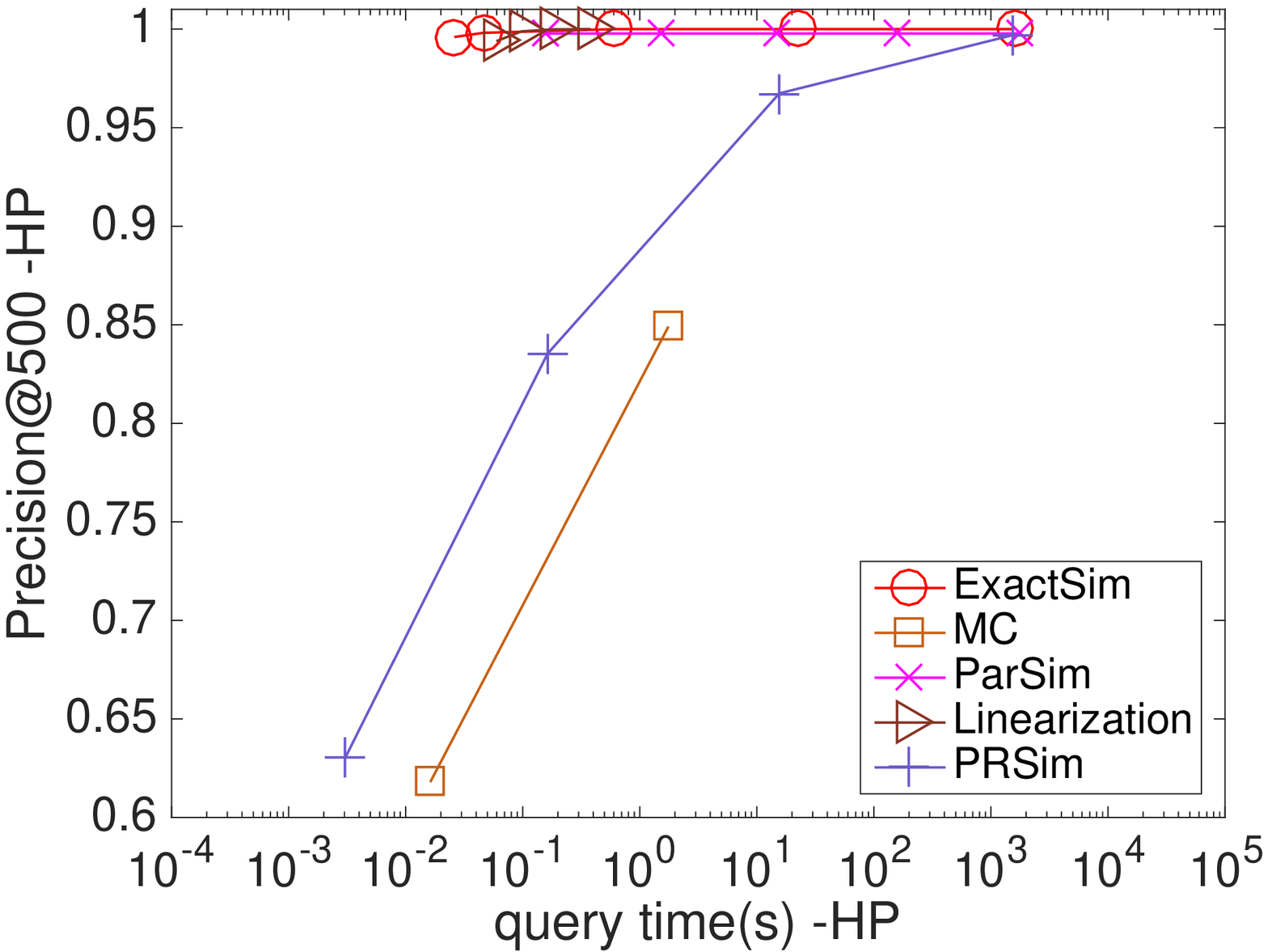}
		\end{tabular}
		\vspace{-3mm}
		\caption{ {Precision@500} v.s. Query time on small graphs} 
		\label{fig:precision-query-smallgraph}
		%\vspace{-1mm}
	\end{small}
\end{figure*}

\begin{figure*}[t]
	\begin{small}
		\centering
		%\vspace{-1mm}
		%    \begin{footnotesize}
		\begin{tabular}{cccc}
			%\multicolumn{4}{c}{\hspace{-4mm} \includegraphics[height=5mm]{./Figs/legend_large.eps}} \vspace{-1mm} \\
			\hspace{-6mm} \includegraphics[height=35mm]{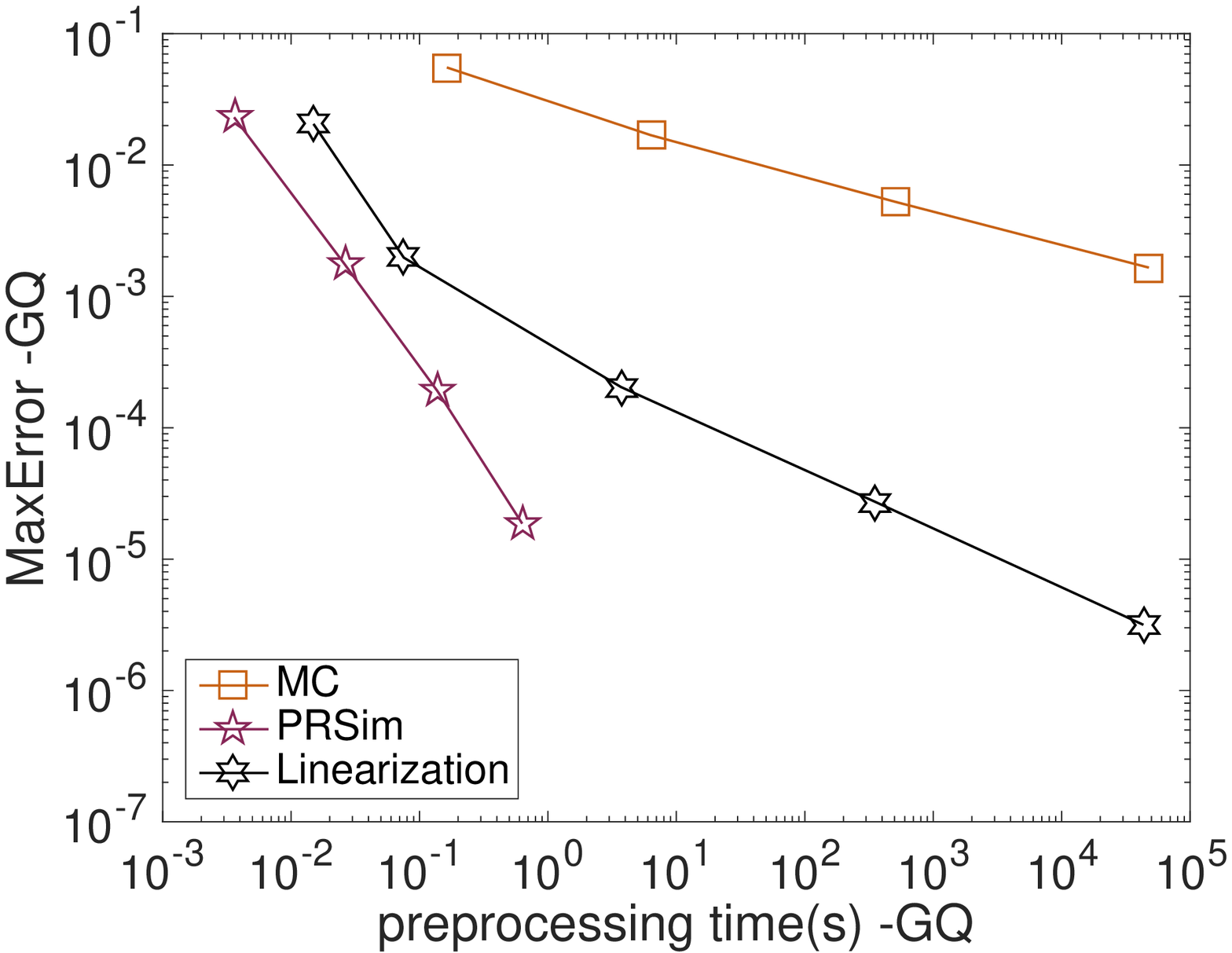} &
			\hspace{-3mm} \includegraphics[height=35mm]{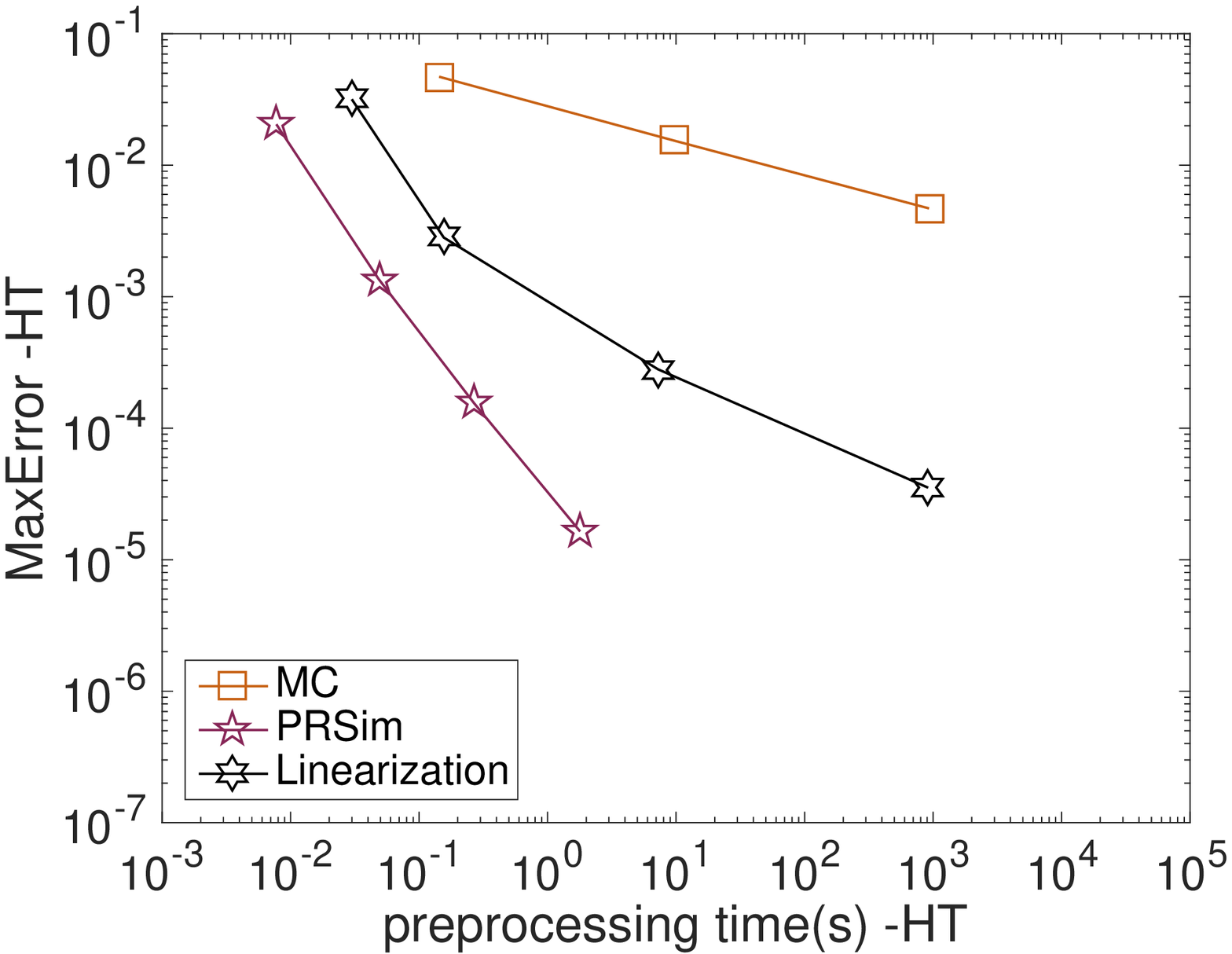} &
			\hspace{-3mm} \includegraphics[height=35mm]{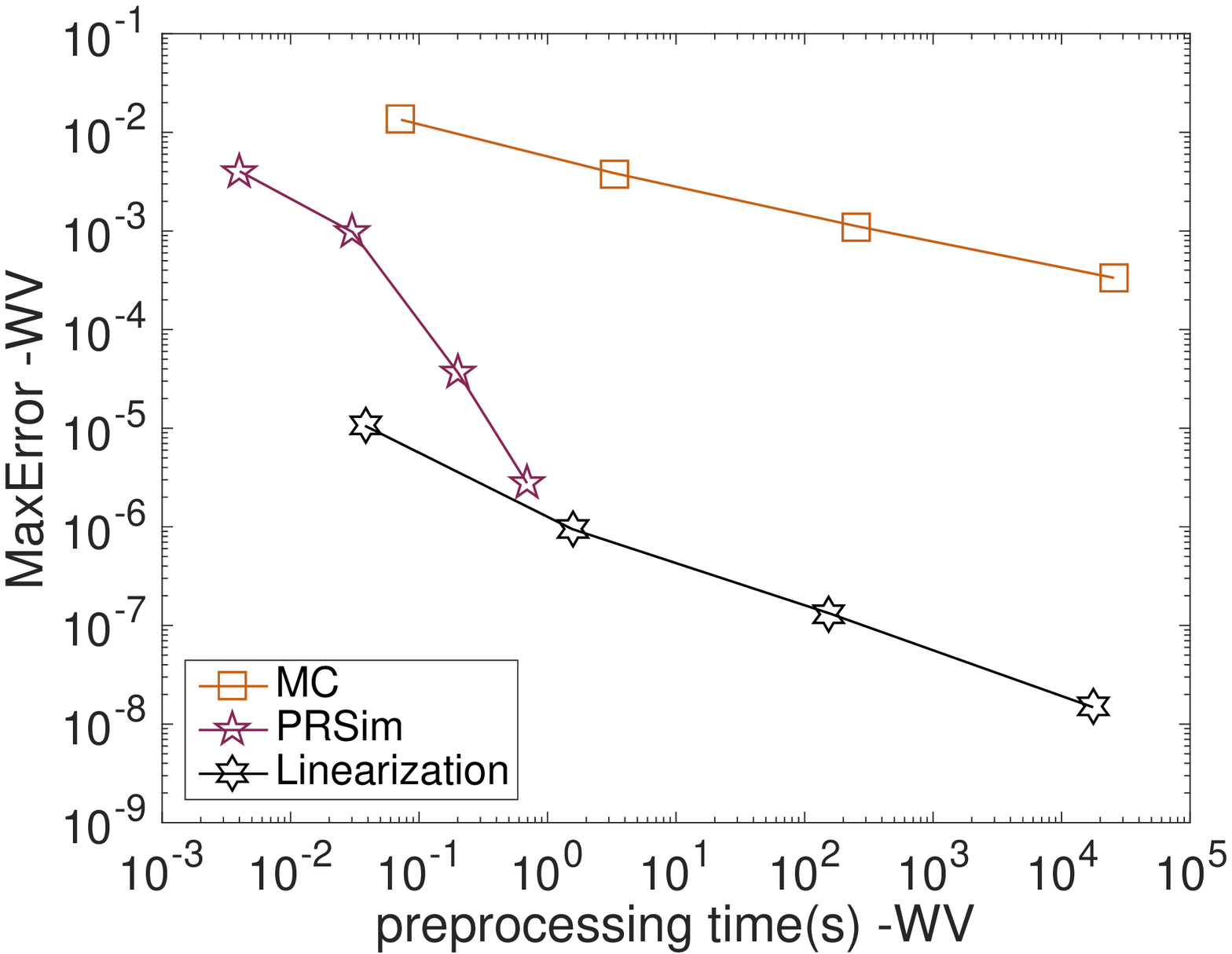}&
			\hspace{-3mm} \includegraphics[height=35mm]{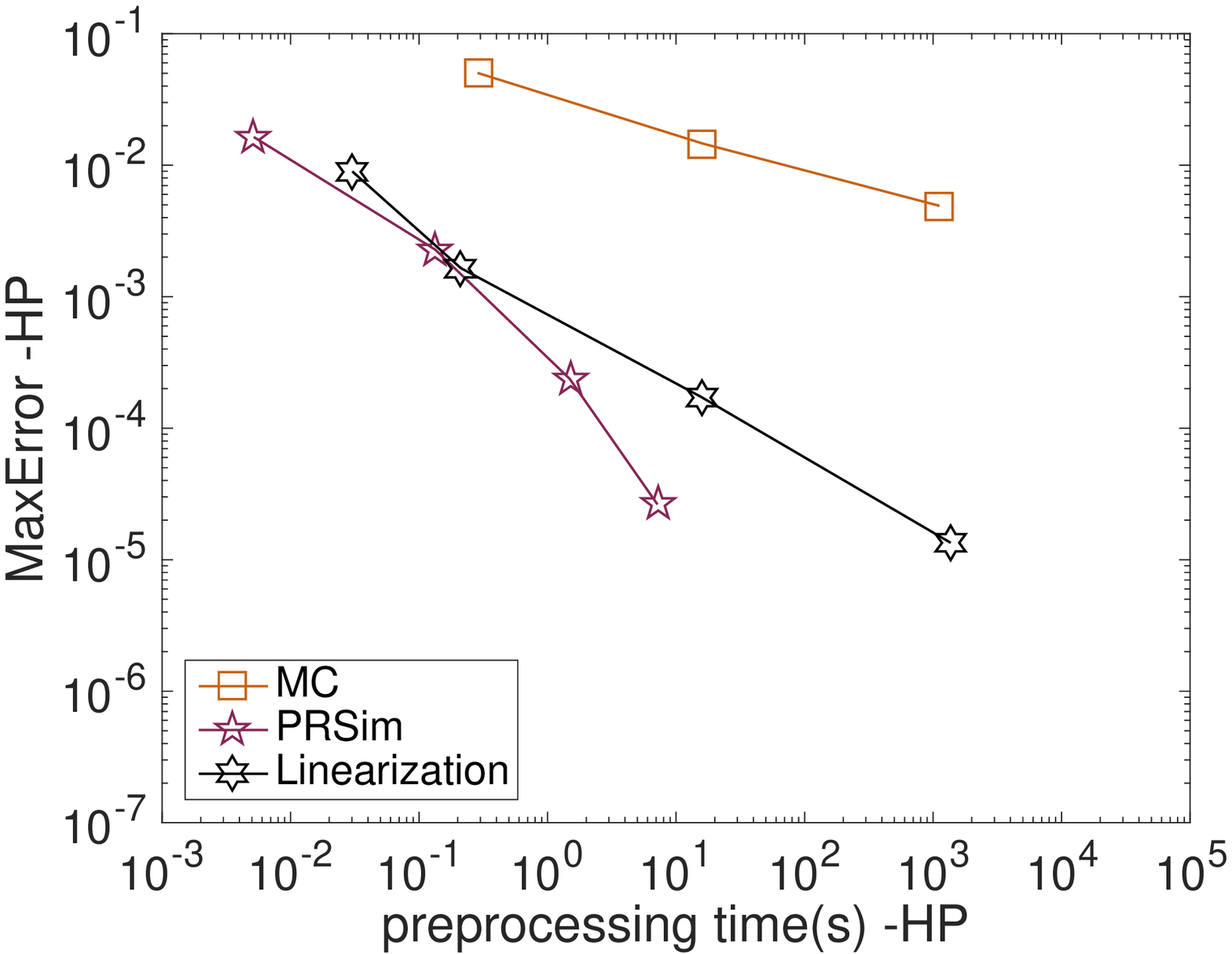}
		\end{tabular}
		\vspace{-3mm}
		\caption{ {MaxError} v.s. Preprocessing time on small graphs} 
		\label{fig:maxerr-pretime-smallgraph}
		%\vspace{-1mm}
	\end{small}
\end{figure*}

\begin{figure*}[t]
	\begin{small}
		\centering
		%\vspace{-1mm}
		%    \begin{footnotesize}
		\begin{tabular}{cccc}
			%\multicolumn{4}{c}{\hspace{-4mm} \includegraphics[height=5mm]{./Figs/legend_large.eps}} \vspace{-1mm} \\
			\hspace{-6mm} \includegraphics[height=35mm]{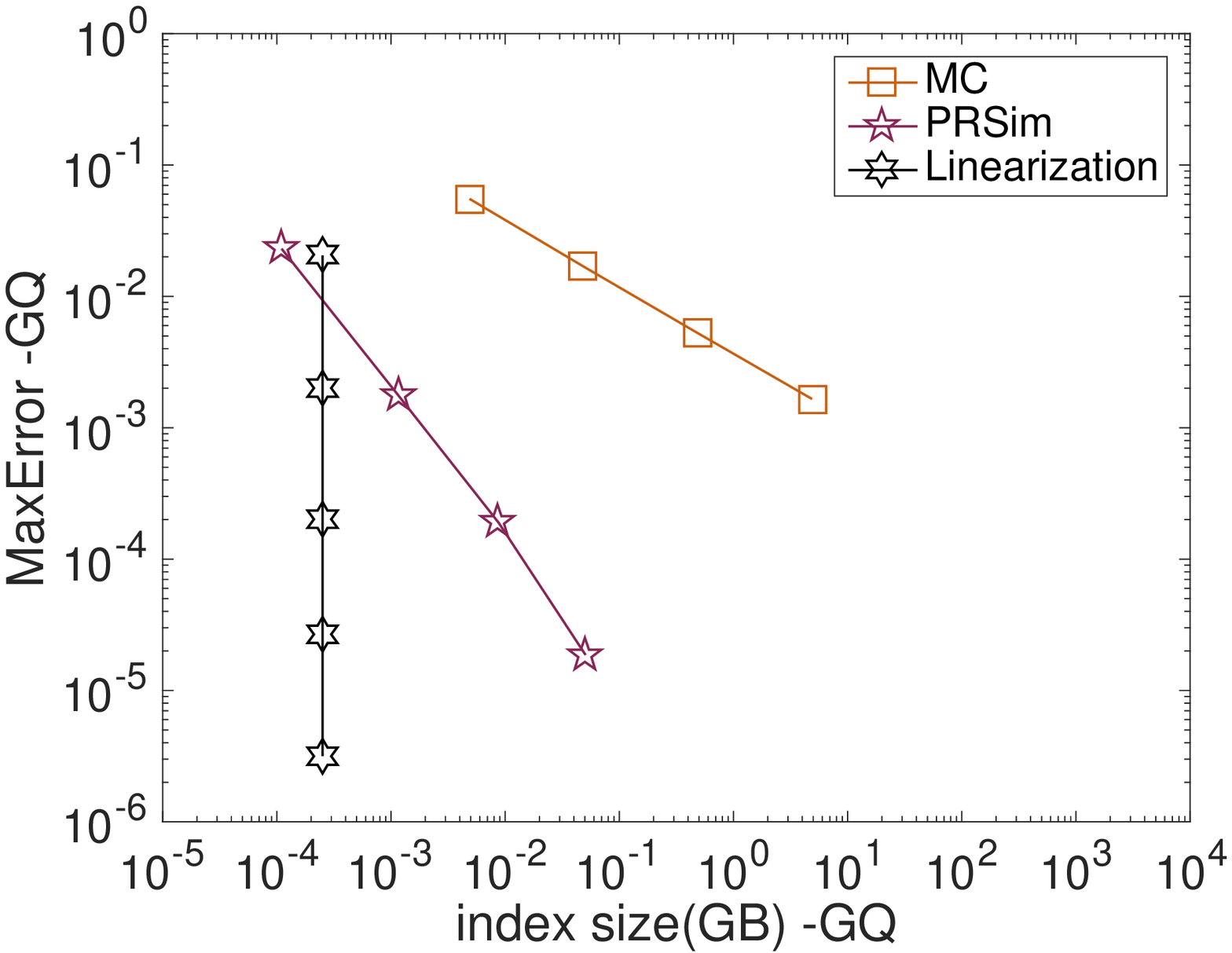} &
			\hspace{-3mm} \includegraphics[height=35mm]{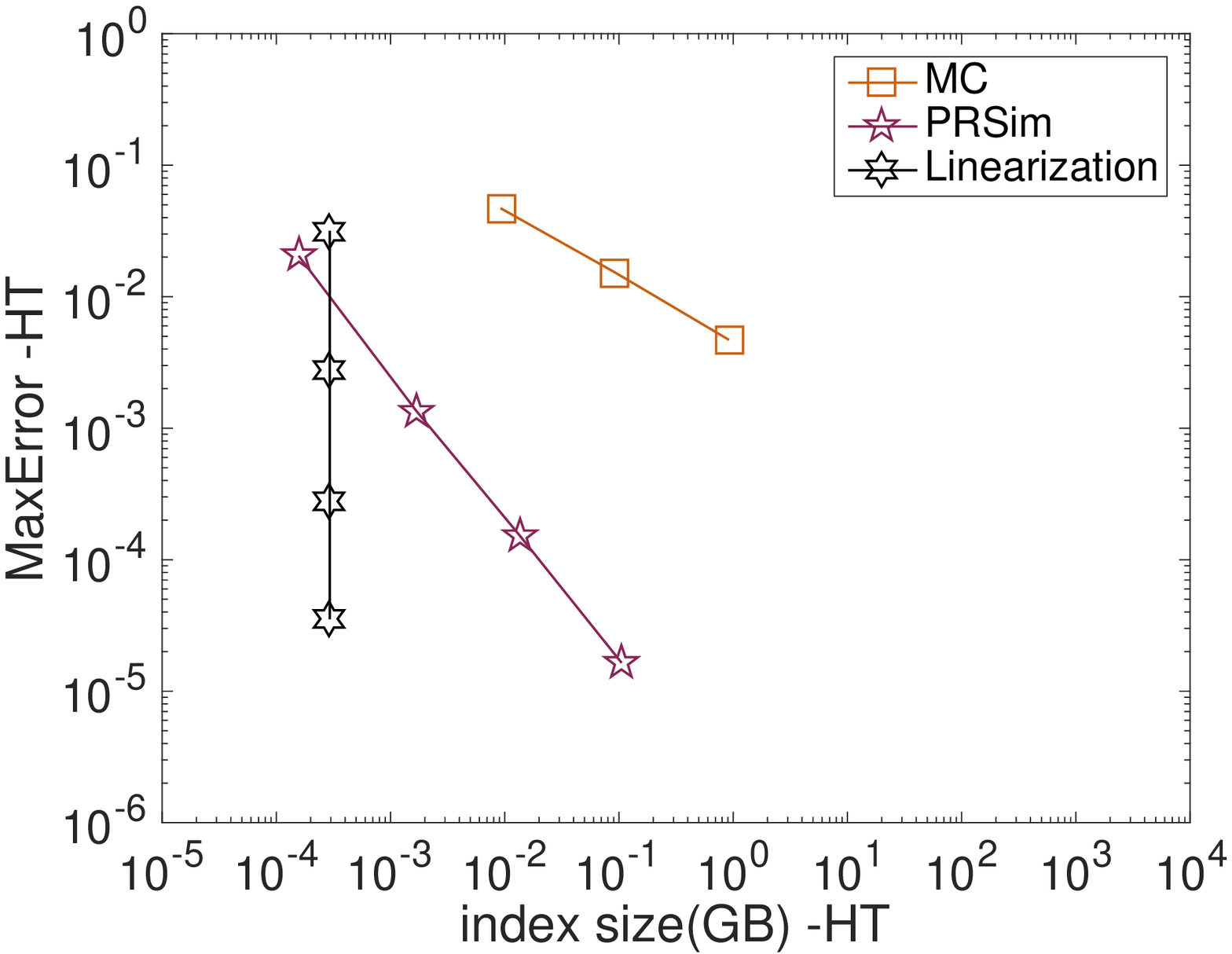} &
			\hspace{-3mm} \includegraphics[height=35mm]{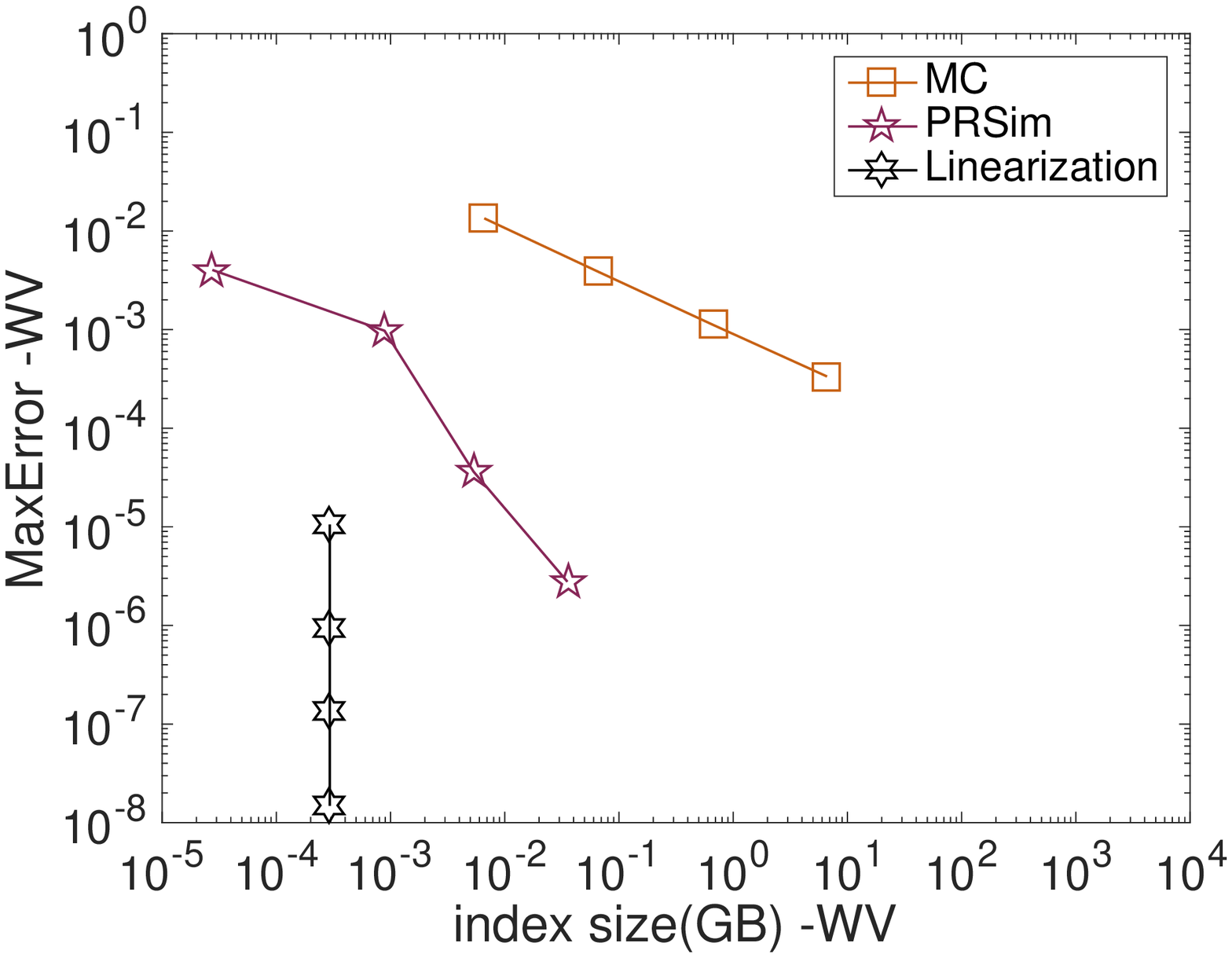}
			\hspace{-0mm} \includegraphics[height=35mm]{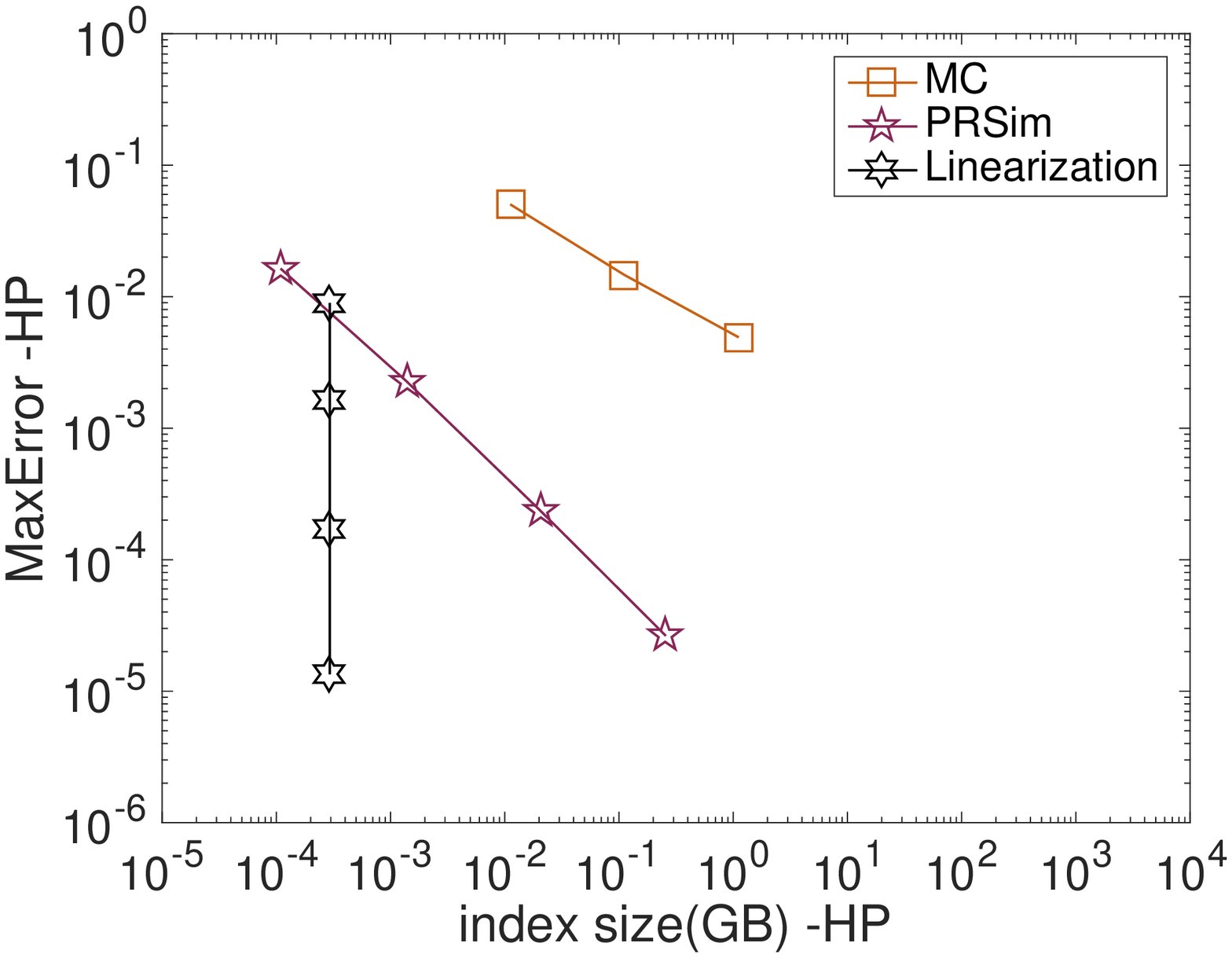}
		\end{tabular}
		\vspace{-3mm}
		\caption{ {MaxError} v.s. Index size on small graphs} 
		\label{fig:maxerr-indexsize-smallgraph}
		%\vspace{-0mm}
	\end{small}
\end{figure*}

\begin{figure*}[t]
	\begin{small}
		\centering
		%\vspace{-1mm}
		%\vspace{-1mm}
		%    \begin{footnotesize}
		\begin{tabular}{cccc}
			%\multicolumn{4}{c}{\hspace{-4mm} \includegraphics[height=5mm]{./Figs/legend_large.eps}} \vspace{-1mm} \\
			\hspace{-6mm} \includegraphics[height=35mm]{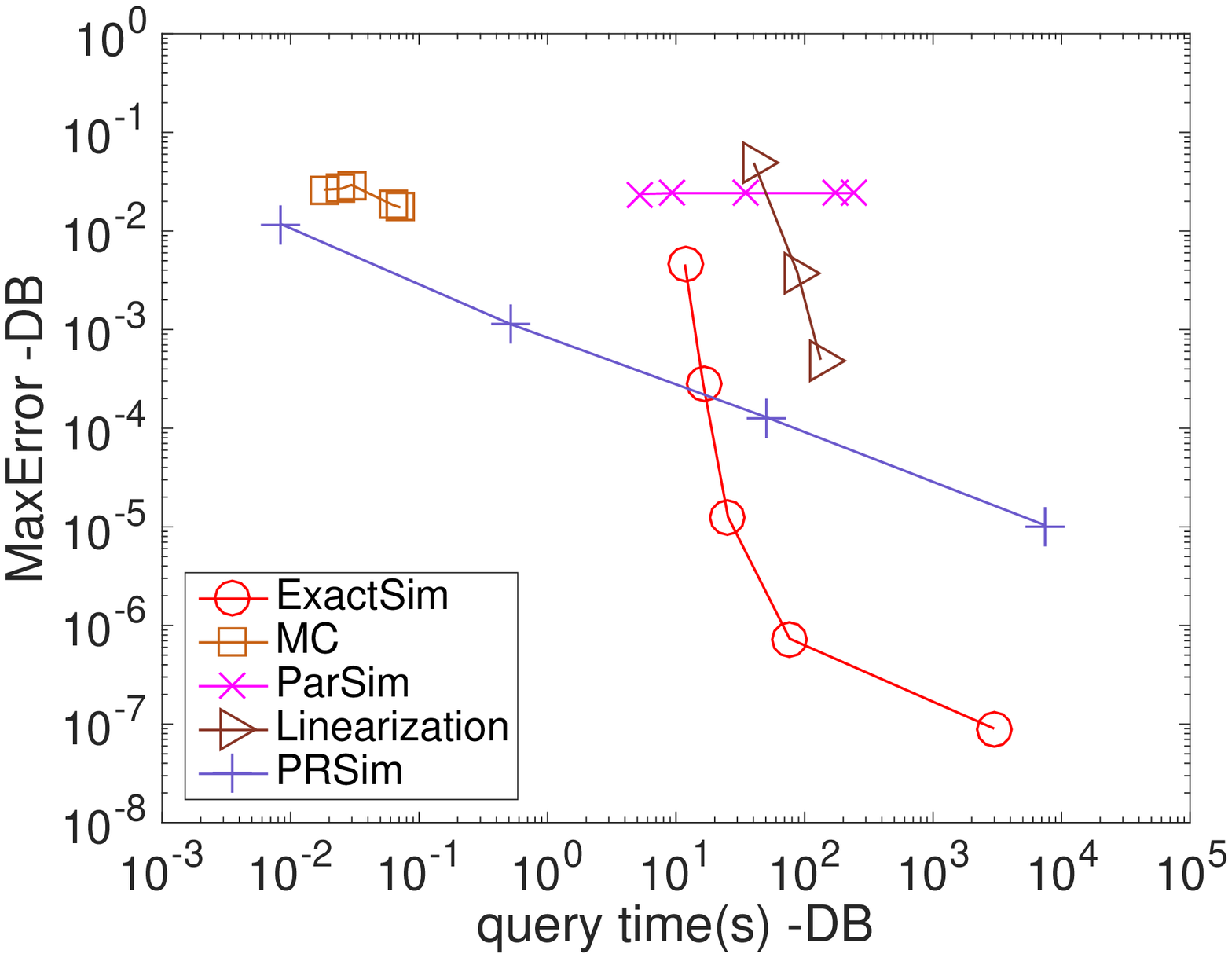} &
			%\hspace{-3mm} \includegraphics[height=25mm]{./Figs/maxerr-query-largegraph-LJ.eps} &
			\hspace{-3mm} \includegraphics[height=35mm]{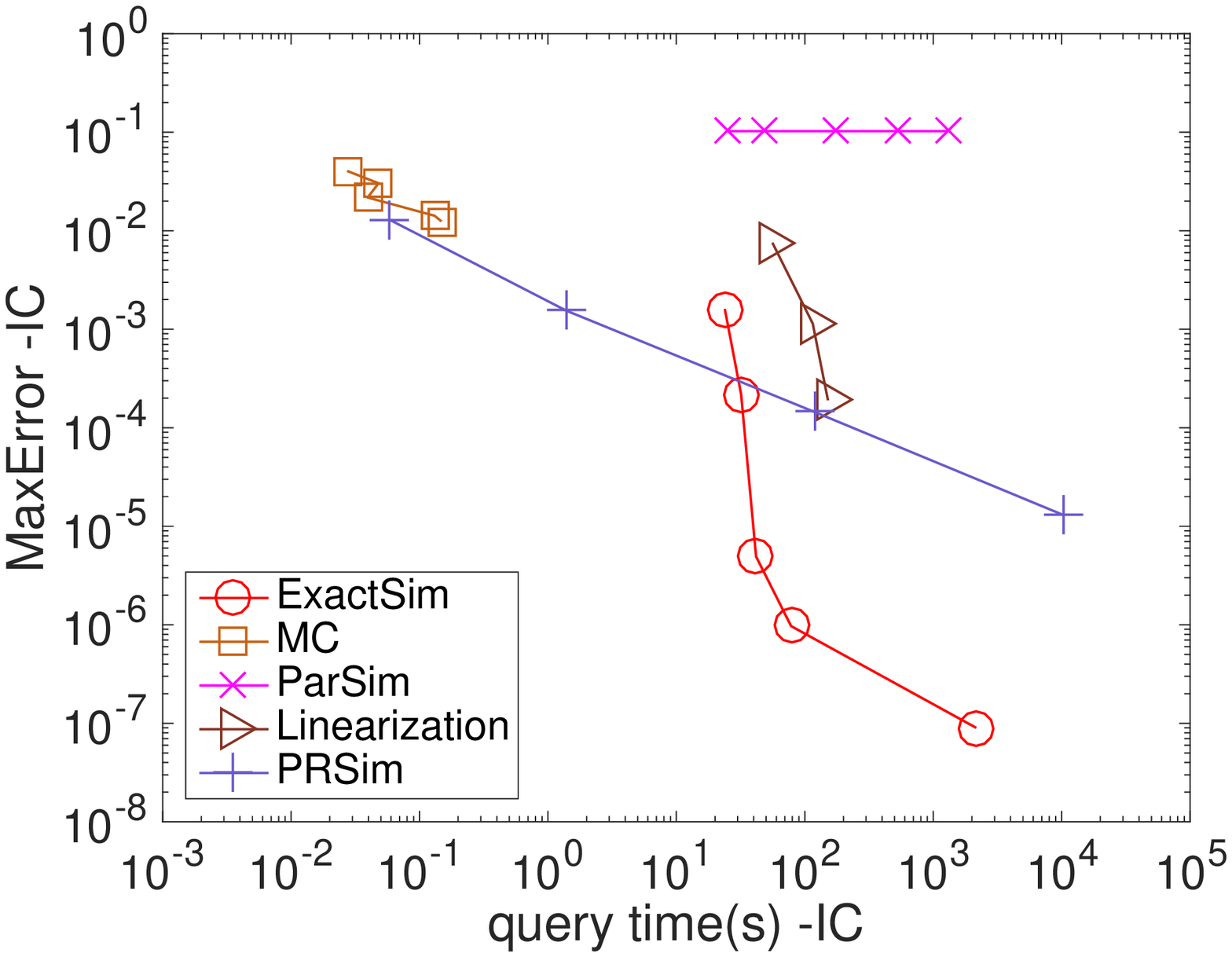} &
			%\hspace{-3mm} \includegraphics[height=25mm]{./Figs/maxerr-query-largegraph-OL.eps} &
			\hspace{-3mm} \includegraphics[height=35mm]{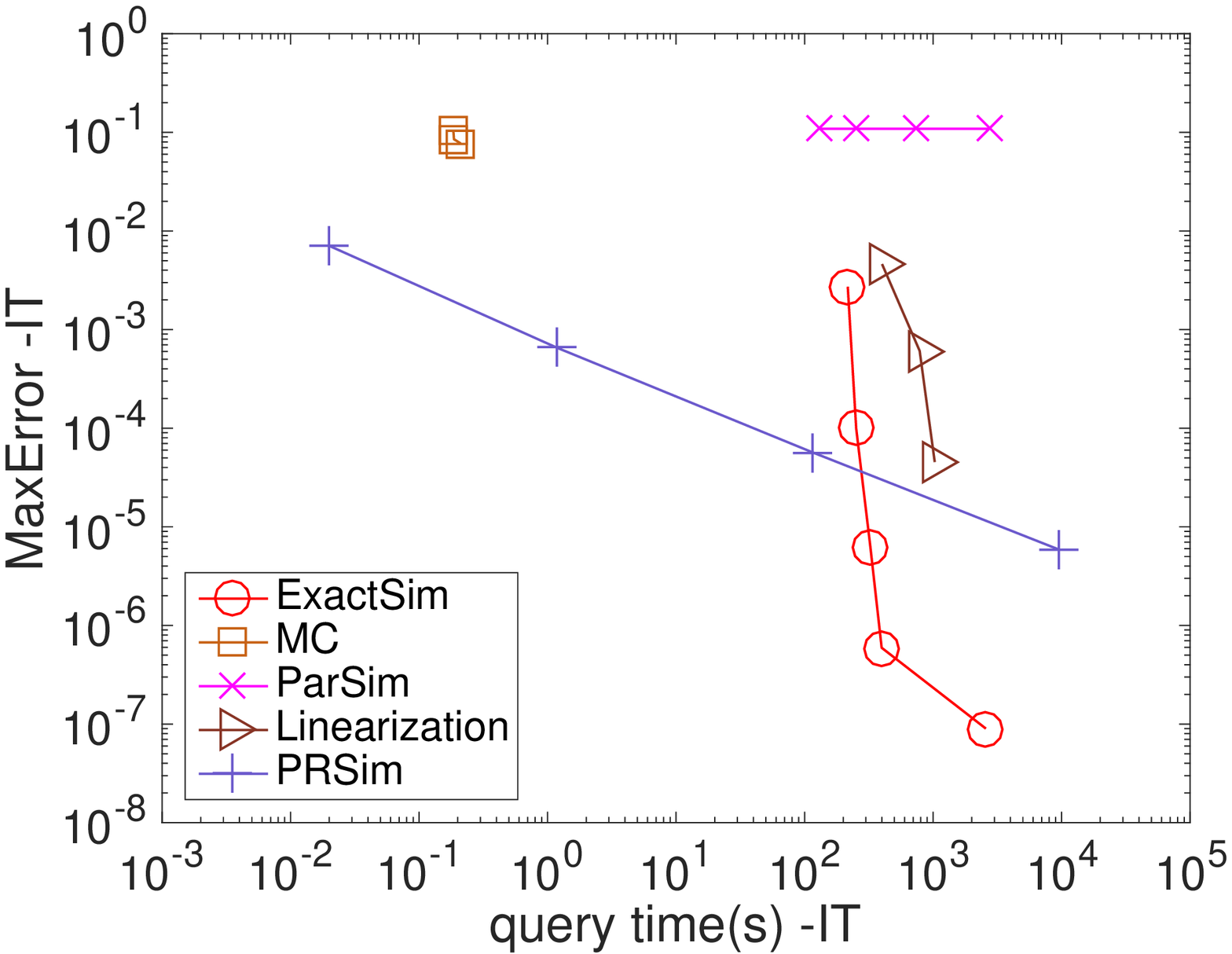} &
			\hspace{-3mm} \includegraphics[height=35mm]{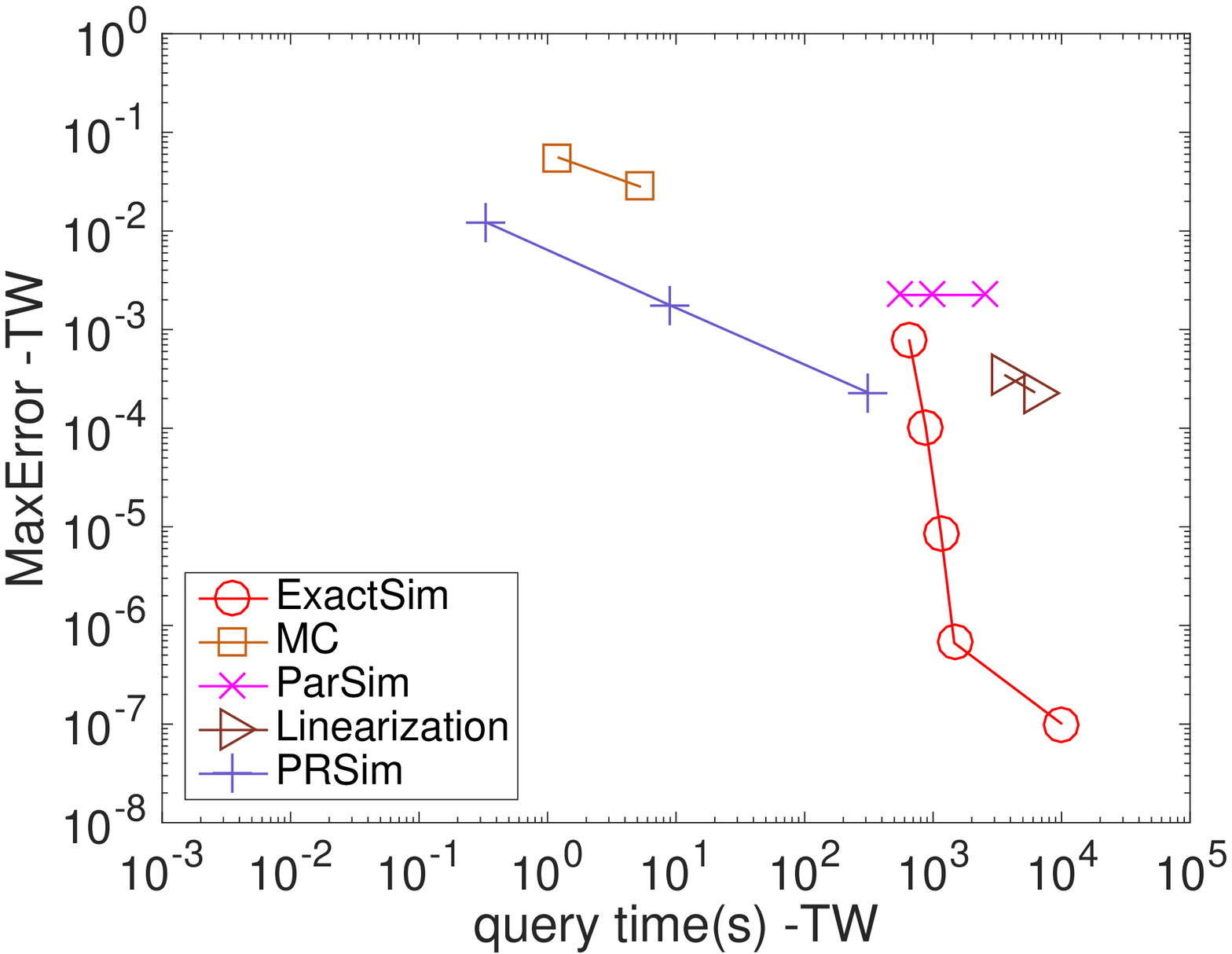}
		\end{tabular}
		\vspace{-3mm}
		\caption{ {MaxError} v.s. Query time on large graphs} 
		\label{fig:maxerr-query-largegraph}
		%\vspace{-1mm}
	\end{small}
\end{figure*}

\begin{figure*}[!t]
	\begin{small}
		\centering
		%\vspace{-1mm}
		%    \begin{footnotesize}
		\begin{tabular}{cccc}
			%\multicolumn{4}{c}{\hspace{-4mm} \includegraphics[height=5mm]{./Figs/legend_large.eps}} \vspace{-1mm} \\
			\hspace{-6mm} \includegraphics[height=34mm]{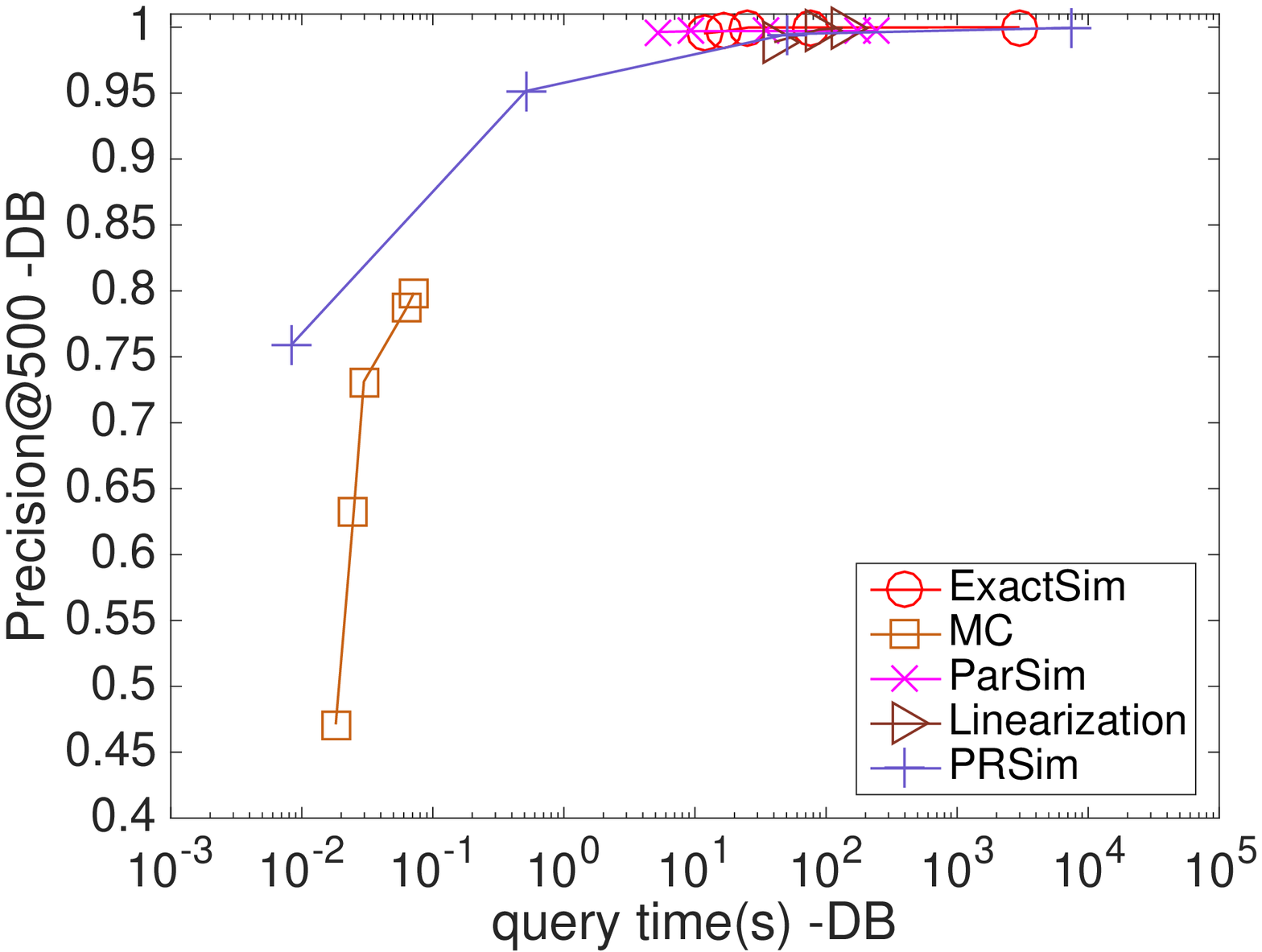} &
			%\hspace{-3mm} \includegraphics[height=24.5mm]{./Figs/precision-query-largegraph-LJ.eps} &
			\hspace{-3mm} \includegraphics[height=34mm]{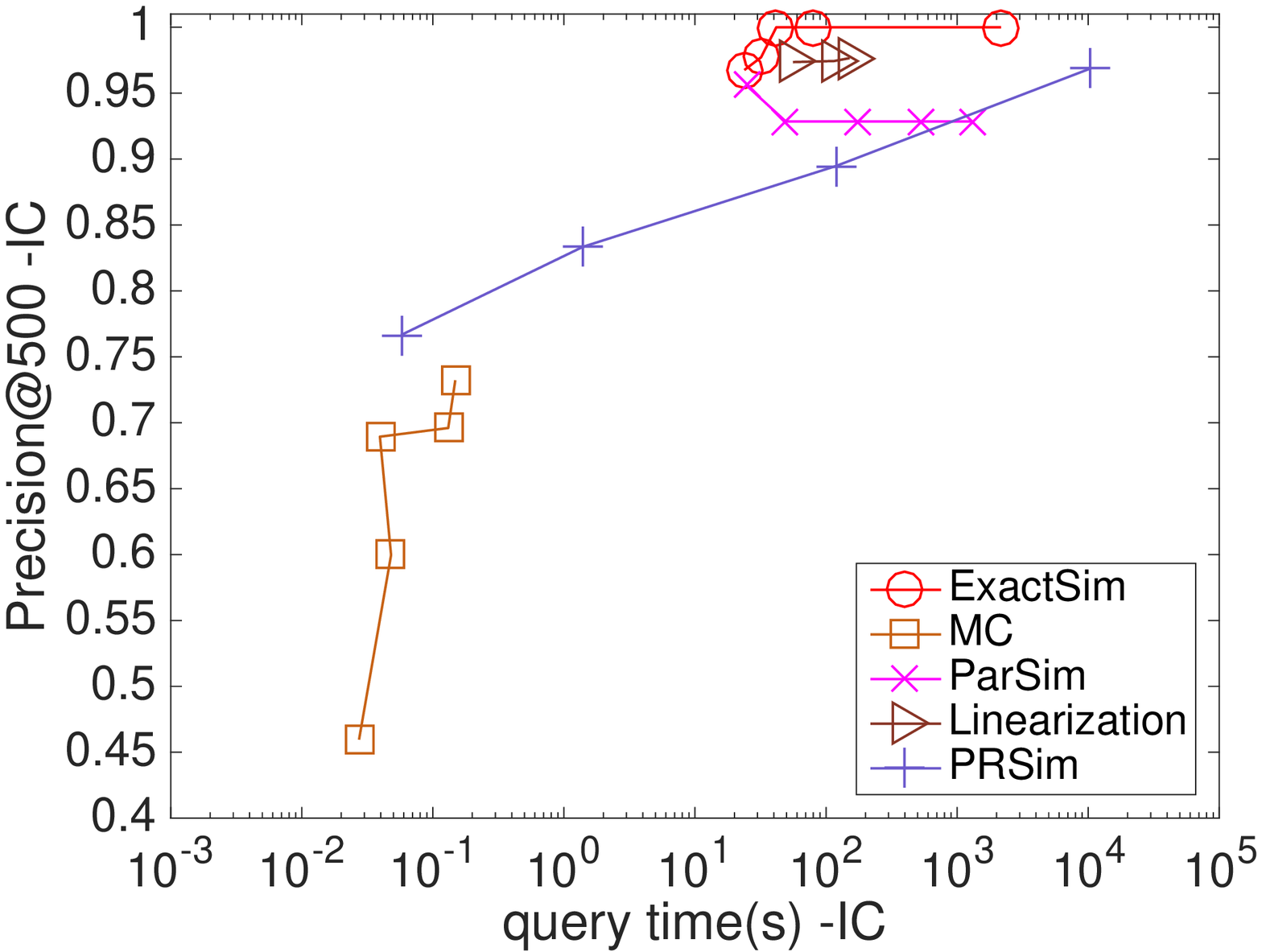} &
			%\hspace{-3mm} \includegraphics[height=24.5mm]{./Figs/precision-query-largegraph-OL.eps} &
			\hspace{-3mm} \includegraphics[height=34mm]{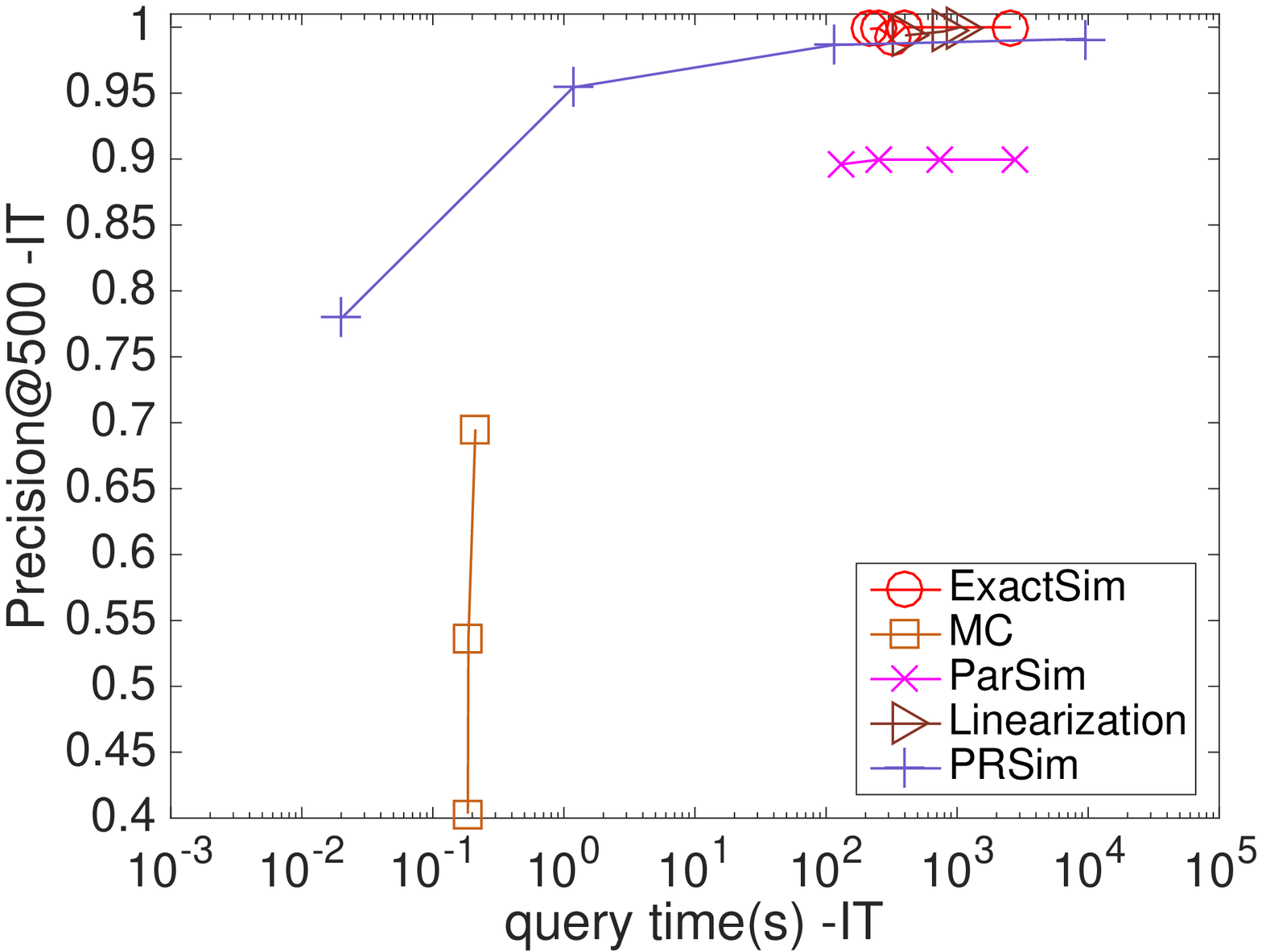} &
			\hspace{-3mm} \includegraphics[height=34mm]{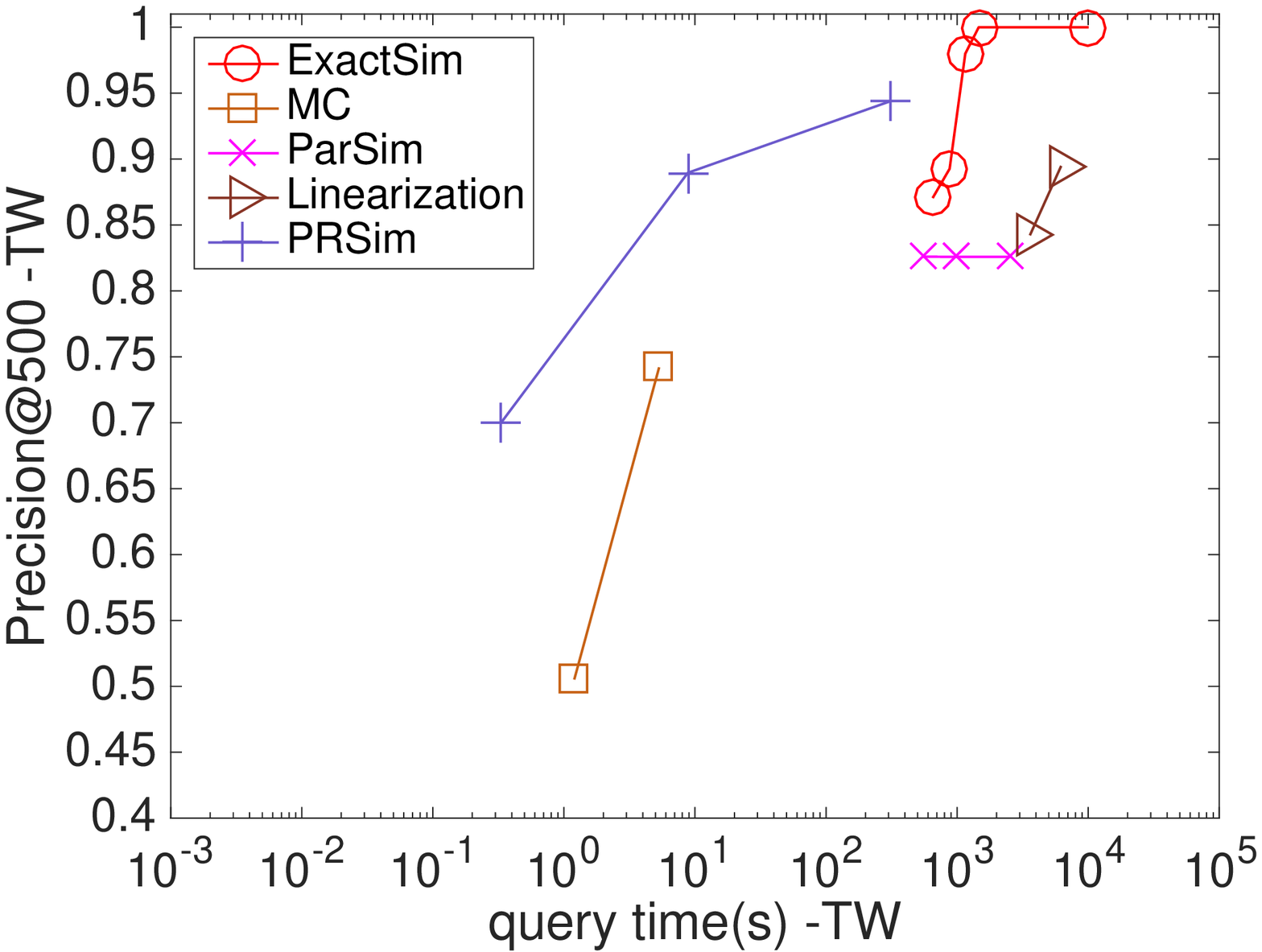}
		\end{tabular}
		\vspace{-3mm}
		\caption{ {Precision@500} v.s. Query time on large graphs} 
		\label{fig:precision-query-largegraph}
		%\vspace{-1mm}
	\end{small}
\end{figure*}

\begin{figure*}[t]
	\begin{small}
		\centering
		%\vspace{-1mm}
		%    \begin{footnotesize}
		\begin{tabular}{cccc}
			%\multicolumn{4}{c}{\hspace{-4mm} \includegraphics[height=5mm]{./Figs/legend_large.eps}} \vspace{-1mm} \\
			\hspace{-6mm} \includegraphics[height=35mm]{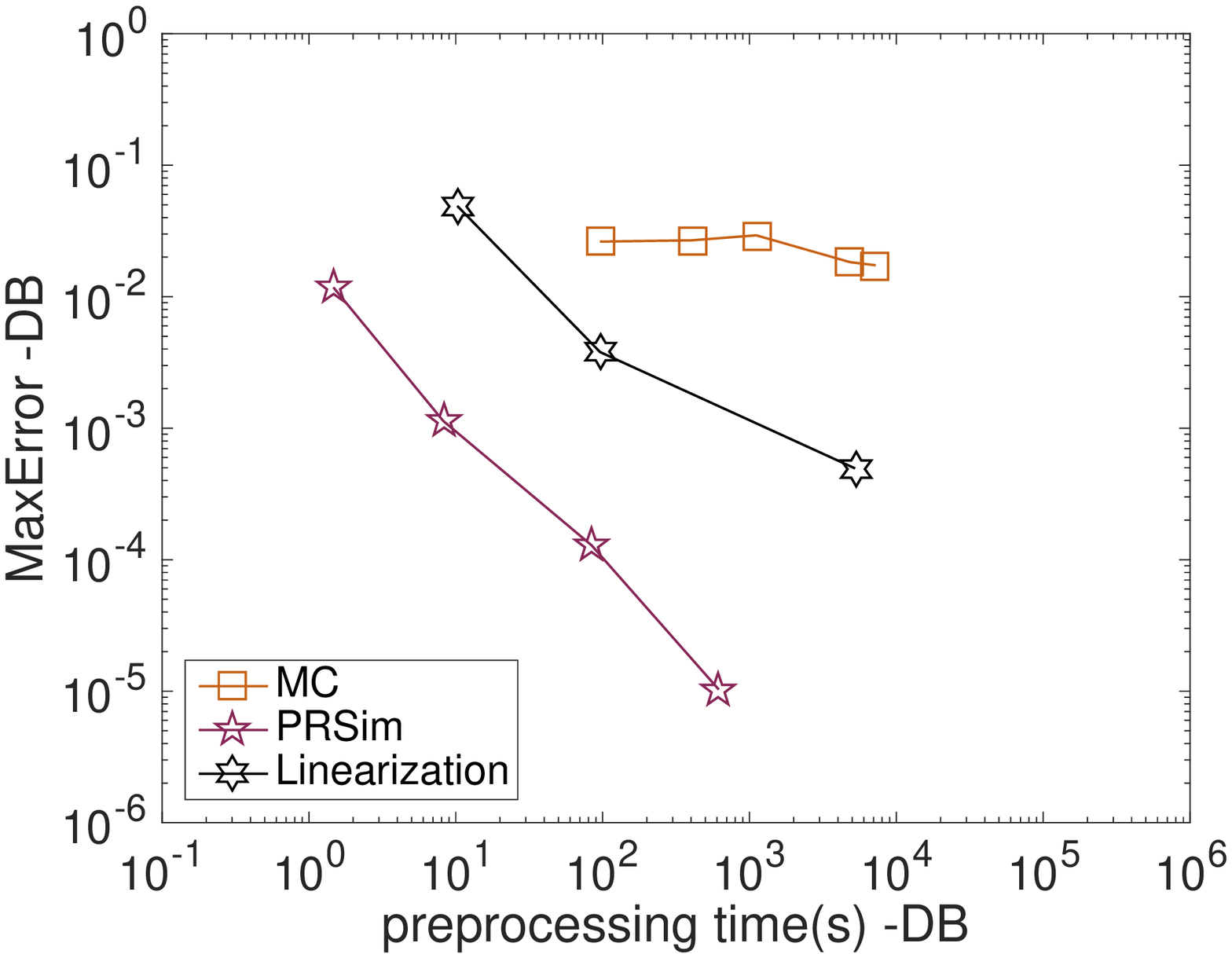} &
			%\hspace{-3mm} \includegraphics[height=25mm]{./Figs/maxerr-pretime-largegraph-LJ.eps} &
			\hspace{-3mm} \includegraphics[height=35mm]{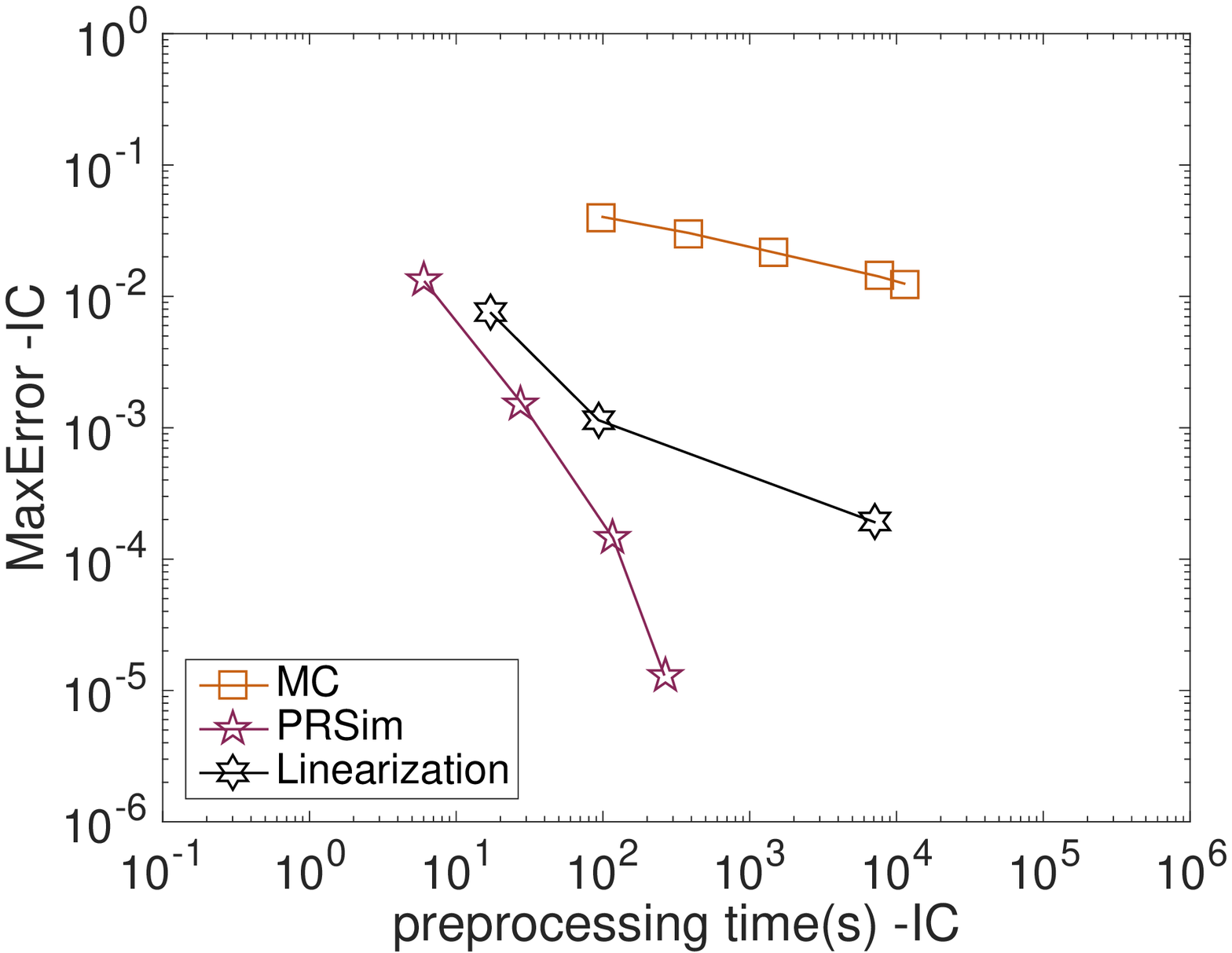} &
			%\hspace{-3mm} \includegraphics[height=25mm]{./Figs/maxerr-pretime-largegraph-OL.eps} &
			\hspace{-3mm} \includegraphics[height=35mm]{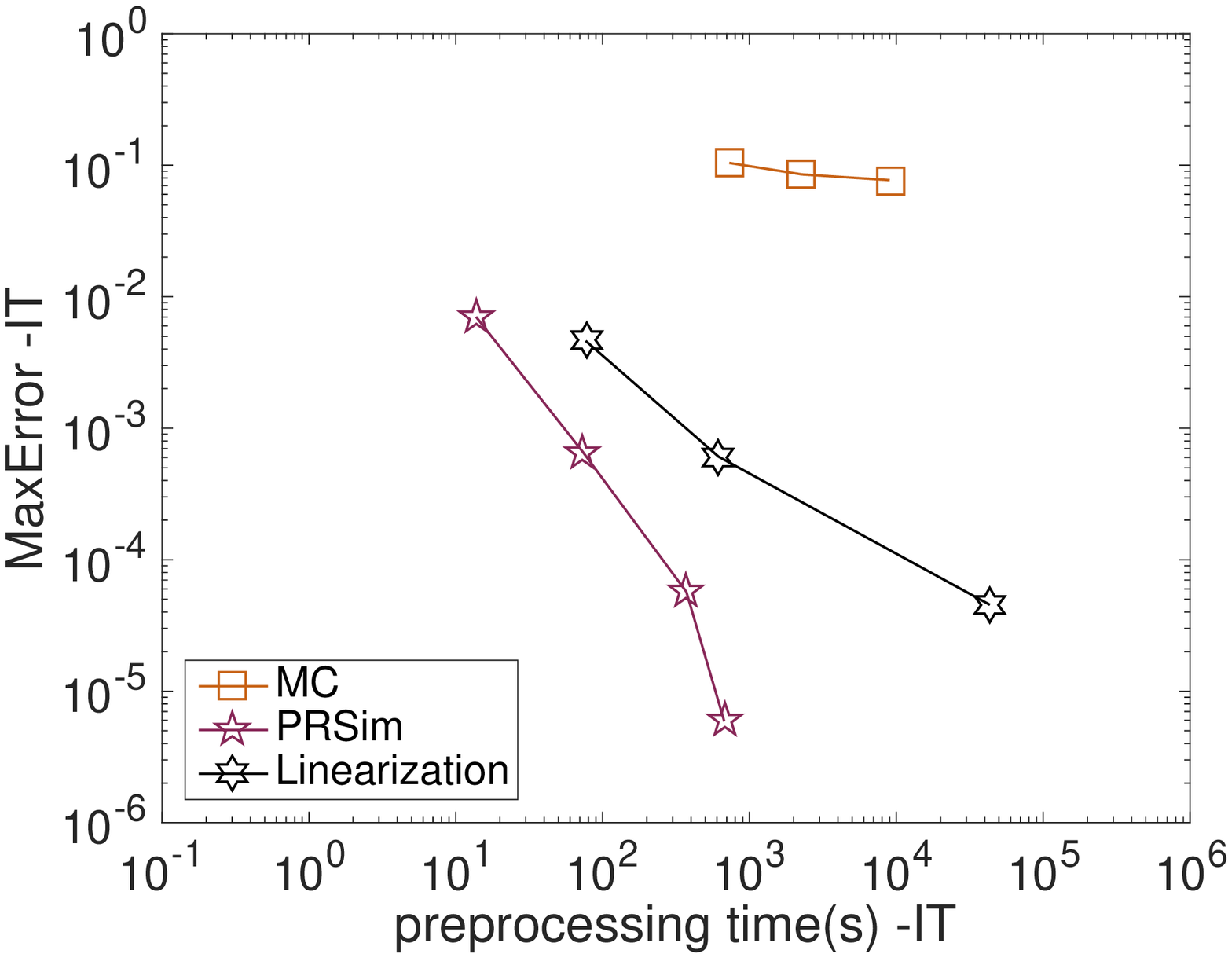} &
			\hspace{-3mm} \includegraphics[height=35mm]{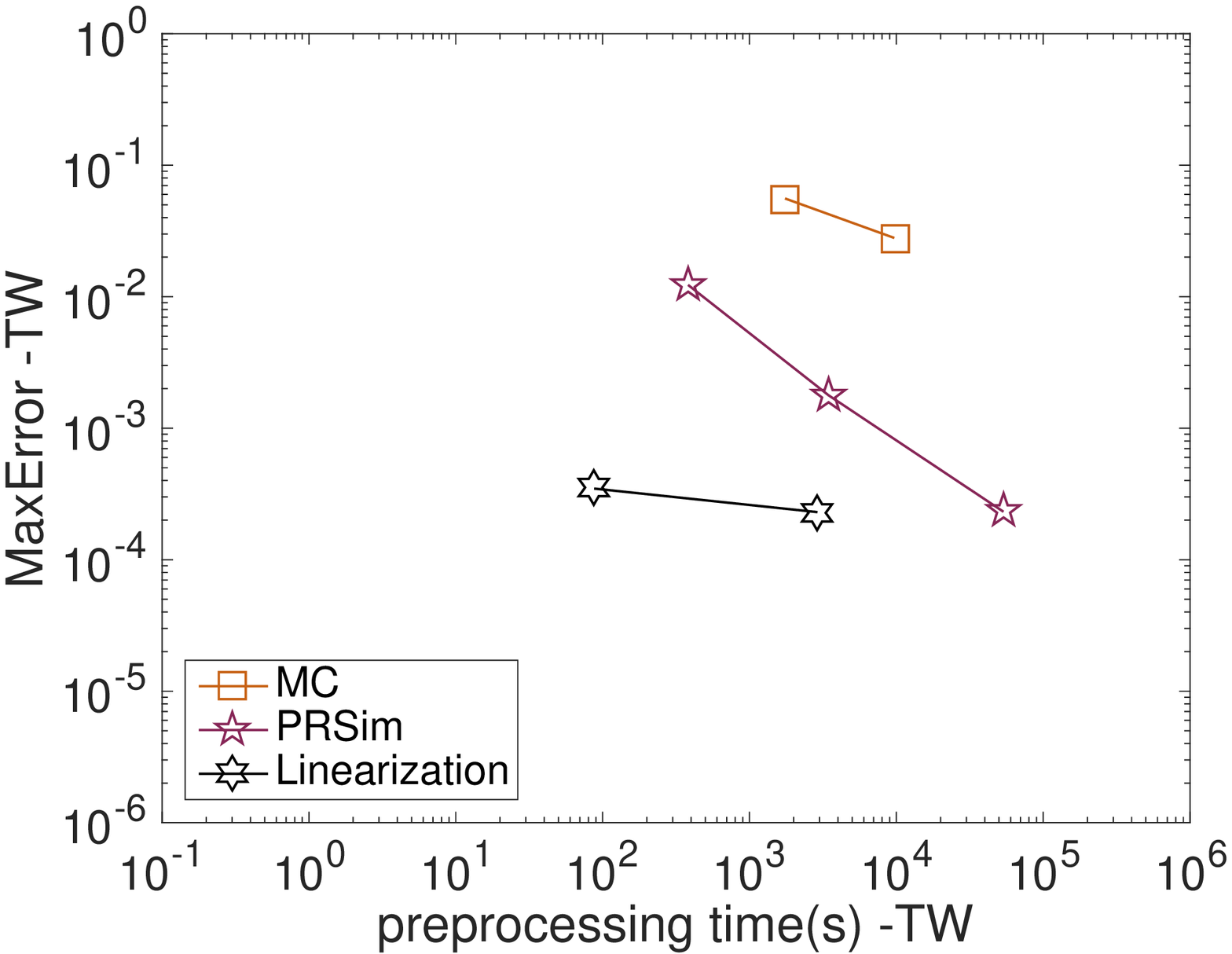}
		\end{tabular}
		\vspace{-3mm}
		\caption{ {MaxError} v.s. Preprocessing time on large graphs} 
		\label{fig:maxerr-pretime-largegraph}
		%\vspace{-2mm}
	\end{small}
\end{figure*}

\begin{figure*}[t]
	\begin{small}
		\centering
		%\vspace{-1mm}
		%    \begin{footnotesize}
		\begin{tabular}{cccc}
			%\multicolumn{4}{c}{\hspace{-4mm} \includegraphics[height=5mm]{./Figs/legend_large.eps}} \vspace{-1mm} \\
			\hspace{-6mm} \includegraphics[height=35mm]{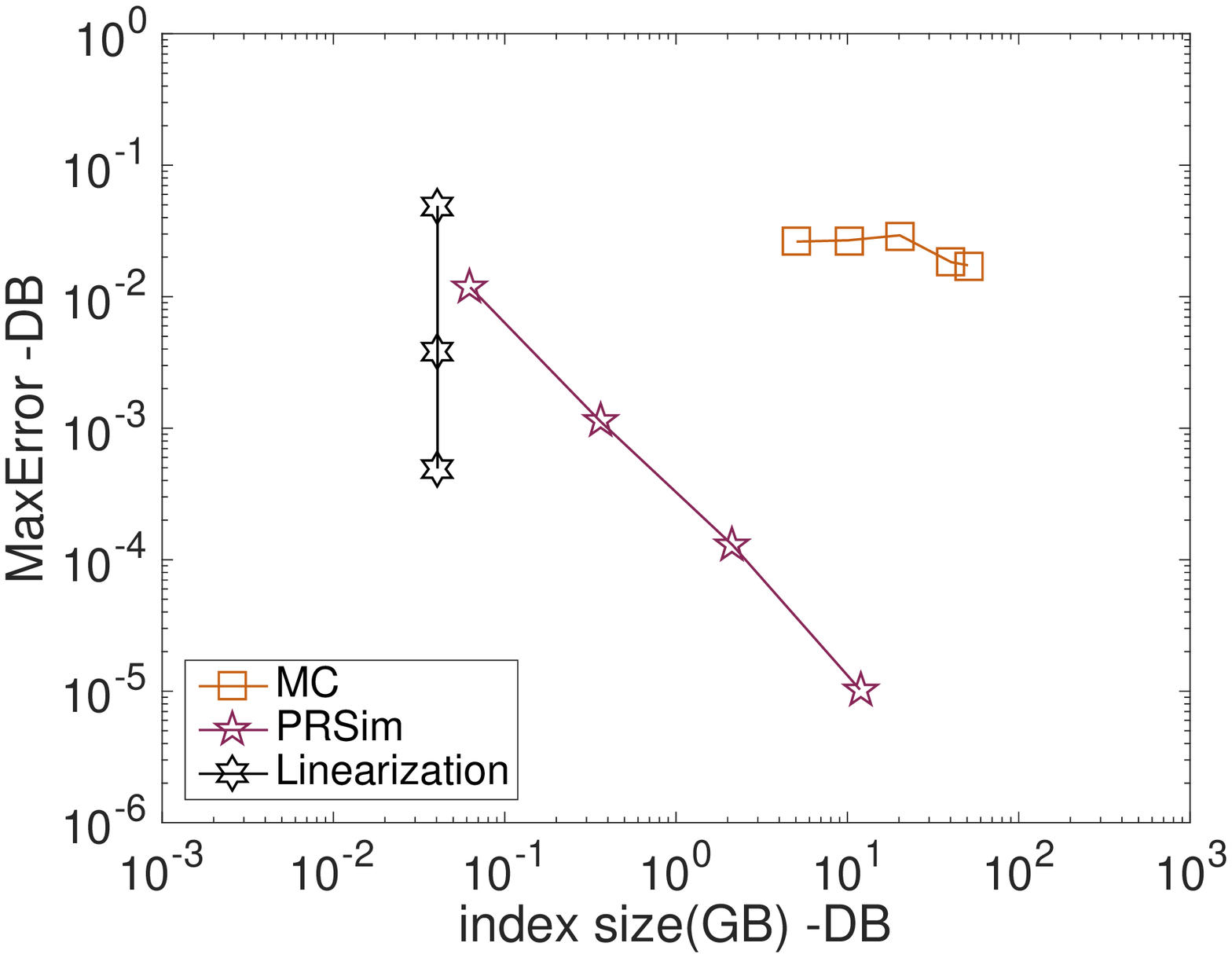} &
			%\hspace{-3mm} \includegraphics[height=25mm]{./Figs/maxerr-indexsize-largegraph-LJ.eps} &
			\hspace{-3mm} \includegraphics[height=35mm]{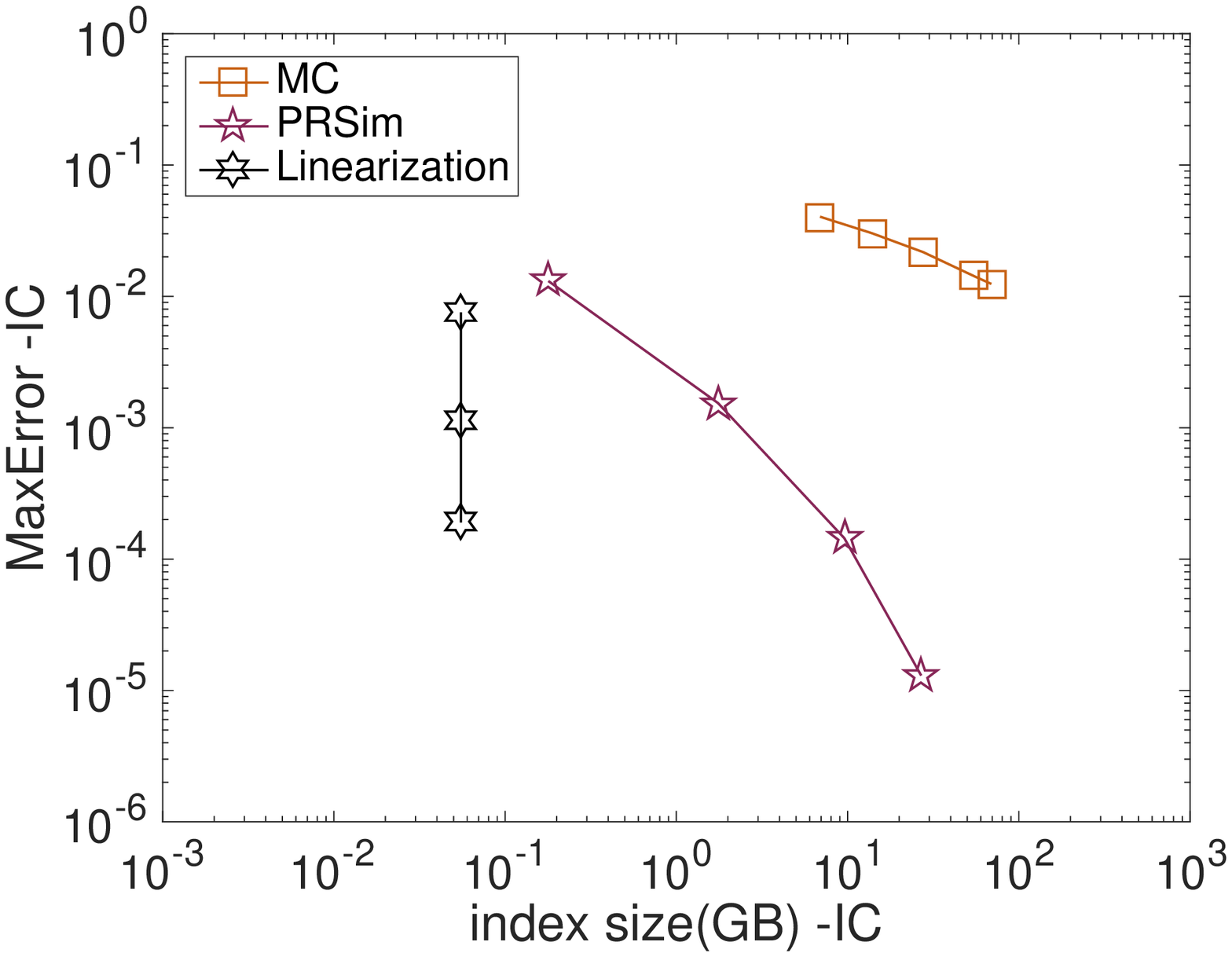} &
			%\hspace{-3mm} \includegraphics[height=25mm]{./Figs/maxerr-indexsize-largegraph-OL.eps} &
			\hspace{-3mm} \includegraphics[height=35mm]{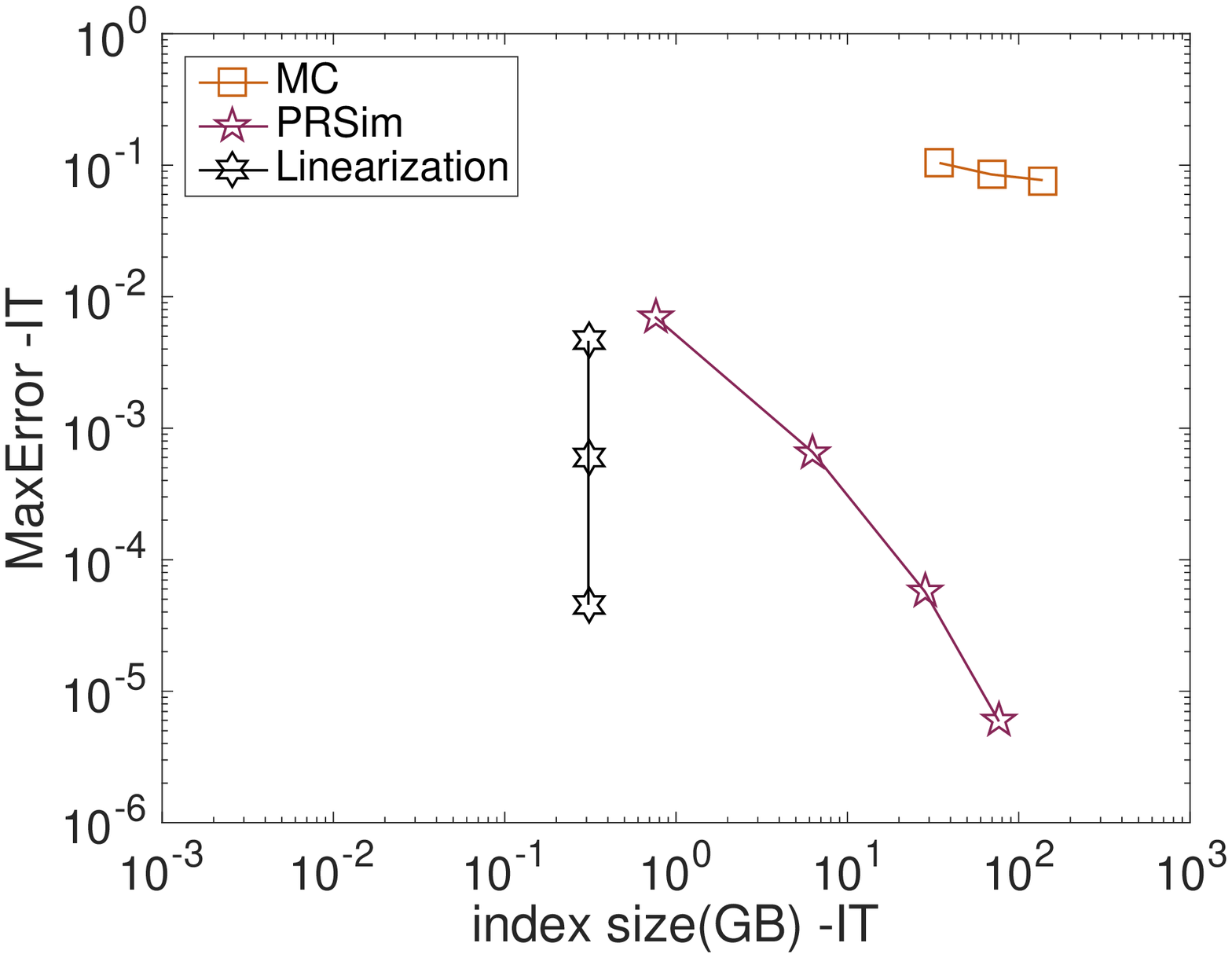} &
			\hspace{-3mm} \includegraphics[height=35mm]{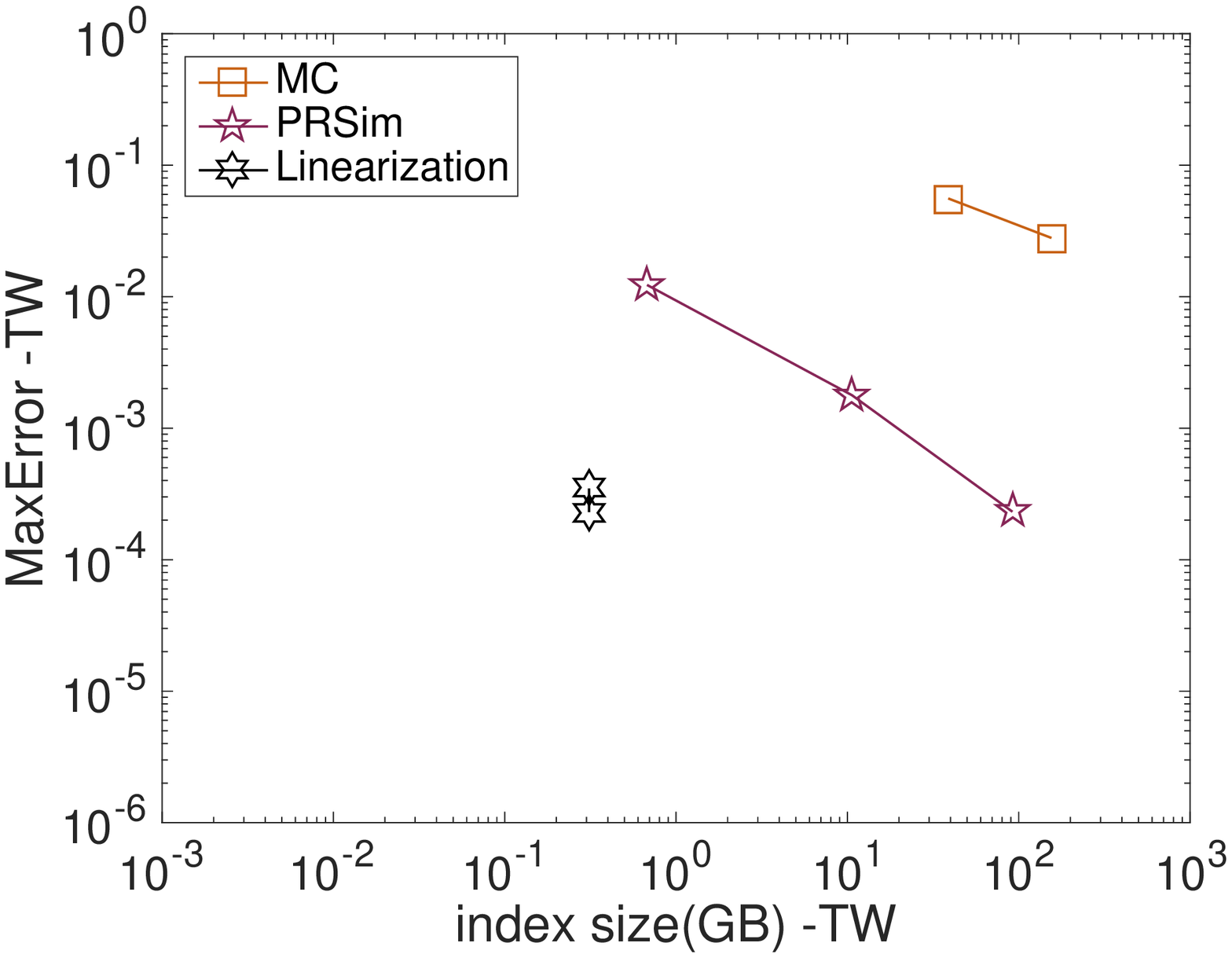}
		\end{tabular}
		\vspace{-3mm}
		\caption{ {MaxError} v.s. Index size on large graphs} 
		\label{fig:maxerr-indexsize-largegraph}
		%\vspace{-2mm}
	\end{small}
\end{figure*}

\header{\bf Methods and Parameters.} 
\begin{comment}
We evaluate \exsim 
and other nine single-source algorithms, including three Monte Carlo methods: MC \cite{FR05}, READS
\cite{jiang2017reads} and TSF \cite{SLX15}, two iterative methods:
Linearization \cite{MKK14} and ParSim \cite{yu2015efficient}, and four Local
push/sampling methods: ProbeSim \cite{liu2017probesim}, PRSim
\cite{wei2019prsim}, SLING \cite{TX16} and TopSim
\cite{LeeLY12}. Among them, \exsim, ProbeSim and  ParSim are
index-free methods, and the others are index-based methods; 
\exsim, READS, TSF, ProbeSim, TopSim and ParSim
can handle {\em dynamic} graphs, and the other methods can only handle
{\em static} graphs. In this paper, we only focus on static
graphs. For a fair comparison, we run each
algorithm in the single thread mode.
\end{comment}
We evaluate \exsim 
and other four single-source algorithms, including MC \cite{FR05}, Linearization \cite{MKK14}, ParSim \cite{yu2015efficient} and PRSim \cite{wei2019prsim}, 
Among them, \exsim, ParSim are index-free methods, and the others are index-based methods; 
\exsim and ParSim can handle {\em dynamic} graphs, and the other methods can only handle {\em static} graphs. 
In this paper, we only focus on static graphs.
For a fair comparison, we run each algorithm in the single thread mode.

MC has two parameters: the length of each random walk $L$ 
and the number of random walks per node $r$. We vary $(L,r)$
from $(5,50)$ to $(5000,50000)$ on small graphs and from $(5,50)$ to $(50,500)$ on
large graphs.
% $\{(5,50),(50,500),(500,5000), (5000,50000)\}$ on small graphs and
% $\{(5,50),(10,100),$ $(20,200),(40,400),(50,500)\}$ on large graphs.  
% In READS, there are there methods proposed in the paper, that is
% READS, READS-D and READS-Rq. We choose the fastest static method READS
% to do our experiments. It
\begin{comment}
READS shares the same parameter set $(L, r)$. To cope with its
better optimization,  we vary $(L,r)$
in larger ranges, from $(10^2,10^3)$ to $(10^6,10^7)$ on small graphs
and from $(10,100) $ to $(500,5000)$ on large graphs. 
% $\{(10^2,10^3),(10^3,10^4),(10^ 4,10^5),$ $(10^5,10^6),(10^6,10^7)\}$
% on small graphs and $\{(10,100),(50,500),$ $(100,1000),(500,5000),$
% $(10^3,10^4)\}$ on large graphs.
TSF has three parameters $R_g, R_q$ and $T$, where $R_g$ is the number
of one-way graphs, $R_q$ is the number of samples at query time and
$T$ is the number of iterations/steps. We vary $(R_g,R_q,T)$ from
$(100,20,10)$ to $(10000, 2000,1000)$ on small graphs and from
$(100,20,10)$ to $(4000,$ $800,400)$ on large graphs.
% $\{(100,20,10),(400,80,40),$ $(1000,200,100),(1600,3200,160),$
% $(6400,1280,640),$ $(10000, 2000,1000)\}$ on small graphs and
% $(R_g,R_q,T)$ in $\{(100,20,10),$ $(200,40,20),(400,80,40),(800,160,$
% $80),(2000,400,200),(4000,$ $800,400)\}$ on large graphs.  
 % Linearization only has one error parameter $\epsilon$, and we vary it
 % in from $10^{-2}$ to $10^{-6}$ on small graphs and from  $10^{-1}$ to
 % $10^{-5}$ on large graphs.
\end{comment}
ParSim has one parameter $L$, the number of
 iterations. We vary it  from $50$ to $5\times 10^5$
 % in $\{50,500,5000,50000,500000 \}$ 
 on small graphs and from $10$ to $500$ % $\{5,10,20,50,200,500 \}$
 on large graphs.
\begin{comment}
TopSim has four parameters $T,h,\eta$, and $H$, which correspond to the maximum length of a
random walk, the lower bound of the degree to identify a high
degree node, the probability threshold to eliminate a
path, and the size of priority pool, respectively. As advised in paper
\cite{LeeLY12}, we fix $1/h=100$ and $\eta=0.001$ and vary $(T,H)$
from $(3,100)$ to $(20, 10^9)$ on small graphs and from  $(3,100)$ to
$(7,10^6)$ on large graphs. 
% $\{(3,100),(5,10^4),(7,10^6),(12,10^6 )$ $,($ $16,10^8),(20,10^9) \}$
% on small graphs and $\{(3,100),(4,1000),(5,$
% $10^4),(6,10^5),(7,10^6)\}$ on large graphs.  
\end{comment}
%Finally, Linearization, ProbeSim, PRSim, SLING, and \exsim share the same error parameter
Linearization, PRSim and \exsim share the same error parameter $\e$, 
and we vary $\e$ from $10^{-1}$ to $10^{-7}$ (if possible) on both small and large graphs. 
We evaluate the optimized \exsim unless otherwise stated. 
In all experiments, we set the decay factor $c$ of SimRank as 0.6.

 % For ProbeSim, we vary the error parameter $\epsilon$ in $\{0.1,0.01,0.001,\\0.0001,0.00001 \}$ on both small and large graphs. 
 % PRSim also have one error parameter $\e$ and we very it in

 % $\{0.1,0.01,0.001,$ $10^{-4},10^{-5}\}$ on small graphs and $\{0.05,0.01,0.005,0.001,0.0005$ $,0.0001\}$ on large graphs. 
 % SLING has two parameters $\epsilon_d$ and $\theta$, respectively
 % represent the error parameter for estimation of $d_k$ and
 % $H(v_i)$. We vary them in $\{(0.01,0.00145),(5.0e-3,
 % 7.25e-4),(10^{-3},1.45\times10^{-4}),(5\times10^{-4},7.25\times10^{-5}),(10^{-4},$
 % $1.45\times10^{-5})\}$.

\header{\bf Metrics. }  We use ${MaxError}$ and {\em Precision@k} to
evaluate the quality of the single-source and top-$k$ results. Given a
source node $v_i$ and an approximate single-source result with $n$ similarities
$\S(i,j), j=1,\ldots, n$, ${MaxError}$ is defined to be the maximum
error over $n$ similarities: $MaxError =
\max_{j=1}^n\left|\S(i,j)-S(i,j)\right|$. Given a source node $v_i$
and an approximate top-$k$ result $V_k = \{v_1, \ldots, v_k\}$,  {\em Precision@k} is defined to be the
percentage of nodes in $V_k$ that coincides with the actual top-$k$
results. In our experiment, we set $k$ to be $500$. Note that this is
the first time that top-$k$ queries with $k>100$ are evaluated on large
graphs. On each dataset, we issue 50 queries and
report the average ${MaxError}$ and {\em Precision@500}.

\subsection{Experiments on small graphs}
We first evaluate \exsim against  other single-source algorithms on four small graphs. % For each graph, we select 50 nodes uniformly at  random as the query nodes and return all results averaged in the 50 nodes. 
We compute the ground truths of the single-source and top-$k$ queries
using the Power Method~\cite{JW02}. % The single-source SimRank ground truths of the 50 nodes on each graphs are computed by the $Power$ $Method$.
% For each algorithm, we generate the estimated single-source SimRank scores and obtain the top-500 nodes to calculate ${MaxError}$ and $Precision@500$. 
We omit a method if its query or preprocessing time exceeds $24$
hours.

Figure~\ref{fig:maxerr-query-smallgraph} shows the tradeoffs between
${MaxError}$ and the query time of each algorithm. 
The first
observation is that \exsim is the only algorithm that consistently
achieves an error of $10^{-7}$ within $10^4$ seconds. Linearization is
able to achieve a faster query time when the error parameter $\e$ is
large. However, as we set $\e \le 10^{-5}$, Linearization is troubled
by its $O\left( {n\log n \over \e^2}
\right)$ preprocessing time and is unable to finish the computation of the
diagonal matrix $D$ in $24$ hours. 
%TopSim comes close to \exsim on GQ and HT, but it performs not as well on WV and HP. 
%This is because the performance of TopSim depends heavily on the graph structures. 
%For the other methods, SLING and READS generally outperforms other index-free algorithms in terms of query-time/error tradeoffs, while ProbeSim
%provides the most stable performance among the index-free algorithms.

Figure~\ref{fig:precision-query-smallgraph} presents the tradeoffs
between {\em Precision@500} and query time of each algorithm. We
observe that \exsim with $\e = 10^{-7}$ is able to achieve a precision of $1$ on all four
graphs. This confirms the exactness of \exsim. We also note that
ParSim is able to achieve high precisions on all four graphs despite
its large ${MaxError}$ in Figure~\ref{fig:maxerr-query-smallgraph}. This observation demonstrates the effectiveness
of the $D \sim (1-c) I$ approximation on small datasets. 
%Finally, for the index-based methods MC, PRSim, READS, TSF, SLING and Linearization,
Finally, for the index-based methods MC, PRSim and Linearization,
we also plot the tradeoffs between ${MaxError}$ and preprocessing
time/index size in Figure~\ref{fig:maxerr-pretime-smallgraph}
and~\ref{fig:maxerr-indexsize-smallgraph}. The index sizes of
Linearization form a vertical line, as the algorithm only recomputes
and stores a diagonal matrix $D$. PRSim generally achieves the
smallest error given a fixed amount of preprocessing time and index
size.

\subsection{Experiments on large graphs}
For now, we have experimental evidence that ExactSim is able to obtain the exact single-source and top-k SimRank results on small graphs. 
On the other hand, the theoretical analysis in section 3 guarantees that ExactSim with $\e=10^{-7}$ can achieve a precision of 7 decimal places with high probability. 
Hence, we will treat the results computed by ExactSim with $\e=10^{-7}$ as the ground truths to evaluate the performance of various algorithms (including \exsim with larger $\e$) on large graphs. 
%For now, we have both theoretical and experimental evidence that  \exsim is
%able to obtain the exact single-source and top-$k$ SimRank results. 
%And the theoretical analysis given in section~\ref{sec:exsim} guarantees that the results computed by \exsim with  $\e=10^{-7}$ can reach 7 decimal places' precision. 
%Hence we can treat the \exsim's results with  $\e=10^{-7}$ as ground truth in the float type. 
%Based on the ground truth results obtained by \exsim with  $\e=10^{-7}$, 
%In this section, we will treat the results computed by \exsim with  $\e=10^{-7}$ as the ground truths 
%to evaluate the performance of various algorithms (including \exsim with larger $\e$) on large graphs.  
%we evaluate the performance of various algorithms (including \exsim with larger $\e$) on large graphs.  
%To the best of our knowledge, this is the first experimental study on the accuracy/cost tradeoffs on SimRank algorithms on large graphs.  
We also omit a method if its query or preprocessing time exceeds 24 hours. 

% We select four large graph datasets and select 50 nodes  uniformly at
% random on each graph as the single-source query
% nodes.
Figure~\ref{fig:maxerr-query-largegraph} and Figure~\ref{fig:precision-query-largegraph}  show the trade-offs
between the query time and ${MaxError}$ and  {\em Precision@500} of each
algorithm. 
Figure~\ref{fig:maxerr-pretime-largegraph} and
Figure~\ref{fig:maxerr-indexsize-largegraph} display the
${MaxError}$ and preprocessing time/index size plots of the index-based
algorithms.  
  For \exsim with $\e=10^{-7}$, we set its ${MaxError}$ as
$10^{-7}$ and {\em Precision@500}  as $1$. We observe from
Figure~\ref{fig:precision-query-largegraph}  that \exsim with $\e=10^{6}$ also
achieves a precision of $1$ on all four graphs. This suggests that the
top-$500$ result of  \exsim with $\e=10^{-6}$ is the same as
that of \exsim with $\e=10^{-7}$. In other words, the top-$500$
result of \exsim actually {\em converges} after $\e=10^{-6}$. This is
another strong evidence of the exact nature of \exsim.  From
Figure~\ref{fig:maxerr-query-largegraph}, We also
observe that $\exsim$ is the only algorithm that achieves an error of
less than $10^{-6}$ on all four large graphs. In particular, on the  TW
dataset, no existing algorithm can  achieve an error of less than
$10^{-4}$, while \exsim is able to achieve exactness within $10^4$
seconds. 

\subsection{Ablation study.} 
We now evaluate the effectiveness of the
optimization techniques. Recall that we use sampling according to
$\vec{\pi}_i(k)^2$ and local deterministic exploitation to 
reduce the query time, and sparse Linearization to reduce the space
overhead.  Figure~\ref{fig:exsim-comparision} shows the time/error
tradeoffs of the basic \exsim and the optimized
\exsim algorithms. Under similar actual error, we observe a speedup of
$10-100$ times. Table~\ref{tbl:memory-cost} shows the memory overhead of the basic \exsim and the optimized
\exsim algorithms. We observe that the space overhead of the basic
\exsim algorithm is usually larger than the graph size, while sparse
Linearization reduces the memory usage by a factor of $5-6$ times. This
demonstrates the effectiveness of our optimizing techniques.

% We also implemented \exsim on a Tesla V100-16GB GPU  to
% demonstrate its parallelization ability. Figure~\ref{tbl:memory-cost}
% shows that the GPU implementation is able to achieve an 
% speed-up of 100 times comparing to the single-thread CPU implementation.

% $\vec{\pi}_i(k)^2$ and local exploitation to save calculating time
% without precision loss. On the other side, sparse linearization is
% adopted to reduce peak memory cost in the running process. From the
% two aspects, we separately run \exsim using basic algorithm and
% optimized version both on small and large graph datasets. We select
% one small graph and one large graph to record their query time and
% ${MaxError}$ plots to compare the difference, which are shown in
% figure~\ref{fig:exsim-comparision}.  

% We observe that optimized \exsim can reach lower ${MaxError}$ in the same time compared with the basic method. And the difference between basic and optimized version is gradually clear with error parameter's reduction. When $\epsilon=10^{-6}$ on DBLP-Author (DB), there is one order of magnitude difference of ${MaxError}$ with extra little time cost, which reflects our optimization methods are effective.
 
%  Besides, we record the extra memory cost of the two version of \exsim, which is shown in table~\ref{tbl:memory-cost}. Extra money is the memory cost except for graph structures. The data shows in table~\ref{tbl:memory-cost} shows that optimized \exsim use smaller memory compared with the basic one, which is the ruction of the sparse linearization method.

\begin{figure}[h]
	\begin{small}
		\centering
		%\vspace{-1mm}
		%    \begin{footnotesize}
		\begin{tabular}{cc}
				%\multicolumn{4}{c}{\hspace{-4mm} \includegraphics[height=5mm]{./Figs/legend_large.eps}} \vspace{-1mm} \\
				\hspace{-5mm} \includegraphics[height=33mm]{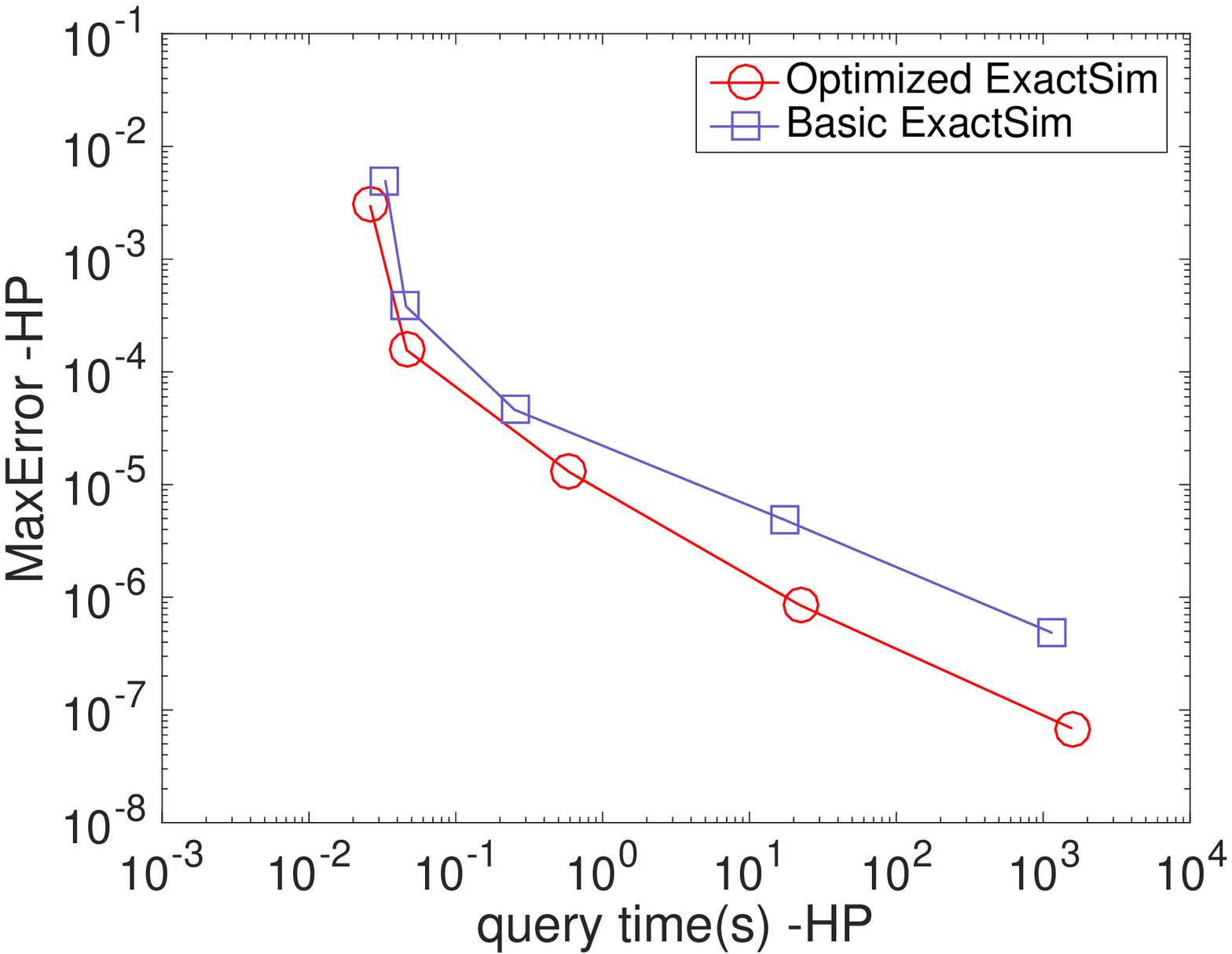} &
				\hspace{-1mm} \includegraphics[height=33mm]{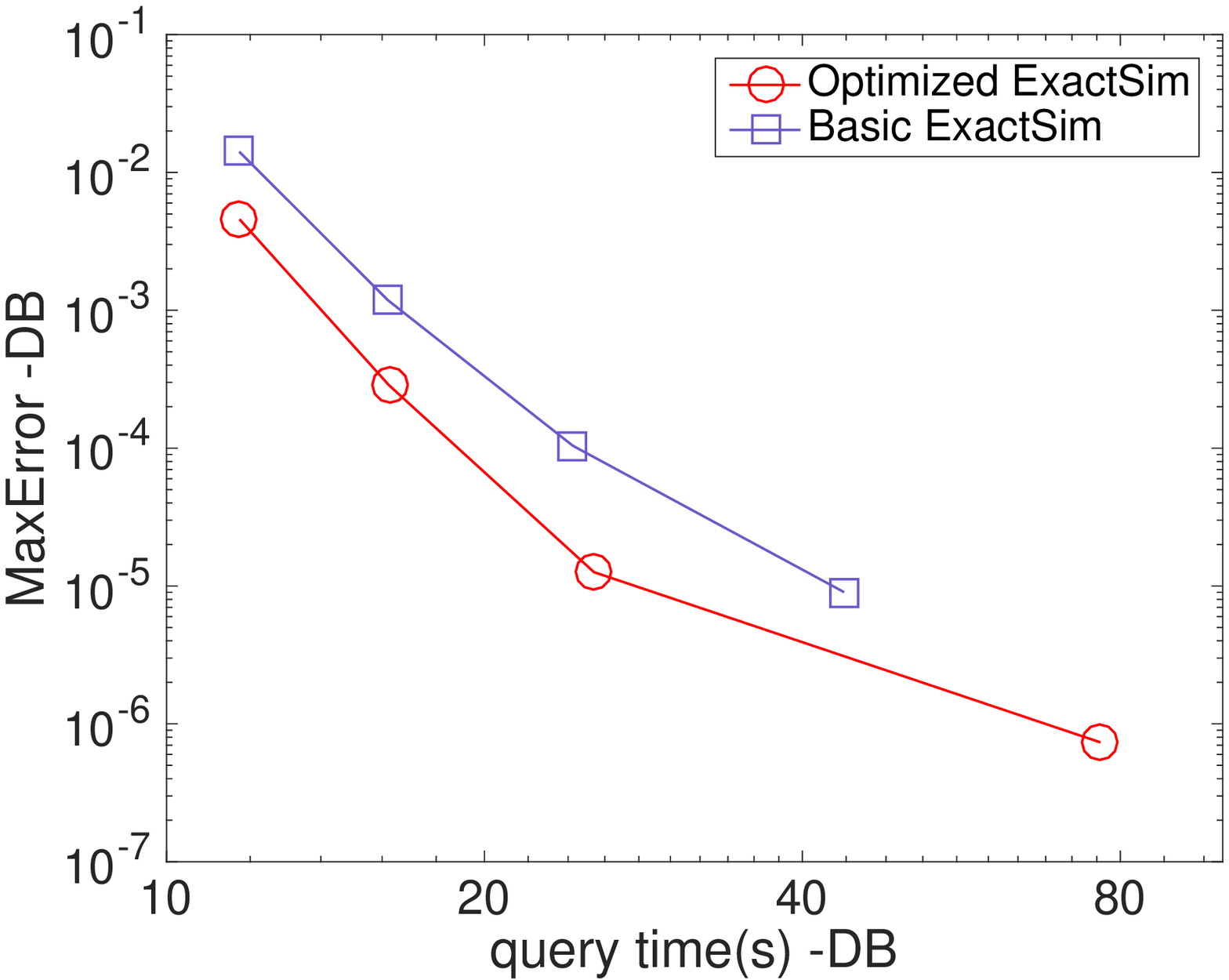} 
			\end{tabular}
			%\vspace{-1mm}
			\caption{Basic ExactSim v.s. Optimized ExactSim} 
			\label{fig:exsim-comparision}
			%\vspace{-3mm}
		\end{small}
	\end{figure}

\begin{table} [t]
	\centering
	\renewcommand{\arraystretch}{1.5}
	%\vspace{-2mm}
	\begin{small}
		%\tblcapup
		\caption{Memory overhead on large graphs.} %\label{tbl:large_query}
		%\tblcapdown
		%p{2.3in}
		\vspace{-2mm}
		\begin{tabular} {|c|c|c|c|c|} \hline
			%\begin{comment}
			Memory overhead (GB)
			&\hspace{-2mm} {\em DB}\hspace{-1mm} 
			%&\hspace{-1mm} {\em LiveJournal} \hspace{-1mm}
			&\hspace{-2mm} {\em IC}\hspace{-1mm} 
			%&\hspace{-1mm} {\em Orkut-Links}\hspace{-1mm} 
			&\hspace{-2mm} {\em IT}\hspace{-1mm} 
			&\hspace{-2mm}{\em TW} \\ \hline
			
			Basic ExactSim  &2.49  & 3.40 &   18.95 & 19.12 \\ \hline		
                        \exsim & 0.47  & 0.58 & 3.26 &  3.54 \\ \hline
                  Graph size (GB) & 0.48  & 1.88 & 10.94 &  13.30 \\ \hline	
			%\end{comment}
			
		\end{tabular}
	\label{tbl:memory-cost}
	\end{small}
	\normalsize
	%\tbldown
	%\vspace{-2mm}
\end{table}

% \begin{figure}[h]
% 	\begin{small}
% 		\centering
% 		\vspace{-4mm}
% 		%    \begin{footnotesize}
% 		\begin{tabular}{cc}
% 			%\multicolumn{4}{c}{\hspace{-4mm} \includegraphics[height=5mm]{./Figs/legend_large.eps}} \vspace{-1mm} \\
% 			\hspace{-3mm} \includegraphics[height=33mm]{./Figs/maxerr-query-simgpu-HP.eps} &
% 			\hspace{-1mm} \includegraphics[height=33mm]{./Figs/maxerr-query-simgpu-DB.eps} 
% 		\end{tabular}
% 		\vspace{-3mm}
% 		\caption{ExactSim on CPU v.s. ExactSim on GPU} 
% 		\label{fig:exsim-gpu}
% 		\vspace{-2mm}
% 	\end{small}
% \end{figure}

%%% Local Variables:
%%% mode: latex
%%% TeX-master: "paper"
%%% End:

%\input{analysis.tex}
%\input{experiment.tex}
%\vspace{-1mm}
\section{Conclusions} \label{sec:conclusions}
%\balance
This paper presents \exsim, an algorithm that produces the ground
truths for  single-source and top-$k$ SimRank
queries with precision up to 7 decimal places on large graphs. 
%Simultaneously, we propose some effective optimization techniques to sharpen this algorithm. 
We also design various optimization techniques to improve the space and time complexity of the proposed algorithm. 
We believe the ExactSim algorithm can be used to produce the ground truths for evaluating single-source SimRank algorithms on large graphs.
%Using the ground truths computed by ExactSim, 
%we present the first experimental study of the accuracy/cost tradeoffs of existing SimRank algorithms on large graphs. 
%we exploit various properties of the distributions of SimRank on large real-world graphs. 
For future work, we
note that the $O\left({\log n \over \e^2}\right)$ complexity of \exsim
prevents it from achieving a precision of $10^{-14}$ (i.e., the precision
of the double type). To achieve such extreme precision, we
need an algorithm with  $O\left({\log n \over \e}\right)$
complexity, which remains a major open problem in SimRank study.

%due to the index-free feature. 

%%% Local Variables:
%%% mode: latex
%%% TeX-master: "paper"
%%% End:

%\vspace{+1mm}
\section{ACKNOWLEDGEMENTS}
This research was supported in part by National Natural Science Foundation of China (No. 61832017, No. 61972401, No. 61932001, No. U1711261, No. 61932004 and No. 61622202), by FRFCU No. N181605012, by Beijing Outstanding Young Scientist Program NO. BJJWZYJH012019100020098, and by the Fundamental Research Funds for the Central Universities and the Research Funds of Renmin University of China under Grant 18XNLG21.

%%% Local Variables:
%%% mode: latex
%%% TeX-master: "paper"
%%% End:

\begin{small}
	\bibliographystyle{plain}
	\bibliography{paper}
\end{small}

\appendix
%\vspace{-2mm}
\section{Inequalities} \label{sec:inequalities_appendix}
\subsection{ Bernstein Inequality }
	\begin{lemma}[Bernstein inequality~\cite{ChungL06}]\label{lem:conc}
	Let $X_1, \cdots, X_R$ be independent  random variables with
	$|X_i| <  b$ for $i=1,\ldots, R$. Let
	$X=\frac{1}{R}\cdot\sum^R_{i=1}X_i$, we have
	\begin{equation}
	\Pr[|X-\E[X]|\ge \lambda] \le 2\cdot \exp\left(-\frac{\lambda^2\cdot
		R}{2R\cdot \Var[X] + 2b\lambda/3}\right),
	\end{equation}
	where $\Var[X]$ is the variance of $X$.
\end{lemma}

\section{Proofs} \label{sec:proof_appendix}
%\vspace{-2mm}

\begin{comment}
\subsection{Proof of theorem~\ref{thm:basic_time}}
\begin{proof}
To derive the running time and space overhead of the basic \exsim
algorithm, note that computing and storing each $\ell$-hop
Personalized PageRank vector $\vec{\pi}_i^{\ell}$ takes $O(m)$
time and $O(n)$ space. This results a running time of $O(mL)$ and a
space overhead of $O(nL)$. To estimate the diagonal correction matrix
$D$, the algorithm simulates $R$ pairs of $\scw$-walks, each of which
takes ${1\over \scw} = O(1)$ time. Therefore, the running time for
estimating $D$ can be bounded by $O(R)$. Finally, computing each $\vec{s}^\ell$ also takes $O(m)$
time, resulting an additional running time of $O(mL)$. Summing up all
costs, and we have the
total running time is bounded by $O(mL + R) = O\left(
{\log n \over \e^2} + m\log {1\over \e}\right)$, and the space
overhead is bounded by $O(nL) = O\left(n\log {1\over \e}\right)$. 
\end{proof}
\end{comment}

\subsection{Proof of Lemma~\ref{lem:variance}}

\begin{proof} [Proof of Lemma~\ref{lem:variance}]
	Note that $\D_r(k,k)$ is a Bernoulli random variable with
	expectation $D(k,k)$, and thus has variance $\Var[\D_r(k,k)]
	=D(k,k)(1-D(k,k)) \le D(k,k)$. Since $\D_r(k,k)$'s are   independent
	random variables, we have
	\begin{align*}
	&\Var[\vec{s}^{L}(j) ] \\
	&=  {1\over (1-\scw)^4R^2}\sum_{k=1}^n
	\sum_{r=1}^{R\rho(k)}  \left( {\sum_{\ell=0}^{L}
		\vec{\pi}_i^\ell(k)\cdot  \vec{\pi}_j^\ell(k) \over \rho(k) } \right)^2\cdot
	{\Var[\D_r(k,k)]}\\
	& =  {1\over (1-\scw)^4R}\sum_{k=1}^n
	{  \left(\sum_{\ell=0}^{L}
		\vec{\pi}_i^\ell(k)\cdot  \vec{\pi}_j^\ell(k) \right)^2 \over \rho(k)} \cdot
	{D(k,k)(1-D(k,k))}.\\
	\end{align*}
	By the Cauchy-Schwarz inequality, we have
	\begin{equation}\nonumber
	\begin{aligned}
	\left(\sum_{\ell=0}^{L}\vec{\pi}_i^\ell(k)\cdot  \vec{\pi}_j^\ell(k) \right)^2  
	&\le  \left( \sum_{\ell=0}^{L} \vec{\pi}_i^\ell(k) \right)^2 \left( \sum_{\ell=0}^{L} \vec{\pi}_j^\ell(k) \right)^2  \\
	&\le  \vec{\pi}_i(k)^2 \vec{\pi}_j(k)^2 .
	\end{aligned}
	\end{equation}
	Combining with the fact that $1-D(k,k) \le 1$, we have
	\begin{equation}
	\label{eqn:variance2}
	\Var[\vec{s}^{L}(j) ] \le {1\over (1-\scw)^4R}\sum_{k=1}^n
	{  \vec{\pi}_i(k)^2 \vec{\pi}_j(k)^2\over \rho(k)} \cdot
	{D(k,k)}.
	\end{equation}
	and the first part of the Lemma follows.

	Plugging $\rho(k) = R(k) / R =\lceil R \vec{\pi}_i(k)\rceil /R \ge \vec{\pi}_i(k)$ into
	Lemma~\ref{lem:variance}, we have
	\begin{align*}
	\Var[\vec{s}^{L}(j)  ] &\le {1\over (1-\scw)^4R}\sum_{k=1}^n
	{  \vec{\pi}_i(k)^2 \vec{\pi}_j(k)^2\over \vec{\pi}_i(k)} \cdot
	{D(k,k)} \\
	&\le {1\over (1-\scw)^4R}\sum_{k=1}^n
	{  \vec{\pi}_i(k) }.
	\end{align*}
	For the last inequality, we use the fact that
	$D(k,k) \le 1$ and $\vec{\pi}_j(k) \le 1$. Finally, since $\sum_{k=1}^n
	{  \vec{\pi}_i(k) } = 1$, we have $
	\Var[\vec{s}^{L} (j) ]  \le {1\over (1-\scw)^4R} $,
	and the second part of the Lemma follows.
\end{proof}

%\begin{comment}
\subsection{Proof of Theorem~\ref{thm:basic}}
\begin{proof}
	we first note that by
	equation~\eqref{eqn:sL}, $\vec{s}^{L}(j)  $ can be expressed as
	%\vspace{-3mm}
	\begin{equation}\nonumber
	\begin{aligned}
	%\vspace{-3mm}
	\vec{s}^{L}(j) &=\vec{e}_j \cdot \vec{s}^{L}= {1\over 1-\scw} \vec{e}_j^{\top}\cdot \sum_{\l =0}^{L}  \left(\scw P^\top \right)^\l \D \cdot \vec{\pi}_i^{\ell} \\
	&= {1\over (1-\scw)^2} \sum_{\l =0}^{L}  \left((1-\scw) \left(\scw P \right)^\l \cdot \vec{e}_j\right)^\top \cdot  \D \cdot \vec{\pi}_i^{\ell}.
	%\vec{s}^{L}(j) =\vec{e}_j^{\top} \cdot \vec{s}^{L} 
	%= {1\over (1-\scw)^2} \sum_{\l =0}^{L}  \left((1-\scw) \left(\scw P\right)^\l \cdot \vec{e}_j\right)^\top \cdot  \D \cdot \vec{\pi}_i^{\ell}.
	\end{aligned}
	%\vspace{-3mm}
	\end{equation}
	Since $(1-\scw) \left(\scw P \right)^\l \cdot  \vec{e}_j  = \vec{\pi}_j^\ell$, 
	we have %$\vec{s}^{L}(j)  = {1\over (1-\scw)^2} \sum_{\l = 0}^{L} \left(\vec{\pi}_j^\ell\right)^\top \cdot  \D \cdot \vec{\pi}_i^{\ell} .$
	%\vspace{-2mm}
	\begin{align} 
	\vec{s}^{L}(j)  = {1\over (1-\scw)^2} \sum_{\l = 0}^{L} \left(\vec{\pi}_j^\ell\right)^\top \cdot  \D \cdot \vec{\pi}_i^{\ell} .
	%\vspace{-3mm}
	\end{align}
	Summing up over the diagonal
	elements of $D$ follows that %$\vec{s}^{L}(j)  = {1\over (1-\scw)^2}\sum_{\ell=0}^{L}\sum_{k=1}^n \vec{\pi}_i^\ell(k)\cdot  \vec{\pi}_j^\ell(k) \cdot \D(k,k).$
	\begin{align}  
	\vec{s}^{L}(j)  = {1\over (1-\scw)^2}\sum_{\ell=0}^{L}\sum_{k=1}^n \vec{\pi}_i^\ell(k)\cdot  \vec{\pi}_j^\ell(k) \cdot \D(k,k). \label{eqn:formula-ppr}
	\end{align}
	Comparing the equation~\eqref{eqn:formula-ppr} with the actual SimRank value $S(i,j)$ given in~\cite{wei2019prsim} that 
	\begin{equation}
	S(i,j) ={1\over (1-\scw)^2}\sum_{\ell=0}^{\infty}\sum_{k=1}^n
	\vec{\pi}_i^\ell(k)\cdot  \vec{\pi}_j^\ell(k) \cdot D(k,k).
	\end{equation}
	we observe that there are two discrepancies:
	%between $\vec{s}^{L}(j)$ and the actual SimRank value $S(i,j)$ given in~\cite{wei2019prsim} ~\eqref{eqn:prsim}: 
	1) The iteration number changes from $\infty$ to $L$, 
	%and 2) we use the estimator $\D$ to replace actual diagonal correction matrix  $D$. 
	and 2) Estimator $\D$ replaces actual diagonal correction matrix  $D$. 
	For the	first approximation, we can bound the error by $c^L \le \e/2$ if \exsim
	sets $L = \left \lceil
	\log_{1\over c} {2 \over \e}  \right\rceil$.  Consequently, we only need to
	bound the error from replacing $D$ with $\D$ utilizing Bernstein inequality given in Lemma~\ref{lem:conc}. 
	%In particular, we will make use of the following Bernstein inequality. 
	
   %To make use of Bernstein inequality, %Lemma~\ref{lem:conc}, 
   According to Bernstein inequality, 
   we need to express $\vec{s}^{L}(j)$ as the average of independent random variables. In particular,
   let
   $\D_r(k,k)$, $r=1, \ldots, R(k)$ denote
   the $r$-th estimator of $D(k,k)$ by Algorithm~\ref{alg:Dk}. We observe
   that each $\D_r(k,k)$ is a Bernoulli random variable, that is, $\D_r(k,k) = 1$ with
   probability $D(k,k)$ and $\D_r(k,k) = 0$  with probability
   $1-D(k,k)$. We have
   \begin{equation}
   \begin{aligned}
   \vec{s}^{L}(j)  &= {1\over (1-\scw)^2}\sum_{\ell=0}^{L}\sum_{k=1}^n
   \vec{\pi}_i^\ell(k)\cdot  \vec{\pi}_j^\ell(k) \cdot
   {\sum_{r=1}^{R(k)}\D_r(k,k) \over R(k)}.\\
   & =  {1\over (1-\scw)^2}\sum_{k=1}^n  \sum_{r=1}^{R(k)} {\sum_{\ell=0}^{L}
   	\vec{\pi}_i^\ell(k)\cdot  \vec{\pi}_j^\ell(k) \over R(k) }\cdot
   {\D_r(k,k) }.
   \end{aligned}
   \end{equation}
   Let $\rho(k) = R(k)/R$
   be the fraction of pairs of $\scw$-walks assigned to $v_k$, it follows
   that 
   \begin{equation}
   \label{eqn:basic-formula}
   \vec{s}^{L}(j) = {1\over R}\cdot {1\over (1-\scw)^2}\sum_{k=1}^n  \sum_{r=1}^{R\rho(k)} {\sum_{\ell=0}^{L}
   	\vec{\pi}_i^\ell(k)\cdot  \vec{\pi}_j^\ell(k) \over   \rho(k) }\cdot
   {\D_r(k,k) }.
   \end{equation}
   We will treat each $ {\sum_{\ell=0}^{L}
   	\vec{\pi}_i^\ell(k)\cdot  \vec{\pi}_j^\ell(k) \over   \rho(k) }\cdot
   {\D_r(k,k) }$ as an independent random variable. 
   The number of such random variables is $\sum_{k=1}^n R\rho(k) = R$, so
   we have expressed $
   \vec{s}^{L}(j)$ as the average of  $R$ independent random variables. 
   Lemma~\ref{lem:variance} offers the variance bound of $\vec{s}^{L}(j)$. 
   %To utilize Bernstein inequality,%Lemma~\ref{lem:conc},  
   %we first bound the variance of $\vec{s}^{L}(j)$. 
	%Note that we only need inequality~\eqref{eqn:variance-basic} to derive
	%the error bound for the basic \exsim algorithm. The more complex
	%inequality~\eqref{eqn:variance-all} will be used to design various
	%optimization techniques. 
	%We are now ready to prove Theorem~\ref{thm:basic}. 
	To utilize Bernstein inequality, %Lemma~\ref{lem:conc}, 
	we also need to bound $b$,  the maximum value of the
	random variables ${\sum_{\ell=0}^{L}
		\vec{\pi}_i^\ell(k)\cdot  \vec{\pi}_j^\ell(k) \over   \rho(k) } \cdot
	\D_r(k,k)$. We have
	$${\sum_{\ell=0}^{L}
		\vec{\pi}_i^\ell(k)\cdot  \vec{\pi}_j^\ell(k) \over  \vec{\pi}_i(k) } \cdot
	\D_r(k,k)  \le {\sum_{\ell=0}^{L}
		\vec{\pi}_i^\ell(k) \over  \vec{\pi}_i(k) } \le   {\vec{\pi}_i(k)
		\over  \vec{\pi}_i(k) } = 1. $$
	%By applying the Lemma~\ref{lem:conc}  with $b=1$ and
	Applying Bernstein inequality with $b=1$ and
	$\Var[\vec{s}^{L}(j) ] \le {1\over (1-\scw)^4R}$, where $R =
	{6\log n \over (1-\scw)^4  \e^2}$, we have
	$\Pr[|\vec{s}^{L}(j)  - E[\vec{s}^{L}(j)] |>  \e/2] < 1/n^3.$
	Combining with the $\e/2$ error introduced by the truncation $L$, we
	have $\Pr[|\vec{s}^{L}(j)  - S(i,j) |>  \e] < 1/n^3.$
	By union bound over all possible target nodes $j=1,\ldots, n$ and all
	possible source nodes  $i = 1, \ldots, n$, we ensure
	that for all $n$ possible source node and $n$ target nodes, 
	$$\Pr[\forall i,j,\,  |\vec{s}^{L}(j)  - S(i, j)|>  \e] < 1/n,$$ 
	and the Theorem follows.
	
\end{proof}
%\end{comment}

\begin{comment}
\subsection{Proof of Theorem~\ref{thm:basic}}
%\vspace{-1mm}
\begin{proof}
	%We are now ready to prove Theorem~\ref{thm:basic}. 
	To utilize bernstein inequality, we need to bound $b$,  the maximum value of the
	random variables ${\sum_{\ell=0}^{L}
		\vec{\pi}_i^\ell(k)\cdot  \vec{\pi}_j^\ell(k) \over   \rho(k) } \cdot
	\D_r(k,k)$. We have
	$${\sum_{\ell=0}^{L}
		\vec{\pi}_i^\ell(k)\cdot  \vec{\pi}_j^\ell(k) \over  \vec{\pi}_i(k) } \cdot
	\D_r(k,k)  \le {\sum_{\ell=0}^{L}
		\vec{\pi}_i^\ell(k) \over  \vec{\pi}_i(k) } \le   {\vec{\pi}_i(k)
		\over  \vec{\pi}_i(k) } = 1. $$
	Applying bernstein inequality with $b=1$ and
	$\Var[\vec{s}^{L}(j) ] \le {1\over (1-\scw)^4R}$, where $R =
	{6\log n \over (1-\scw)^4  \e^2}$, we have
	$\Pr[|\vec{s}^{L}(j)  - E[\vec{s}^{L}(j)] |>  \e/2] < 1/n^3.$
	Combining with the $\e/2$ error introduced by the truncation $L$, we
	have $\Pr[|\vec{s}^{L}(j)  - S(i,j) |>  \e] < 1/n^3.$
	By union bound over all possible target nodes $j=1,\ldots, n$ and all
	possible source nodes  $i = 1, \ldots, n$, we ensure
	that for all $n$ possible source node and $n$ target nodes, 
	$$\Pr[\forall i,j,\,  |\vec{s}^{L}(j)  - S(i, j)|>  \e] < 1/n,$$ 
	and the Theorem follows.
\end{proof}

\end{comment}

\subsection{Proof of Lemma~\ref{lem:sparse}}
\begin{proof}
	We note that the sparse Linearization introduces an extra error of $(1-\scw)^2\e$ to
	each $\vec{\pi}_i^{\ell}(k)$, $k=1,\ldots,n$, $\ell = 0, \ldots,
	\infty$. According to equation~\eqref{eqn:formula-ppr}, the 
	estimator $\vec{s}^L(j)$ can be expressed as
	\begin{equation}
	\label{eqn:sparse1}
	\vec{s}^L(j)={1\over (1-\scw)^2}\sum_{\ell=0}^{L}\sum_{k=1}^n
	\left(\vec{\pi}_i^\ell(k) \pm (1-\scw)^2\e \right)  \cdot \vec{\pi}_j^\ell(k)\cdot \D(k,k).
	\end{equation}
	Thus, the  error introduced by sparse Linearization can be bounded by 
	\begin{equation}
	\label{eqn:sparse2}
	{1\over (1-\scw)^2}\sum_{\ell=0}^{\infty}\sum_{k=1}^n (1-\scw)^2\e  \cdot
	\vec{\pi}_j^\ell(k) \cdot \D(k,k).
	\end{equation}
	Using the facts that $\sum_{\ell=0}^{\infty}\sum_{k=1}^n
	\vec{\pi}_j^\ell(k)= 1$ and $\D(k,k) \le 1$, the above error can be
	bounded by ${1\over (1-\scw)^2}  \cdot (1-\scw)^2\e = {\e} $.
\end{proof}

\subsection{Proof of Lemma~\ref{lem:variance1}}
\begin{proof} 
	% By setting $R(k) =
	%   R \left \lceil {\vec{\pi}_i(k)^2 \over  \|\vec{\pi}_i\|^2 }
	%   \right\rceil$, we can reduce the variance of the estimation by a
	%   factor of $\|\vec{\pi}_i\|^2 $.  In particular, 
	Recall that $\rho(k)
	$ is the fraction of sample  assigned to
	$D(k,k)$. We have  $\rho(k) =  \left\lceil {R \vec{\pi}_i(k)^2 \over
		\|\vec{\pi}_i\|^2} \right\rceil  /R \ge {\vec{\pi}_i(k)^2 \over
		\|\vec{\pi}_i\|^2}$. By the inequality~\eqref{eqn:variance-all} in Lemma~\ref{lem:variance}, we can bound the
	variance of estimator $\vec{s}^L(j)$ as
	\begin{align*}
	&\quad \Var[\vec{s}^L(j) ] \le {1\over (1-\scw)^4R}\sum_{k=1}^n
	{  \vec{\pi}_i(k)^2 \vec{\pi}_j(k)^2\over \rho(k)} \cdot
	{D(k,k)}\\
	& \le {1\over (1-\scw)^4R}\|\vec{\pi}_i\|^2\sum_{k=1}^n
	{  \vec{\pi}_j(k)^2 } ={1\over (1-\scw)^4R}\|\vec{\pi}_i\|^2 \|\vec{\pi}_j\|^2.
	\end{align*}
	Here,  we use the facts that $\|\vec{\pi}_j\|^2 = \sum_{k=1}^n
	\vec{\pi}_j(k)^2  $ and $D(k,k) \le 1$. Since we need to bound the
	variance for all possible nodes $v_j$ (and hence all
	possible $ \|\vec{\pi}_j\|^2$), we make the
	relaxation that $\|\vec{\pi}_j\|^2 \le
	\|\vec{\pi}_j\|_1^2 =  1$, where $	\|\vec{\pi}_j\|_1^2 =(\sum_{k=1}^{n}\vert \vec{\pi}_j(k)\vert)^2$. 
	And thus $ \Var[\vec{s}^L(j) ] 
	\le {1\over (1-\scw)^4R}\|\vec{\pi}_i\|^2.$
	%                        \begin{align*}
	%   \Var[\vec{s}^L(j) ] 
	% \le {1\over (1-\scw)^4R}\|\vec{\pi}_i\|^2.
	%                        \end{align*}
	This suggest that by sampling according to
	$\vec{\pi}_i(k)^2$, we reduce the variance of the
	estimators by a factor
	$\|\vec{\pi}_i\|^2$. Recall that the \exsim
	algorithm computes the Personalized PageRank
	vector $\vec{\pi}_i$ before estimating $D$, we
	can obtain the value of   $\|\vec{\pi}_i\|^2$
	and  scale $R$ down by a factor of
	$\|\vec{\pi}_i\|^2$. This simple modification
	will  reduce the running time to $O\left({\|\vec{\pi}_i\|^2 \log n \over
		\e^2}\right)$.
	
	One small technical issue is that the maximum of the random variables $ {\sum_{\ell=0}^{\infty}
		\vec{\pi}_i^\ell(k)\cdot  \vec{\pi}_j^\ell(k) \over   \rho(k) }\cdot
	{\D_r(k,k) }$ may gets  too large as the fraction $\rho(k)$ gets
	too small. However, by the facts that   $\rho(k) =
	\left \lceil { R \vec{\pi}_i(k)^2 \over  \|\vec{\pi}_i\|^2 }
	\right\rceil /R$  and $\D_r(k,k) \le 1$, we have
	\begin{align*}
	{\sum_{\ell=0}^{\infty}
		\vec{\pi}_i^\ell(k)\cdot  \vec{\pi}_j^\ell(k) \over   \rho(k) }\cdot
	{\D_r(k,k) } \le  {
		\vec{\pi}_i(k) \over   \rho(k) }  = R \vec{\pi}_i(k) /\left \lceil {R\vec{\pi}_i(k)^2 \over  \|\vec{\pi}_i\|^2 }
	\right\rceil.
	\end{align*}
	If we view the right side of the above equality as a function of $\vec{\pi}_i(k)$,
	it takes maximum when ${R\vec{\pi}_i(k)^2 \over  \|\vec{\pi}_i\|^2}= 1$, or equivalently $\vec{\pi}_i(k) = \sqrt{  \|\vec{\pi}_i\|^2\over R 
	} $.  Thus, the random variables in equation~\eqref{eqn:basic-formula} can be bounded by
	$ R  \sqrt{  \|\vec{\pi}_i\|^2\over R 
	} =   \|\vec{\pi}_i\| \sqrt{R}$. Plugging $b =  \|\vec{\pi}_i\|
	\sqrt{R}$ and $  \Var[\vec{s}^L(j)]  \le {1\over
		(1-\scw)^4R}\|\vec{\pi}_i\|^2$ into bernstein inequality, and the
	Lemma follows.
\end{proof}

\subsection{Proof of Lemma~\ref{lem:zl}}
\begin{proof}
	%\vspace{-2mm}
	Note that $\left(\scw \right)^\ell \left(P^\top\right)^\ell(k,q)$ is the probability that a $\scw$-walk from $v_k$
	visits $v_q$ at  its $\ell$-th step. Consequently, $c^\ell \left(P^\top\right)^{\ell}(k,q)^2$ is the probability that two
	$\scw$-walks from node $v_k$ visit node $v_q$ at their $\ell$-th
	step simultaneously. To ensure this is the first time that the two
	$\scw$-walks meet, we subtract the probability
	mass that the two $\scw$-walks have met before. In particular,
	recall that 
	$Z_{\ell'}(k, q')$ is the probability that two
	$\scw$-walks from node $v_k$ first meet at $v_{q'}$ in exactly
	$\ell'$    steps.  Due to the memoryless property of the $\scw$-walk,
	the two $\scw$-walks will behave as two new $\scw$-walks from
	$v_{q'}$ after their $\ell'$-th step. The probability that these two new $\scw$-walks
	visitis $v_q$ in exact $\ell-\ell'$ steps is $ c^{\ell-\ell'}
	\left(P^\top\right)^{\ell-\ell'}(q',q)^2 $.  Summing up $q'$ from $1$ to $n$ and
	$\ell'$ from $1$ to $\ell-1$, and the Lemma follows. 
\end{proof}

\end{document}